\documentclass[aap]{imsart}

\RequirePackage{amsthm,amsmath,amsfonts,amssymb}
\RequirePackage[numbers]{natbib}
\RequirePackage[colorlinks,citecolor=blue,urlcolor=blue]{hyperref}
\RequirePackage{graphicx}

\usepackage{amsmath}
\usepackage{amstext}
\usepackage{amssymb}
\usepackage{mathtools}
\usepackage{amsfonts}

\usepackage[english]{babel}
\usepackage{verbatim}
\usepackage{color}

\usepackage{algorithm}
\usepackage{algorithmic}

\usepackage{multicol}
\usepackage{lipsum}
\startlocaldefs
\theoremstyle{plain}
\newtheorem{theorem}{Theorem}[section]
\newtheorem{lemma}[theorem]{Lemma}
\theoremstyle{remark}

\newtheorem{definition}{Definition}

\endlocaldefs

\def\cR{\mathbb{R}}
\def\N{\mathbb{N}}
\def\rP{\mathbb{P}}
\def\rE{\mathbb{E}}

\def\reg{\mathop{\rm reg}}
\def\diam{\mathop{\rm diam}}
\def\new{\mathop{\rm new}}

\def\ddim{\mathop{\rm ddim}}
\def\Lip{\mathop{\rm Lip}}
\def\av{\mathop{\rm av}}
\def\sym{\mathop{\rm sym}}
\def\opt{\mathop{\rm opt}}

\def\T{\mathop{\rm T}}

\def\L{\mathop{\rm L}}
\def\e{\mathop{\rm e}}

\def\Pr{\mathop{\rm Pr}}
\def\Id{\mathop{\rm Id}}
\def\intr{\mathop{\rm int}}

\def\argmin{\mathop{\rm arg\, min}}

\def\hx{\hat{x}}

\def\B{{\mathcal B}}

\def\F{{\mathcal F}}

\def\P{{\mathcal P}}

\def\K{{\mathcal K}}

\def\M{{\mathcal M}}
\def\S{{\mathcal S}}

\def\balpha{{\boldsymbol \alpha}}

\def\bpi{{\boldsymbol \pi}}

\def\by{{\bf y}}
\def\ba{{\bf a}}
\def\bx{{\bf x}}
\def\bz{{\bf z}}

\def\tx{{\tilde x}}
\def\ta{{\tilde a}}

\def\hx{{\hat x}}

\def\sPr{{\mathsf{Pr}}}

\def\sX{{\mathsf X}}
\def\sS{{\mathsf S}}

\def\sA{{\mathsf A}}

\def\sZ{{\mathsf Z}}

\def\sE{{\mathsf E}}

\def\sU{{\mathsf U}}

\def\e{\mathsf{e}}
\def\g{\mathsf{g}}
\def\m{\mathsf{m}}
\def\fC{\mathsf{C}}
\def\Cp{\mathsf{Cp}}
\def\L{\mathsf{L}}
\def\K{\mathsf{K}}
\def\R{\mathsf{R}}
\def\bp{{\bar p}}
\def\bc{{\bar c}}
\def\ha{{\hat a}}
\def\hx{{\hat x}}
\def\tbx{\tilde{{\bf x}}}
\def\tba{\tilde{{\bf a}}}
\def\talpha{\tilde{\alpha}}
\def\halpha{\hat{\alpha}}
\def\balpha{\bar{\alpha}}
\def\hkappa{\hat{\kappa}}
\def\bpi{{\boldsymbol \pi}}
\def\rL{\mathbb{L}}
\def\a{\mathsf{a}}
\def\b{\mathsf{b}}
\def\c{\mathsf{c}}

\allowdisplaybreaks

\begin{document}

\begin{frontmatter}
\title{Linear Mean-Field Games with Discounted Cost}
\runtitle{Linear MFGs with Discounted Cost}

\begin{aug}
\author[A]{\fnms{Naci}~\snm{Saldi}\ead[label=e1]{naci.saldi@bilkent.edu.tr}}
\address[A]{Department of Mathematics,
Bilkent University\printead[presep={,\ }]{e1}}
\end{aug}

\begin{abstract}
In this paper, we introduce discrete-time linear mean-field games subject to an infinite-horizon discounted-cost optimality criterion. The state space of a generic agent is a compact Borel space.  At every time, each agent is randomly coupled with another agent via their dynamics and one-stage cost function, where this randomization is generated via the empirical distribution of their states (i.e., the mean-field term). Therefore, the transition probability and the one-stage cost function of each agent depend linearly on the mean-field term, which is the key distinction between classical mean-field games and linear mean-field games. Under mild assumptions, we show that the policy obtained from infinite population equilibrium is $\varepsilon(N)$-Nash when the number of agents $N$ is sufficiently large, where $\varepsilon(N)$ is an explicit function of $N$. Then, using the linear programming formulation of MDPs and the linearity of the transition probability in mean-field term, we formulate the game in the infinite population limit as a generalized Nash equilibrium problem (GNEP) and establish an algorithm for computing equilibrium with a convergence guarantee.
\end{abstract}



\end{frontmatter}
\tableofcontents

\section{Introduction}\label{sec1}

This paper introduces linear mean-field games, which are discrete-time stochastic dynamic games with one-stage cost and transition probability that are linear with respect to the empirical distribution of the states (i.e., mean-field term). Specifically, at each time step, a generic agent is randomly coupled with another agent, where the randomization is generated via a mean-field term. The standard method for analyzing these game models is to look at the problem's infinite population limit to obtain an approximate Nash equilibrium. This idea was first introduced to deal with continuous-time differential games with a large number of agents interacting via a mean-field term in the works of \cite{HuMaCa06} and \cite{LaLi07}. We refer the reader to \cite{HuCaMa07,TeZhBa14,Hua10,BeFrPh13,Ca11,CaDe13,GoSa14,MoBa16} for studies on continuous-time classical mean-field games with different models and cost functions.

In this paper, our goal is to obtain approximate Nash equilibria for linear mean-field games by considering a stationary infinite population limit. In the finite-agent setting, we first show that the \emph{stationary} or \emph{oblivious} mean-field equilibrium (see \cite{WeBeRo05,WeBeRo08}), which is the equilibrium notion in the infinite population limit, is approximately Nash. In particular, we can tell how stationary mean-field equilibrium is close to the true Nash equilibrium by looking at the number of players. Then, by formulating the stationary infinite population game as a generalized Nash equilibrium problem (GNEP), we establish an algorithm to compute stationary mean-field equilibrium.

In classical stationary mean-field games, a generic agent models the collective behavior of other agents (\cite{WeBeRo05}) as a time-homogeneous distribution, and therefore faces a  Markov decision process (MDP) with a constraint on the stationary distribution of the state. In this model, both one-stage cost and transition probability depend, in general, nonlinearly on the infinite population limit of the mean-field term, which is the key distinction between linear mean-field games and classical mean-field games. The stationary mean-field equilibrium consists of a policy and a distribution that satisfy the Nash certainty equivalence (NCE) principle (\cite{HuMaCa06}). This principle requires that the policy should be optimal under a given distribution, which is supposed to be the stationary infinite population limit of the mean-field term, and that when the generic agent applies this policy, the resulting stationary distribution of the agent's state must be the same as this distribution.  Under quite mild assumptions, the existence of stationary mean-field equilibrium can be proved via Kakutani's fixed point theorem. Moreover, it can be established that when the number of agents is large enough, the policy in stationary mean-field equilibrium is an approximate Nash equilibrium for a finite-agent setting (\cite{AdJoWe15}).

In the literature for classical mean-field games, an algorithm is established in \cite{WeBeRo10} for computing oblivious equilibrium in a stationary mean-field industry dynamics model. \cite{AdJoWe15} consider a stationary mean-field game model with a countable state-space under an infinite-horizon discounted-cost criterion.  \cite{HuMa19} study stationary mean-field games with binary action space, demonstrating the existence and uniqueness of the stationary mean-field equilibrium. \cite{WeLi22} consider stationary mean-field games with a continuum of states and actions, and establishes a novel uniqueness result for stationary mean-field equilibrium. \cite{GoMoSo10} study both stationary and non-stationary mean-field games with a finite state space over a finite horizon and establishes the existence and uniqueness of the mean-field equilibrium for both cases. The references \cite{ElLiNi13,MoBa15,NoNa13,MoBa16-cdc} consider discrete-time mean-field games with state dynamics that is linear in state, action, and mean-field term. While the state dynamics in the latter case is linear in the mean-field term, the corresponding transition probability can still be nonlinear in the mean-field term, making it distinct from the current game model being considered.

The previous studies reviewed only establish the existence and uniqueness of the mean-field equilibrium, but do not provide an algorithm with a guarantee of convergence to compute it, with the exception of the model with linear state dynamics and two other papers \cite{WeBeRo10,AnKaSa20}. This work investigates this problem for linear mean-field games and proposes an algorithm that formulates the game as a GNEP, proving the convergence of the algorithm to the stationary mean-field equilibrium.

\subsection{Contributions}

\begin{itemize}
\item[1.] In this paper, we introduce a novel mean-field game model that is called linear mean-field games. In this model, agents randomly interact with each other, and the probability of this random interaction is given by the mean-field term. As a result, both the one-stage cost and the transition probability depend linearly on the mean-field term. This is, in general, not true in classical mean-field games, even for models with linear state dynamics. 
\item[2.] In Lemma~\ref{lemma1}, it is shown that an equilibrium policy in the infinite population limit is Lipschitz continuous. Similarly, Lemma~\ref{lemma6} demonstrates that the best-response policy to the infinite population equilibrium in the finite-agent setting is also Lipschitz continuous. To obtain this result, it is proved in Lemma~\ref{lemma3} and Lemma~\ref{lemmaa-new} that the optimal value function of the agent that computes the best-response policy is Lipschitz continuous. These results do not impose any conditions on the set to which the best-response policy belongs, unlike in previous work such as \cite{SaBaRa18}, where the best-response policy is assumed to be a component of the set of Markov policies that use only local state information.
\item[3.] Using Lemma~\ref{lemma1} and Lemma~\ref{lemma6}, Theorem~\ref{theorem3} proves that the policy in the infinite population equilibrium is $\varepsilon(N)$-Nash when the number of agents $N$ is sufficiently large, where $\varepsilon(N)$ is an explicit function of $N$. In the literature, there exist various results related to both continuous and discrete time, as well as static classical mean-field games, that establish a relationship between the rate at which $\varepsilon(N)$ approaches zero as $N$ increases. Some examples of papers that explore this relationship include \cite{MoBa16,MoBa16-cdc,MoBa18,MoBa19,Bas18}. Our result is in line with those found in the literature, but with the possibility of having more relaxed conditions. This is achievable due to the linearity of the one-stage cost and transition probability with respect to the mean-field term. 
\item[4.] We use the linear programming formulation of MDPs and the linearity of the transition probability in the mean-field term to express the game in the infinite population limit as a GNEP, for which there exists a large body of literature on computing equilibrium solutions. By adapting one of these methods to our problem, we develop an algorithm for computing equilibrium in the infinite population limit with a convergence guarantee.   
\end{itemize}

\smallskip

\noindent\textbf{Notation.} 
For a metric space $\sE$, we let $C_b(\sE)$ denote the set of all bounded continuous real functions on $\sE$, $\P(\sE)$ denote the set of all Borel probability measures on $\sE$, and $\B(\sE)$ denote the collection of Borel sets. For any $\sE$-valued random element $x$, ${\cal L}(x)(\,\cdot\,) \in \P(\sE)$ denotes the distribution of $x$. A sequence $\{\mu_n\}$ of measures on $\sE$ is said to converge weakly to a measure $\mu$ if $\int_{\sE} g(e) \, \mu_n(de)\rightarrow\int_{\sE} g(e) \, \mu(de)$ for all $g \in C_b(\sE)$. The set of probability measures $\P(\sE)$ is endowed with the Borel $\sigma$-algebra induced by weak convergence. In this paper, $\|\cdot\|$ denotes the Euclidean norm and $\|\cdot\|_{\infty}$ denotes the $l_{\infty}$-norm. The notation $v\sim \nu$ means that the random element $v$ has distribution $\nu$. Unless otherwise specified, the term ``measurable" will refer to Borel measurability.

\section{Preliminary Definitions and Results}\label{sec1}

In this section, we provide definitions used throughout the paper and present some preliminary lemmas that are necessary to prove the main results. These lemmas are generally easy to prove, but we include their proofs for completeness.

For each $M\geq1$, let us define the mean-field term function $\e$ from $\sX^M$ to $\P(\sX)$ as follows, where $\sX$ is the compact state space of the mean-field game
$$
\e[\,\cdot\,|\,\bx] := \frac{1}{N} \sum_{i=1}^M \delta_{x_i}(\,\cdot\,)
$$
We also define the following set $\P_{M}(\sX):=\e[\,\cdot\,|\,\sX^M]$, which is the image of $\sX^M$ under $\e$ into $\P(\sX)$.
For any set $\{a_1,\ldots,a_R\}$, let $\S_{\{a_1:a_R\}}$ denote the permutation group of this set; that is, $\S_{\{a_1:a_R\}}$ is the set of all permutations of $\{a_1,\ldots,a_R\}$. For any two distributions $\mu,\nu$ on a metric space $\sZ$ with metric $d_{\sZ}$, let $\Cp(\mu,\nu)$ denote the set of all couplings of $\mu,\nu$ in $\P(\sZ\times\sZ)$; that is, $\xi \in \Cp(\mu,\nu)$ if  $\xi(\cdot\,\times\sX) = \mu$ and $\xi(\sX \times\cdot\,) = \nu$. In view of this, Wasserstein distance of order~1 between $\mu$ and $\nu$ is defined as follows  \cite{Vil09}
$$
W_1(\mu,\nu) := \inf\{\rE_{\xi}[d_{\sZ}(x,y)]: \xi \in \Cp(\mu,\nu)\}
$$
Define $\bc_{\max}:=\sup_{(x,a,z) \in \sX\times\sA\times\sX} \bc(x,a,z)$, where $\bc$ is the non-negative one-stage cost function of the mean-field game. Introduce the following function class
\begin{align*}
\F := \left\{g:\sX\rightarrow[0,K]; \|g\|_{\Lip}\leq 1, \, g(y_*)=0\right\}
\end{align*}
where $y_*\in \sX$ is some arbitrary fixed point and $K\geq \max\{\bc_{\max},\diam(\sX)\}$. For any $\varepsilon>0$, we define the $\varepsilon$-covering number of $\F$ as follows
\begin{align*}
\N(\varepsilon,\F) := \inf \left\{ N: \exists \{g_1,\ldots,g_N\} \subset \F \,\, \text{such that} \, \inf_{i=1,\ldots,N} \|g-g_i\|_{\infty} < \varepsilon, \, \forall g \in \F \right\}
\end{align*}
Since $\F$ is totally bounded in $C_b(\sX)$ with respect to the sup-norm by Arzela–Ascoli theorem \cite[Theorem 2.4.7]{Dud89}, $\N(\varepsilon,\F)$ is finite for all $\varepsilon>0$. For some cases, one can obtain an upper bound on $\N(\varepsilon,\F)$ in terms of $\varepsilon$, $K$, diameter of $\sX$, and dimension of $\sX$ (see \cite[Lemma 6]{GoKoRo17} and \cite[Lemma 4.2]{GoKoRo16}). Indeed, by \cite[Lemma 6]{GoKoRo17}, we have
$$
\N(\varepsilon,\F)  \leq \left(\frac{8K}{\varepsilon}\right)^{\N\left(\frac{\varepsilon}{8K},\sX\right)}
$$
where $\N\left(\frac{\varepsilon}{8K},\sX\right)$ is $\frac{\varepsilon}{8K}$-covering number of $\sX$ with respect to its metric $d_{\sX}$. If $\sX$ is doubling space with doubling dimension $\ddim(\sX)$ (see \cite[p. 107]{GoKoRo16} for the definition of doubling spaces), we also have 
$$
\N\left(\frac{\varepsilon}{8K},\sX\right) \leq \left(\frac{16K \diam(\sX)}{\varepsilon}\right)^{\ddim(\sX)}
$$
This implies that 
$$
\N(\varepsilon,\F)  \leq \left(\frac{8K}{\varepsilon}\right)^{\left(\frac{16K \diam(\sX)}{\varepsilon}\right)^{\ddim(\sX)}}
$$
If $\sX$ is a subset of $d$-dimensional Euclidean space, then doubling dimension of $\sX$ is $O(d)$ \cite[p. 107]{GoKoRo16}. Hence, in this case we have the following upper bound for $\N(\varepsilon;\F)$:
$$
\N(\varepsilon,\F)  \leq \left(\frac{8K}{\varepsilon}\right)^{\left(\frac{16K \diam(\sX)}{\varepsilon}\right)^{O(d)}}
$$

Note that Wasserstein distance of order 1 between any probability measures $\mu,\nu$ on $\sX$ can also be expressed  by duality as follows \cite{Vil09}
\begin{align*}
W_1(\mu,\nu) = \inf_{\|g|_{\Lip}\leq 1} \left|\int_{\sX} g(x) \, \mu(dx) - \int_{\sX} g(x) \, \nu(dx) \right|
\end{align*}
In the definition above, without loss of generality, we can add the following additional conditions on $g$ that do not change the result: (i) $g\geq0$ and (ii) $g(y_*)=0$. In view of these additional conditions, we can bound the sup-norm of $g$ as follows
$$
\sup_{x\in\sX} g(x) = \sup_{x\in\sX} g(x) - g(y_*) \leq \sup_{x\in\sX} d_{\sX}(x,y_*) \leq \diam(\sX)
$$
Hence any $g$ with $\|g\|_{\Lip}\leq 1$ satisfying conditions (i) and (ii) is an element of $\F$. Therefore, we can write
\begin{align}\label{wasserstein}
W_1(\mu,\nu) = \inf_{g \in \F} \left|\int_{\sX} g(x) \, \mu(dx) - \int_{\sX} g(x) \, \nu(dx) \right|
\end{align}
This definition of $W_1$ on probability measures is used to prove some of the important results in the paper. 

We now prove a series of results about $W_1$ that are necessary for establishing our main theorems. Note that for any $M\geq1$, the metric space $\sX^M$ is endowed with the following metric $d_{\av}(\bx,\by):=\frac{1}{M} \sum_{i=1}^M d_{\sX}(x_i,y_i)$ in the remainder of the paper.

\begin{lemma}\label{result1}
$W_1(\e[\,\cdot\,|\,\bx],\e[\,\cdot\,|\,\bz]) \leq d_{\av}(\bx,\by)$ for all $\bx,\bz \in \sX^M$. 
\end{lemma}

\begin{proof}
Fix any $\bx,\bz$. Then $\xi(\,\cdot\,):= \frac{1}{M}\sum_{i=1}^M \delta_{(x_i,z_i)}(\,\cdot\,)$ is a coupling of $\e[\,\cdot\,|\,\bx],\e[\,\cdot\,|\,\bz]$ (not necessarily optimal coupling that achieves $W_1(\e[\,\cdot\,|\,\bx],\e[\,\cdot\,|\,\bz])$). Hence
$$
W_1(\e[\,\cdot\,|\,\bx],\e[\,\cdot\,|\,\bz]) \leq \rE_{\xi} [d_{\sX}(x,y)] = \frac{1}{M} \sum_{i=1}^M d_{\sX}(x_i,y_i)
$$
\end{proof}

\begin{lemma}\label{result2}
If $\mu_1,\ldots,\mu_M,\nu_1,\ldots,\nu_M \in \P(\sX)$, then 
$$
W_1(\mu_1\otimes\cdots\otimes\mu_M,\nu_1\otimes\cdots\otimes\nu_M) \leq \frac{1}{M} \sum_{i=1}^M W_1(\mu_i,\nu_i)
$$
\end{lemma}

\begin{proof}
By definition, we have
\begin{align}
&W_1(\mu_1\otimes\cdots\otimes\mu_M,\nu_1\otimes\cdots\otimes\nu_M) \nonumber \\
&= \inf \left\{\rE_{\xi}[d_{\av}(\bx,\by)]: \xi \in \Cp(\mu_1\otimes\cdots\otimes\mu_M,\nu_1\otimes\cdots\otimes\nu_M) \right\} \nonumber \\
&\leq \inf \left\{\rE_{\xi_1\otimes\cdots\otimes\xi_M}[d_{\av}(\bx,\by)]: \xi_i \in \Cp(\mu_i,\nu_i) \, \text{for all} \, i=1,\ldots,M \right\} \nonumber \\
&=\frac{1}{M} \sum_{i=1}^M  \inf \left\{\rE_{\xi_i}[d_{\sX}(x_i,y_i)]: \xi_i \in \Cp(\mu_i,\nu_i)\right\} \nonumber \\
&=:\frac{1}{M} \sum_{i=1}^M W_1(\mu_i,\nu_i)\nonumber 
\end{align}
\end{proof}

\begin{lemma}\label{result3}
$W_1(\e[\,\cdot\,|\,\bx],\e[\,\cdot\,|\,\bz])  = \inf_{\sigma \in \S_{\{1:M\}}} \frac{1}{M} \sum_{i=1}^M d_{\sX}(x_i,z_{\sigma(i)})$ for all $\bx,\bz \in \sX^M$. 
\end{lemma}

\begin{proof}
Let $\xi \in \Cp(\e[\,\cdot\,|\,\bx],\e[\,\cdot\,|\,\bz])$. Hence
\begin{align*}
\int_{\sX\times\sX} d_{\sX}(x,y) \, \xi(dx,dy) &= \int_{\sX\times\sX} d_{\sX}(x,y) \, \xi(dx|y) \, \xi(dy) \\
&= \frac{1}{M} \sum_{i=1}^M \int_{\sX} d_{\sX}(x,z_i) \, \xi(dx|z_i)
\end{align*}
Note that for all $z_i$, the support of $\xi(\,\cdot\,|z_i)$ must be a subset of $\{x_1,\ldots,x_M\}$. This implies that the last term is equal to the following:
\begin{align*}
\frac{1}{M} \sum_{i=1}^M \sum_{j=1}^M  d_{\sX}(x_j,z_i) \, \xi(x_j|z_i) \\
\end{align*}
Define the following matrix $A^{\xi} \in \cR^{M\times M}$ as follows: $A^{\xi}_{j,i} := \xi(x_j|z_i)$ for all $j,i=1,\ldots,M$. Then, for any $i=1,\ldots,M$, we have
\begin{align*}
\sum_{j=1}^M A^{\xi}_{j,i} = \sum_{j=1}^M \xi(x_j|z_i) = 1
\end{align*}
Moreover, for any $j=1,\ldots,M$, we have
\begin{align*}
\sum_{i=1}^M A^{\xi}_{j,i} &= \sum_{i=1}^M \xi(x_j|z_i) \, \xi(z_i) \, \frac{1}{\xi(z_i)} \\
&= M \sum_{i=1}^M \xi(x_j,z_i)  = M \xi(x_j) = 1
\end{align*}
Hence, $A^{\xi}$ is a doubly stochastic matrix. By Birkhoff–von Neumann theorem \cite[Theorem 2.1.6]{BaRa97}, for some $k\geq1$ and $\lambda_1,\ldots,\lambda_k \in (0,1)$, $\sum_{l=1}^k \lambda_l = 1$, we can write 
$$
A^{\xi} = \sum_{l=1}^k \lambda_k \, A^l
$$
where $A^l$ is a permutation matrix which realizes some permutation $\sigma_l \in \S_{\{1:M\}}$, for each $l=1,\ldots,k$.  Hence, we have 
\begin{align*}
\int_{\sX\times\sX} d_{\sX}(x,y) \, \xi(dx,dy) &= \frac{1}{M} \sum_{i=1}^M \sum_{j=1}^M  d_{\sX}(x_j,z_i) \, \xi(x_j|z_i) \\
&= \frac{1}{M} \sum_{i=1}^M \sum_{j=1}^M \sum_{l=1}^k d_{\sX}(x_j,z_i) \, \lambda_k \, A^l_{j,i} \\
&= \sum_{l=1}^k \lambda_k \left\{\frac{1}{M} \sum_{i=1}^M d_{\sX}(x_j,z_{\sigma_l(j)}) \right\} \\
&\geq \inf_{\sigma \in \S_{\{1:M\}}} \frac{1}{M} \sum_{i=1}^M d_{\sX}(x_i,z_{\sigma(i)})
\end{align*}
Since $\xi \in \Cp(\e[\,\cdot\,|\,\bx],\e[\,\cdot\,|\,\bz])$ is arbitrary and $\frac{1}{M} \sum_{i=1}^M \delta_{(x_i,z_{\sigma(i)})} \in \Cp(\e[\,\cdot\,|\,\bx],\e[\,\cdot\,|\,\bz])$ for any $\sigma \in \S_{\{1:M\}}$, this completes the proof.
\end{proof}

\begin{lemma}\label{result4}
$W_1(\e[\,\cdot\,|\,\bx],\e[\,\cdot\,|\,\bz])  \leq \frac{M-1}{M} W_1(\e[\,\cdot\,|\,\bx_{\{2:M\}}],\e[\,\cdot\,|\,\bz_{\{2:M\}}]) + \frac{1}{M} d_{\sX}(x_1,z_1)$ for all $\bx,\bz \in \sX^M$, where $\bx_{\{2:M\}}:=(x_2,\ldots,x_M)$. 
\end{lemma}

\begin{proof}
Fix any $\bx,\bz \in \sX^M$. Then, by Lemma~\ref{result3}, we have 
\begin{align*}
W_1(\e[\,\cdot\,|\,\bx],\e[\,\cdot\,|\,\bz])  &= \inf_{\sigma \in \S_{\{1:M\}}} \frac{1}{M} \sum_{i=1}^M d_{\sX}(x_i,z_{\sigma(i)}) \\
&\leq \inf_{\sigma \in \S_{\{2:M\}}} \frac{1}{M} \left\{ d_{\sX}(x_1,z_1) + \sum_{i=2}^M d_{\sX}(x_i,z_{\sigma(i)}) \right\} \\
&= \frac{1}{M} d_{\sX}(x_1,z_1) + \inf_{\sigma \in \S_{\{2:M\}}} \frac{1}{M}\sum_{i=2}^M d_{\sX}(x_i,z_{\sigma(i)}) \\
&= \frac{1}{M} d_{\sX}(x_1,z_1) + \frac{M-1}{M} W_1(\e[\,\cdot\,|\,\bx_{\{2:M\}}],\e[\,\cdot\,|\,\bz_{\{2:M\}}])
\end{align*}
\end{proof}

\begin{lemma}\label{result5}
For any $x \in \sX,\bx_{\{2:N\}},\by_{\{2:N\}} \in \sX^{N-1}$
$$W_1\left(\frac{1}{N}\delta_x+\frac{1}{N}\sum_{i=2}^N \delta_{x_i},\frac{1}{N}\delta_x+\frac{1}{N}\sum_{i=2}^N \delta_{y_i}\right) = \inf_{\sigma \in \S_{\{2:N\}}} \frac{1}{N} \sum_{i=2}^N d_{\sX}(x_i,y_{\sigma(i)})$$ 
\end{lemma}

\begin{proof}
By Lemma~\ref{result3}, we have 
\begin{align*}
W_1\left(\frac{1}{N}\delta_x+\frac{1}{N}\sum_{i=2}^N \delta_{x_i},\frac{1}{N}\delta_x+\frac{1}{N}\sum_{i=2}^N \delta_{y_i}\right) = \hspace{-5pt} \inf_{\sigma \in \S_{\{1:N\}}} \frac{1}{N} \left( d_{\sX}(x,y_{\sigma(1)}) + \sum_{i=2}^N d_{\sX}(x_i,y_{\sigma(i)}) \right)
\end{align*}
where we let $y_1=x$. If the claim is not true, then there exists $\tilde{\sigma}\in \S_{\{1:N\}}$ such that $\tilde{\sigma}(1)=i\neq1$ and $\tilde{\sigma}(j)=1$ ($j\neq1$) for some $i,j$, and 
\begin{align*}
\frac{1}{N} \left( d_{\sX}(x,y_i) + \sum_{i=2 \atop i\neq j}^N d_{\sX}(x_i,y_{\tilde{\sigma}(i)})  + d_{\sX}(x_j,x)\right) < \inf_{\sigma \in \S_{\{2:N\}}} \frac{1}{N} \sum_{i=2}^N d_{\sX}(x_i,z_{\sigma(i)})
\end{align*}
The last bound implies that 
\begin{align*}
\frac{1}{N} \left( d_{\sX}(x,y_i) + \sum_{i=2 \atop i\neq j}^N d_{\sX}(x_i,y_{\tilde{\sigma}(i)})  + d_{\sX}(x_j,x)\right) <  \frac{1}{N}  \left( \sum_{i=2 \atop i\neq j}^N d_{\sX}(x_i,y_{\tilde{\sigma}(i)})  + d_{\sX}(x_j,y_i)\right) 
\end{align*}
But this is a contradiction as $ d_{\sX}(x,y_i) +  d_{\sX}(x_j,x) \geq d_{\sX}(x_j,y_i)$. 
\end{proof}

We complete this preliminary section by proving two important results that are needed in the sequel.

\begin{lemma}\label{result6}
Let $\bp:\sX\times\sX\rightarrow\P(\sX)$ be a transition probability. Let $\by \in \sX^N$ be given. Suppose that $\{z_i\}_{i=1}^N$ are i.i.d. with common distribution $\e[\,\cdot\,|\,\by]$. Moreover, for all $i=1,\ldots,N$, let $x_i \sim \bp(\,\cdot\,|y_i,z_i)$. Then, for any continuous $g:\sX\rightarrow\cR$, we have 
\begin{align*}
&\rE\left[ \left| \int_{\sX} g(x) \, \e[dx|\bx] - \int_{\sX\times\sX\times\sX} g(x) \, \bp(dx|z,y) \, \e[dz|\by] \, \e[dy|\by] \right| \right]^2  \\
&\leq \frac{1}{N^2} \sum_{i=1}^N \left\{\rE[g(x_i)^2] - \rE[g(x_i)]^2 \right\}
\end{align*}
\end{lemma}

\begin{proof}
By Jensen's inequality, we have 
\begin{align}
&\rE\left[ \left| \int_{\sX} g(x) \, \e[dx|\bx] - \int_{\sX\times\sX\times\sX} g(x) \, \bp(dx|z,y) \, \e[dz|\by] \, \e[dy|\by] \right| \right]^2 \nonumber  \\
&\leq \rE\left[ \left| \int_{\sX} g(x) \, \e[dx|\bx] - \int_{\sX\times\sX\times\sX} g(x) \, \bp(dx|z,y) \, \e[dz|\by] \, \e[dy|\by] \right|^2 \right] \nonumber \\
&=  \rE\left[ \left| \frac{1}{N} \sum_{i=1}^N g(x_i) - \frac{1}{N} \sum_{i=1}^N \int_{\sX} g(x) \, p(dx|y_i,\e[\,\cdot\,|\,\by]) \right|^2 \right] \,\, \text{($p(\,\cdot\,|z,\mu) := \int_{\sX} \bp(\,\cdot\,|z,y) \, \mu(dy)$)}  \nonumber \\
&= \rE\left[\frac{1}{N^2} \sum_{i,j=1}^N g(x_i) g(x_j) \right] + \frac{1}{N^2} \sum_{i,j=1}^N \left\{ \int_{\sX} g(x) \, p(dx|y_i,\e[\,\cdot\,|\,\by]) \int_{\sX} g(x) \, p(dx|y_j,\e[\,\cdot\,|\,\by]) \right\}\nonumber \\
&\phantom{xxxxxxxxxxxxxxxxxxxxxx}-\frac{2}{N^2} \sum_{i,j=1}^N \rE[g(x_i)] \int_{\sX} g(x) \, p(dx|y_j,\e[\,\cdot\,|\,\by]) \label{result6-1}
\end{align}
Let $i\neq j$. Then we have 
\begin{align*}
&\rE[g(x_i)g(x_j)] = \rE[\rE[g(x_i)g(x_j)|z_i,z_j]] \\
&= \rE[\rE[g(x_i)|z_i,z_j] \, \rE[g(x_j)|z_i,z_j]]  \,\, \text{(as $x_i \perp x_j$ given $(z_i,z_j)$)} \\
&= \rE[\rE[g(x_i)|z_i] \, \rE[g(x_j)|z_j]]  \,\, \text{(as $x_i$ only depends on $z_i$ and $x_j$ only depends on $z_j$)} \\
&= \rE[\rE[g(x_i)|z_i]] \, \rE[\rE[g(x_j)|z_j]] \,\, \text{(as $z_i \perp z_j$)} \\
&= \rE\left[\int_{\sX} g(x) \, \bp(dx|y_i,z_i)\right] \rE\left[\int_{\sX} g(x) \, \bp(dx|y_j,z_j)\right] \\
&= \int_{\sX} g(x) \, p(dx|y_i,\e[\,\cdot\,|\,\by]) \int_{\sX} g(x) \, p(dx|y_j,\e[\,\cdot\,|\,\by])
\end{align*}
Moreover, we also have 
\begin{align*}
\rE[g(x_i)]&=\rE[\rE[g(x_i)|z_i]] \\
&= \rE\left[\int_{\sX} g(x) \, \bp(dx|y_i,z_i)\right] \\
&= \int_{\sX} g(x) \, p(dx|y_i,\e[\,\cdot\,|\,\by])
\end{align*}
These two observations imply that 
\begin{align*}
(\ref{result6-1}) = \frac{1}{N^2} \sum_{i=1}^N \left\{\rE[g(x_i)^2] - \rE[g(x_i)]^2 \right\}
\end{align*}
\end{proof}

\begin{lemma}\label{result7}
Let $\bp:\sX\times\sX\rightarrow\P(\sX)$ be a transition probability. Let $\by \in \sX^N$ be given. Suppose that $\{z_i\}_{i=1}^N$ are i.i.d. with common distribution $\e[\,\cdot\,|\,\by]$. Moreover, for all $i=2,\ldots,N$, let $x_i \sim \bp(\,\cdot\,|y_i,z_i)$. Then, for any continuous $g:\sX\rightarrow\cR$, we have 
\begin{align*}
&\rE\left[ \left| \int_{\sX} g(x) \, \e[dx|\bx_{\{2:N\}}] - \int_{\sX\times\sX\times\sX} g(x) \, \bp(dx|z,y) \, \e[dz|\by_{\{2:N\}}] \, \e[dy|\by] \right| \right]^2 \\
&\leq \frac{1}{(N-1)^2} \sum_{i=2}^N \left\{\rE[g(x_i)^2] - \rE[g(x_i)]^2 \right\}
\end{align*}
\end{lemma}

\begin{proof}
The proof is very similar to the proof of Lemma~\ref{result6}, and so, we omit the details.
\end{proof}

\section{Linear Mean-field Games}\label{sec2}

In this section, we introduce $N$-agent linear mean-field games, which are called \emph{linear} because both the transition probability and the one-stage cost function are linearly dependent on the mean-field term. This linear dependence on the mean-field term is the main difference between linear mean-field games and classical mean-field games.

In this game model, we have $N$-agents with the following identical state dynamics
\begin{align}\label{dyn1}
x_i(t+1) = f\left(x_i(t),a_i(t),z_i(t),w_i(t)\right)
\end{align}
where $f:\sX\times\sA\times\sX\times\Omega \rightarrow \sX$ and $(x_i(t),a_i(t))$ is the state-action pair of agent~$i$ at time $t$. Here we have
$$
z_i(t) \sim \e[\,\cdot\,|\,\bx(t)] \,\, \text{and} \,\, w_i(t) \sim \rP_{\Omega}(\,\cdot\,)
$$
We assume that the random variables $\{z_i(t),w_i(t)\}_{N\geq i\geq1,t\geq0}$ are independent of each other and independent over both $i=1,\ldots,N$ and time $t\geq0$. This model can be interpreted using statistical physics terminology as follows: if $x_i(t)$ represents the position of the $i^{th}$ particle at time $t$, then the mean-field term $\e[x_j(t) \, |\, \bx(t)]$ gives the interaction probability of the $i^{th}$ particle with the $j^{th}$ particle. Only one particle can interact with the $i^{th}$ particle at each time $t$. After this random interaction, the $i^{th}$ particle takes some action $a_i(t)$ and moves to the next state through state dynamics $f$. With this interpretation, it is also possible to see that linear mean-field games can be used to model the spread and control of infectious diseases, which is a current research focus due to the COVID-19 pandemic.

If we define the transition probability $\bp:\sX\times\sA\times\sX \rightarrow \P(\sX)$ as follows
\begin{align*}
\bp(B|x,a,z) := \int_{\Omega} 1_{\left\{f(x,a,z,w) \in B\right\}} \, \rP_{\Omega}(dw)
\end{align*}
then we can write
\begin{align*}
\Pr\left\{x_i(t+1) \in B | x_i(t)=x,a_i(t)=a,\e[\,\cdot\,|\,\bx(t)]=\nu\right\}
&= \int_{\Omega\times\sX} 1_{\left\{f(x,a,z,w) \in B\right\}} \, \rP_{\Omega}(dw) \, \nu(dz) \\
&\hspace{-10pt}= \int_{\sX} \bp(B|x,a,z) \, \nu(dz) =: p(B|x,a,\nu)
\end{align*}
Therefore, transition probability $p:\sX\times\sA\times\P(\sX) \rightarrow \P(\sX)$ is linear in the mean-field term $\nu$. At each time step $t\geq0$, agent~$i$ pays some cost via one-stage cost function $\bc:\sX\times\sA\times\sX\rightarrow [0,\infty)$, which depends on $x_i(t)$, $a_i(t)$, and $z_i(t)$ for each $t\geq0$. Define the following function
\begin{align*}
c(x,a,\nu) := \int_{\sX} \bc(x,a,z) \, \nu(dz)
\end{align*}
Note that $c(x,a,\nu)$ is also linear in $\nu$.The linear dependence of $p$ and $c$ on the mean-field term $\nu$ is the primary reason for calling this game model \emph{linear mean-field games}.

We have three classes of policy spaces for each agent. The first one is defined as follows
\begin{align*}
\Pi := \{\pi:\sX^N \rightarrow \P(\sA)\}
\end{align*}
that is, agents that apply policies in $\Pi$ can use the global state vector $\bx(t)$ when designing their controls. The second class of policies can use only the local state information plus the mean-field term
\begin{align*}
\Pi_{\e} := \{\pi:\P(\sX)\times\sX\rightarrow\P(\sA)\}
\end{align*}
Finally, in the last class of policies, agents can use only the local state information
\begin{align*}
\Pi_{l} := \{\pi:\sX\rightarrow\P(\sA)\}
\end{align*}
In this paper, since the state dynamics is time-homogeneous, we assume that agents are only allowed to use stationary policies, meaning that the policies do not change over time.

For a given joint policy $\bpi = (\pi_1,\ldots,\pi_N) \in \Pi^N$, the cost of agent~$i$ is the following discounted cost
\begin{align*}
J_i(\mu;\bpi) &:= \rE_{\mu}^{\bpi} \left[\sum_{t=0}^{\infty} \beta^t \, \bc\left(x_i(t),a_i(t),z_i(t)\right) \right] \\
&= \rE_{\mu}^{\bpi} \left[\sum_{t=0}^{\infty} \beta^t \, c\left(x_i(t),a_i(t),\e[\,\cdot\,|\,\bx(t)]\right) \right]
\end{align*}
where $\beta \in (0,1)$ is the discount factor and $\mu$ is the common initial distribution; that is,
$$
(x_1(0),\ldots,x_N(0)) \sim \mu^{\otimes N} := \bigotimes_{i=1}^N \mu
$$

Let us now give the definition of $\varepsilon$-Nash equilibrium, which is the canonical optimality notion adopted in game theory. 

\begin{definition}[$\varepsilon$-Nash equilibrium]
For a given $\varepsilon \geq 0$, a joint policy $\bpi^* = (\pi_1^*,\ldots,\pi_N^*)$ is $\varepsilon$-Nash equilibrium if 
$$
J_i(\mu;\bpi^*) \leq \inf_{\pi \in \Pi} J_i(\mu;\pi,\bpi^*_{-i}) + \varepsilon
$$
for all $i=1,\ldots,N$, where $\bpi^*_{-i} := \{\pi_j^*\}_{j \neq i}$. If $\varepsilon=0$, then we have Nash equilibrium. 
\end{definition}

In this paper, our goal is to obtain approximate Nash equilibrium for $N$-agent games, where $N$ is assumed to be sufficiently large, via studying the stationary infinite population limit. More precisely, we establish that if all agents apply the equilibrium policy in the infinite population limit, which is introduced in the next section, this joint policy is proved to be $\varepsilon(N)$-Nash equilibrium for $N$-agent games, where there is an explicit relation between $\varepsilon(N)$ and $N$.

\subsection{Stationary Infinite Population Limit}\label{sec2sub1}

In this section, we introduce the stationary infinite population limit of the game introduced in the preceding section. In the limiting case, for each $t\geq0$, we pretend that the mean-field term $\e[\,\cdot\,|\,\bx(t)]$ converges (in some sense) to the deterministic probability measure $\mu$ as $N\rightarrow\infty$. Therefore, under this convergence assumption, given the limiting distribution $\mu$, which characterizes the collective behavior of all agents in the infinite population limit, a generic agent has the following state dynamics
\begin{align}\label{dyn2}
x(t+1) = f\left(x(t),a(t),z(t),w(t)\right)
\end{align}
where 
$$
z(t) \sim \mu(\,\cdot\,) \,\, \text{and} \,\, w(t) \sim \rP_{\Omega}(\,\cdot\,)
$$
Here, we can view $z(t)$ as a noise since its distribution does not depend on the state. Hence, state process $\{x(t)\}_{t\geq0}$ becomes a Markov decision process (MDP) with the following transition probability
$$
x(t+1) \sim p(\,\cdot\,|\,x(t),a(t),\mu) := \int_{\sX} \bp(\,\cdot\,|\,x(t),a(t),z) \, \mu(dz)
$$
In this MDP, the one-stage cost function $\bc$ is a function of $x(t)$, $a(t)$, and $z(t)$ at time $t$ or alternatively, one-stage cost function $c$ is a function of $x(t)$, $a(t)$, and $\mu$ when one views $z(t)$ as a noise. In this model, for any policy $\pi \in \Pi_l$, the discounted cost of a generic agent is the following
\begin{align*}
J(\mu;\pi) &:= \rE_{\mu}^{\pi} \left[\sum_{t=0}^{\infty} \beta^t \, \bc\left(x(t),a(t),z(t)\right) \right] \\
&= \rE_{\mu}^{\pi} \left[\sum_{t=0}^{\infty} \beta^t \, c\left(x(t),a(t),\mu \right) \right]
\end{align*}
where we also assume that $x(0) \sim \mu$; that is, the initial distribution is the same as the limiting mean-field term $\mu$. Here, in the infinite population limit, a generic agent is only allowed to use its local state information when applying its action since the corresponding MDP is time-homogeneous. Therefore, we only consider policies in the policy space $\Pi_l$. A policy $\gamma \in \Pi_l$ is optimal for $\mu$ if
$$
J(\mu;\gamma) = \inf_{\pi \in \Pi_l} J(\mu;\pi)
$$ 
In view of this, let us define the following set-valued map
\begin{align*}
\Lambda(\mu) := \left\{\pi \in \Pi_l: \pi \, \text{is optimal for} \, \mu \right\}
\end{align*}
Under mild regularity conditions on $\bp$ and $\bc$, $\Lambda(\mu)$ is non-empty for all $\mu \in \P(\sX)$.

Now we define another set-valued map for any policy $\pi$. To this end, given $\pi \in \Pi_l$, let us define the following transition probability for any $\mu \in \P(\sX)$
$$
p_{\mu}^{\pi}(\,\cdot\,|\,x) := \int_{\sA} p(\,\cdot\,|\,x,a,\mu) \, \pi(da|x)
$$
In view of this, let us define the following set-valued map
\begin{align*}
\Phi(\pi) := \left\{\mu \in \P(\sX): \mu \, \text{is an invariant distribution of} \,\, p_{\mu}^{\pi} \right\}
\end{align*}
Again, under mild regularity conditions on $\bp$, it is possible to prove that $\Phi(\pi)$ is non-empty for all $\pi$ or at least for all continuous $\pi$ under weak convergence topology. 

Now it is time to introduce the equilibrium notion adapted in the infinite population limit, which is called stationary mean-field equilibrium (MFE).

\begin{definition}[Stationary Mean-field equilibrium]
A pair $(\mu^*,\pi^*)$ is a stationary mean-field equilibrium if $\pi^* \in \Lambda(\mu^*)$ and $\mu^* \in \Phi(\pi^*)$.
\end{definition}

Under mild regularity conditions on $\bp$ and $\bc$, it is possible to prove that MFE exists \cite[Theorem 3.3]{SaBaRa18}. For instance, if $\bp$ and $\bc$ are continuous, and $\sX$ is compact, then there exists a MFE. This can be established via the method that is used to prove \cite[Theorem 3.3]{SaBaRa18}. As the conditions for the existence of MFE are weaker than the assumptions that are introduced in the next section, in the remainder of this paper, we suppose that there exists at least one MFE $(\mu^*,\pi^*)$. 

\subsubsection{Assumptions}

We now state the assumptions that are used throughout Section~\ref{sec2} and Section~\ref{sec3}. To avoid potential compactness issues, we only consider games with compact state spaces.

\begin{itemize}
\item[(a)] $\sX$ is a compact Borel space and $\sA$ is a convex and compact subset of some finite-dimensional Euclidean space. 
\item[(b)] The one-stage cost function $\bc:\sX\times\sA\times\sX\rightarrow [0,\infty)$ is Lipschitz continuous with Lipschitz constants $(\L_1,\L_2,\L_3)$; that is, for any $(x,a,y),(z,b,r) \in \sX\times\sA\times\sX$, we have
$$
|\bc(x,a,y)-\bc(z,b,r)| \leq \L_1 \, d_{\sX}(x,z) + \L_2 \, \|a-b\| + \L_3 \, d_{\sX}(y,r)
$$
\item[(c)] The transition probability $\bp:\sX\times\sA\times\sX\rightarrow\P(\sX)$ is Lipschitz continuous with Lipschitz constants $(\K_1,\K_2,\K_3)$; that is, for any $(x,a,y),(z,b,r) \in \sX\times\sA\times\sX$, we have
$$
W_1(\bp(\,\cdot\,|x,a,y),\bp(\,\cdot\,|z,b,r)) \leq \K_1 \, d_{\sX}(x,z) + \K_2 \, \|a-b\| + \K_3 \, d_{\sX}(y,r)
$$
\end{itemize}

One can also prove that $c$ and $p$ are also Lipschitz continuous with the same Lipschitz constants as in (b) and (c), respectively. To define the next assumption, let us introduce the following continuous function
$$
\rL: [0,\infty) \times [0,\infty) \times [0,\infty) \rightarrow [0,\infty) \times [0,\infty)
$$

\begin{itemize}
\item[(d)] For any Lipschitz continuous $J:\sX\times\P(\sX)\times\sX \rightarrow [0,\infty)$ with Lipschitz constants $(\R_1,\R_2,\R_3)$, define 
$$
F^J(x,a,\mu) := c(x,a,\mu) + \beta \int_{\sX} J(y,\mu,x) \, p(dy|x,a,\mu)
$$
We assume that for any $(x,\mu) \in \sX\times\P(\sX)$, $F^J(x,\cdot,\mu)$ is $\rho$-strongly convex; that is, for any $a,b \in \sA$, we have
$$
F^J(x,b,\mu) - F^J(x,a,\mu) \geq \langle \nabla_a F^J(x,a,\mu), b-a \rangle + \frac{\rho}{2} \, \|b-a\|^2
$$
where $\nabla_a F^J(x,a,\mu)$ is the gradient of $F^J(x,a,\mu)$ with respect to $a$. Moreover, the gradient $\nabla_a F^J(x,a,\mu)$ is Lipschitz continuous in $(x,\mu)$ for any $a \in \sA$ with Lipschitz constants $\rL(\R_1,\R_2,\R_3)=\left(\rL_1(\R_1,\R_2,\R_3),\rL_2(\R_1,\R_2,\R_3)\right)$; that is
\begin{align*}
\sup_{a \in \sA} \, \|\nabla_a F^J(x,a,\mu) - \nabla_a F^J(z,a,\nu)\| \leq \rL_1(\R_1,\R_2,\R_3)  \, d_{\sX}(x,z) + \rL_2(\R_1,\R_2,\R_3) \, W_1(\mu,\nu)
\end{align*}
for all $(x,\mu), (z,\nu) \in \sX\times\P(\sX)$. 
\item[(e)] Let $\L_F := \rL_1\left(\frac{\L_1}{1-\beta \K_1},0,0\right)$. Then define $\L^* := \frac{2 \L_F}{\rho}$. We assume that $\beta (\K_1+\K_2 \L^* +\K_3)< 1$. Moreover, we also assume that $\beta (\K_1+\K_2 \L^*_2)< 1$, where 
$$\L_2^* := \frac{2 \rL_2(\R_1,\R_2,0)}{\rho}, \,\, \R_1:= \frac{\L_1}{1-\beta \, \K_1}, \,\, \R_2:=\frac{\L_3 \, (\K_1+\K_2 \L^* +\K_3)}{1-\beta (\K_1+\K_2 \L^* +\K_3)}$$
\end{itemize}

\subsubsection{Sufficient Conditions for Assumption (d)}

Assumption (d) is the most difficult assumption to verify. In this section, we provide some easy-to-check sufficient conditions in terms of $\bp$ and $\bc$ that imply assumption (d). To this end, we first assume that for any $(x,a,z) \in \sX\times\sA\times\sX$, the transition probability $\bp(\,\cdot\,|x,a,z)$ is absolutely continuous with respect to $\m \in \P(\sX)$. Let $\eta(y,x,a,z)$ be the density function of $\bp(dy|x,a,z)$ with respect to $\m$; that is
$$
\bp(dy|x,a,z) = \int_{\{y \in dy\}} \eta(y,x,a,z) \, \m(dy)
$$ 
for any $(x,a,z) \in \sX\times\sA\times\sX$. We also assume that for any $(x,z) \in \sX\times\sX$, $\bc(x,\cdot,z)$ is $\rho$-strongly convex and for any $(y,x,z) \in \sX\times\sX\times\sX$, $\eta(y,x,\cdot,z)$ is convex. Moreover, $\nabla_a \bc(x,a,z)$ is Lipschitz continuous in $(x,z)$ for all $a \in \sA$ with Lipschitz constants $(\L^{\g}_1,\L^{\g}_2)$ and $\nabla_a \eta(y,x,a,z)$ is Lipschitz continuous in $(x,z)$ for all $(y,a) \in \sX\times\sA$ with Lipschitz constants $(\K^{\g}_1,\K^{\g}_2)$. Then, assumption (d) holds with the following function $\rL$
$$
\rL(\R_1,\R_2,\R_3) := (\rL_1(\R_1,\R_2,\R_3),\rL_2(\R_1,\R_2,\R_3))
$$
where 
\footnotesize
\begin{align*}
&\rL_1(\R_1,\R_2,\R_3) \\
&:= \left\{\L_1^{\g}+ \beta \, \left(\R_3 \, \sup_{(x,a,z) \in \sX\times\sA\times\sX} \int_{\sX} \|\nabla_a \eta(y,x,a,z)\| \, \m(dy) +\K_1^{\g} \, \R_1 \, \int_{\sX} d_{\sX}(y,y_*) \, \m(dy)\right) \right\} \\
&\rL_2(\R_1,\R_2,\R_3)  \\ 
&:= \left\{\L_2^{\g} \, \fC_{2,\infty} + \beta \, \left(\R_2 \, \sup_{(x,a,z) \in \sX\times\sA\times\sX} \int_{\sX} \|\nabla_a \eta(y,x,a,z)\| \, \m(dy) +\K_2^{\g} \, \fC_{2,\infty} \, \R_1 \, \int_{\sX} d_{\sX}(y,y_*) \, \m(dy)\right) \right\}
\end{align*}
\normalsize
Indeed let  $J:\sX\times\P(\sX)\times\sX \rightarrow [0,\infty)$ be Lipschitz continuous with Lipschitz constants $(\R_1,\R_2,\R_3)$, and define 
\begin{align*}
F(x,a,\mu) &:= c(x,a,\mu) + \beta \int_{\sX} J(y,\mu,x) \, p(dy|x,a,\mu) \\
&= \int_{\sX} \bc(x,a,z) \, \mu(dz) + \beta \int_{\sX\times\sX} J(y,\mu,x) \, \eta(y,x,a,z) \, \m(dy) \, \mu(dz)
\end{align*}
Obviously, for any $(x,\mu) \in \sX\times\P(\sX)$, $F(x,\cdot,\mu)$ is $\rho$-strongly convex as $\bc(x,\cdot,z)$ is $\rho$-strongly convex and $\eta(y,x,\cdot,z)$ is convex. Moreover, for any $a \in \sA$ and $(x,\mu), (\tx,\nu) \in \sX\times\P(\sX)$, we have 
\small
\begin{align}
&\|\nabla_a F(x,a,\mu) - \nabla_a F(\tx,a,\nu)\| 
\leq \left\|\int_{\sX} \nabla_a \bc(x,a,z) \, \mu(dz)-\int_{\sX} \nabla_a \bc(\tx,a,z) \, \nu(dz) \right\| \nonumber \\
&+ \beta \, \left\|\int_{\sX\times\sX} J(y,\mu,x) \, \nabla_a \eta(y,x,a,z) \, \m(dy) \, \mu(dz) - \int_{\sX\times\sX} J(y,\nu,\tx) \, \nabla_a \eta(y,\tx,a,z) \, \m(dy) \, \nu(dz) \right\| \nonumber \\
&\leq \int_{\sX} \| \nabla_a \bc(x,a,z) -\nabla_a \bc(\tx,a,z) \|\, \mu(dz)  + \left\|\int_{\sX} \nabla_a \bc(\tx,a,z) \, \mu(dz)-\int_{\sX} \nabla_a \bc(\tx,a,z) \, \nu(dz) \right\| \nonumber \\
&+ \beta \, \left\|\int_{\sX\times\sX} J(y,\mu,x) \, \nabla_a \eta(y,x,a,z) \, \m(dy) \, \mu(dz) - \int_{\sX\times\sX} J(y,\nu,\tx) \, \nabla_a \eta(y,x,a,z) \, \m(dy) \, \mu(dz) \right\| \nonumber \\
&+ \beta \, \left\|\int_{\sX\times\sX} J(y,\nu,\tx) \, \nabla_a \eta(y,x,a,z) \, \m(dy) \, \mu(dz) - \int_{\sX\times\sX} J(y,\nu,\tx) \, \nabla_a \eta(y,\tx,a,z) \, \m(dy) \, \mu(dz) \right\| \nonumber \\
&+ \beta \, \left\|\int_{\sX\times\sX} J(y,\nu,\tx) \, \nabla_a \eta(y,\tx,a,z) \, \m(dy) \, \mu(dz) - \int_{\sX\times\sX} J(y,\nu,\tx) \, \nabla_a \eta(y,\tx,a,z) \, \m(dy) \, \nu(dz) \right\| \nonumber \\
&\leq \L_1^{\g} \, d_{\sX}(x,\tx) + \fC_{2,\infty} \, \sup_{i=1,\ldots,\dim(\sA)} \left|\int_{\sX} \nabla_{a_i} \bc(\tx,a,z) \, \mu(dz)-\int_{\sX} \nabla_{a_i} \bc(\tx,a,z) \, \nu(dz) \right| \nonumber \\
&+ \beta \, \int_{\sX\times\sX} |J(y,\mu,x)-J(y,\nu,\tx)| \, \|\nabla_a \eta(y,x,a,z)\| \, \m(dy) \, \mu(dz) \nonumber \\
&+ \beta \, \int_{\sX\times\sX} |J(y,\nu,\tx)| \, \|\nabla_a \eta(y,x,a,z) -\nabla_a \eta(y,\tx,a,z)\| \, \m(dy) \, \mu(dz)  \nonumber \\ 
&+ \beta  \, \fC_{2,\infty} \, \sup_{i=1,\ldots,\dim(\sA)} \bigg| \int_{\sX\times\sX} J(y,\nu,\tx) \, \nabla_{a_i} \eta(y,\tx,a,z) \, \m(dy) \, \mu(dz) \nonumber \\
&\phantom{xxxxxxxxxxxxxxxxxxxxx}- \int_{\sX\times\sX} J(y,\nu,\tx) \, \nabla_{a_i} \eta(y,\tx,a,z) \, \m(dy) \, \nu(dz) \bigg| \nonumber \\
&\text{($\fC_{2,\infty}$ is the constant that let us to go from Euclidean norm to sup-norm on $\sA$: $\|a\| \leq \fC_{2,\infty} \, \|a\|_{\infty}$)} \nonumber \\
&\leq \L_1^{\g} \, d_{\sX}(x,\tx) + \fC_{2,\infty} \, \L_2^{\g} \, W_1(\mu,\nu) \nonumber \\
&+ \beta \left(\R_2 \, W_1(\mu,\nu) + \R_3 \, d_{\sX}(x,\tx) \right) \sup_{(x,a,z) \in \sX\times\sA\times\sX} \int_{\sX} \|\nabla_a \eta(y,x,a,z)\| \, \m(dy) \nonumber \\
&+ \beta \, \K_1^{\g} \, d_{\sX}(x,\tx) \, \int_{\sX} |J(y,\nu,\tx)| \, \m(dy) \nonumber \\
&+ \beta \, \fC_{2,\infty} \, \K_2^{\g} \, W_1(\mu,\nu) \, \int_{\sX} |J(y,\nu,\tx)| \, \m(dy) \nonumber
\end{align}
\normalsize
where  the last bound follows from the facts that (i) $\nabla_{a_i} \bc(\tx,a,\cdot)$ is $\L_2^{\g}$-Lipschitz continuous for all $(\tx,a) \in \sX\times\sA$ and $i=1,\ldots,\dim(\sA)$ and (ii) $ \int_{\sX} J(y,\nu,\tx) \, \nabla_{a_i} \eta(y,\tx,a,\cdot) \, \m(dy)$ is $\K_2^{\g} \, \int_{\sX} |J(y,\nu,\tx)| \, \m(dy)$-Lipschitz continuous for all $(\tx,a) \in \sX\times\sA$ and $i=1,\ldots,\dim(\sA)$. Note that without loss of generality, we can assume that $J(y_*,\nu,\tx) = 0$ for all $(\tx,a) \in \sX\times\sA$, and so,
\begin{align*}
\int_{\sX} |J(y,\nu,\tx)| \, \m(dy) &= \int_{\sX} |J(y,\nu,\tx)-J(y_*,\nu,\tx)| \, \m(dy) \\
&\leq \R_1 \, \int_{\sX} d_{\sX}(y,y_*) \, \m(dy)
\end{align*}  
This implies that
\small
\begin{align}
&\|\nabla_a F(x,a,\mu) - \nabla_a F(z,a,\nu)\| \nonumber \\ 
&\leq \left\{\L_1^{\g}+ \beta \, \left(\R_3 \hspace{-15pt} \sup_{(x,a,z) \in \sX\times\sA\times\sX} \int_{\sX} \|\nabla_a \eta(y,x,a,z)\| \, \m(dy) +\K_1^{\g} \, \R_1 \, \int_{\sX} d_{\sX}(y,y_*) \, \m(dy)\right) \right\} \, d_{\sX}(x,\tx) \nonumber \\
&+ \bigg\{\L_2^{\g} \, \fC_{2,\infty} + \beta \, \bigg(\R_2 \hspace{-15pt} \sup_{(x,a,z) \in \sX\times\sA\times\sX} \int_{\sX} \|\nabla_a \eta(y,x,a,z)\| \, \m(dy) \nonumber \\
&\phantom{xxxxxxxxxxxxxxxxxxxxxxxxxxxxxxxx}+\K_2^{\g} \, \fC_{2,\infty} \, \R_1 \, \int_{\sX} d_{\sX}(y,y_*) \, \m(dy)\bigg) \bigg\} 
W_1(\mu,\nu) \nonumber \\
&=: \rL_1(\R_1,\R_2,\R_3) \, d_{\sX}(x,\tx) + \rL_2(\R_1,\R_2,\R_3) \, W_1(\mu,\nu) \label{aux-eq5}
\end{align}
\normalsize
Note that in order to have the bound (\ref{aux-eq5}), we must have a Lipschitz function $J$ inside the integral in the definition of $F$. Without Lipschitz continuity of $J$, it is not possible to establish the same result. In the absence of Lipschitz continuity of $J$, we must substantially strengthen the assumptions on $\bc$ and $\bp$ to obtain a similar bound. For instance, instead of Wasserstein distance of order~1, we might need to use total variation distance, which is much stronger than $W_1$, in assumptions (c) and (d).

We also note that the bound in (\ref{aux-eq5}) is fairly general and crude. By using further properties of the transition probability and the one-stage cost function in addition to conditions above in specific examples, one can significantly improve this bound.

\subsection{Lipschitz Continuity of MFE Policy}\label{sec2sub2}

In this section, we establish that the policy $\pi^*$ in the mean-field equilibrium $(\mu^*,\pi^*)$ is Lipschitz continuous. Recall that given any mean-field term $\mu$, the corresponding optimization problem in the infinite population limit can be formulated as a Markov decision processes, denoted as $\text{MDP}_{\mu}$, with the following components
\begin{align*}
\left\{\sX,\sA,c_{\mu},p_{\mu},\mu\right\}
\end{align*}
where 
\begin{align*}
c_{\mu}(x,a) &:= c(x,a,\mu) = \int_{\sX} \bc(x,a,z) \, \mu(dz) \\
p_{\mu}(\,\cdot\,|\,x,a) &:= p(\,\cdot\,|\,x,a,\mu) = \int_{\sX} \bp(\,\cdot\,|\,x,a,z) \, \mu(dz)
\end{align*}
Note that both $c_{\mu}$ and $p_{\mu}$ are Lipschitz continuous with Lipschitz constants $(\L_1,\L_2)$ and $(\K_1,\K_2)$, respectively, under assumptions (b) and (c). Then, by \cite[Theorem 5.1]{SaYuLi17}, the optimal value function of $\text{MDP}_{\mu}$, denoted as $J_{\mu}$, is Lipschitz continuous with Lipschitz constant $\frac{\L_1}{1-\beta \K_1}$. We now prove that the optimal policy of $\text{MDP}_{\mu}$, denoted as  $\pi_{\mu}$, is Lipschitz continuous. 

\begin{lemma}\label{lemma1}
Given any $\mu$, there exits a single element $\pi_{\mu}$ in the set $\Lambda(\mu)$, which is the unique optimal policy of $\text{MDP}_{\mu}$. Moreover, $\pi_{\mu}$ is Lipschitz continuous with Lipschitz constant $\frac{2\L_F}{\rho} =: \L^*$.
\end{lemma}

\begin{proof}
Define 
$$
F_{\mu}(x,a) := c_{\mu}(x,a) + \beta \int_{\sX} J_{\mu}(y) \, p_{\mu}(dy|x,a).
$$
For any $x \in \sX$, $F_{\mu}(x,\cdot)$ is $\rho$-strongly convex by assumption (d); that is, for any $a,\ha \in \sA$, we have 
$$
F_{\mu}(x,\ha) \geq F_{\mu}(x,a) + \langle \nabla_a F_{\mu}(x,a), \ha-a \rangle + \frac{\rho}{2} \|a-\ha\|^2.
$$
By strong convexity of $F_{\mu}(x,\cdot)$, there exists an unique optimal policy $\pi_{\mu}:\sX\rightarrow\sA$ as the following optimality equation admits unique solutions for all $x \in \sX$
$$
\min_{a \in \sA} \left[c_{\mu}(x,a) + \beta \int_{\sX} J_{\mu}(y) \, p_{\mu}(dy|x,a) \right] =: \min_{a \in \sA} F_{\mu}(x,a) 
$$
To prove the Lipschitz continuity of $\pi_{\mu}$, which is the unique minimizer of the above optimality equation for all $x \in \sX$, fix any $x,y \in \sX$. Then, we have 
\begin{align*}
0 &\geq F_{\mu}(x,\pi_{\mu}(x)) - F_{\mu}(x,\pi_{\mu}(y)) \\
&\geq \langle \nabla_a F_{\mu}(x,\pi_{\mu}(y)), \pi_{\mu}(x)-\pi_{\mu}(y) \rangle +\frac{\rho}{2} \|\pi_{\mu}(x)-\pi_{\mu}(y)\|^2
\end{align*}
Hence
$$
\|\pi_{\mu}(x)-\pi_{\mu}(y)\|^2 \leq \frac{2}{\rho} \langle -\nabla_a F_{\mu}(x,\pi_{\mu}(y)), \pi_{\mu}(x)-\pi_{\mu}(y) \rangle
$$
By first order optimality condition, we also have 
$$
\langle \nabla_a F_{\mu}(y,\pi_{\mu}(y)), a-\pi_{\mu}(y) \rangle \geq 0 
$$
for all $a \in \sA$. Therefore, by combining these two bounds, we can write the following 
\begin{align*}
\|\pi_{\mu}(x)-\pi_{\mu}(y)\|^2 &\leq  \frac{2}{\rho} \langle \nabla_a F_{\mu}(y,\pi_{\mu}(y))-\nabla_a F_{\mu}(x,\pi_{\mu}(y)), \pi_{\mu}(x)-\pi_{\mu}(y) \rangle \\
&\leq \frac{2}{\rho} \|\nabla_a F_{\mu}(y,\pi_{\mu}(y))-\nabla_a F_{\mu}(x,\pi_{\mu}(y))\| \|\pi_{\mu}(x)-\pi_{\mu}(y) \| \\
&\leq \frac{2}{\rho} \L_F \, d_{\sX}(x,y) \, \|\pi_{\mu}(x)-\pi_{\mu}(y) \|
\end{align*}
This implies that $\|\pi_{\mu}(x)-\pi_{\mu}(y)\| \leq \frac{2\L_F}{\rho} \, d_{\sX}(x,y) =: \L^* \, d_{\sX}(x,y)$ for all $x,y \in \sX$, which completes the proof.  
\end{proof}

Lemma~\ref{lemma1} implies that if $(\mu^*,\pi^*)$ is a mean-field equilibrium, then $\pi^* = \Lambda(\mu^*)$ is Lipschitz continuous with Lipschitz constant $\L^*$.

\subsection{Lipschitz Continuity of Best Response to MFE Policy}\label{sec2sub2}

Let $(\mu^*,\pi^*)$ be a mean-field equilibrium where $\pi^*:\sX\rightarrow\sA$ is $\L^*$-Lipschitz continuous. Suppose that all agents, except agent~$1$, apply the policy $\pi^*$. Since agent~$1$ is allowed to observe the entire state vector $\bx(t)$ for all $t\geq 0$, it faces with a Markov decision process, denoted as $\text{MDP}_{N}$, with the following components in response to the MFE policy $\pi^*$ used by other agents
$$
\left\{\sX^N,\sA,C,P,\mu^{*,\otimes^N}\right\}
$$
where 
\begin{align*}
P(d\by|\bx,a) &:= \prod_{i=2}^N p(dy_i|x_i,\pi^*(x_i),\e[\,\cdot\,|\,\bx]) \bigotimes p(dy_1|x_1,a,\e[\,\cdot\,|\,\bx]) \\
C(\bx,a) &:= c(x_1,a,\e[\,\cdot\,|\,\bx])
\end{align*}
Recall that we endow $\sX^N$ with the following metric
$$
d_{\av}(\bx,\by) := \frac{1}{N} \sum_{i=1}^N d_{\sX}(x_i,y_i)
$$
The optimality equation for $\text{MDP}_{N}$ is the following
\begin{align*}
V^*(\bx) = \min_{a \in \sA} \left[ C(\bx,a) + \beta \int_{\sX^N} V^*(\by) \, P(d\by|\bx,a) \right]
\end{align*}
where 
$$
V^*(\bx) := \inf_{\pi \in \Pi} \rE_{\bx}^{\pi} \left[ \sum_{t=0}^{\infty} \beta^t C(\bx(t),a(t)) \right]
$$
In subsequent two sections, we prove some properties of $V^*$. The first one is the symmetry property.

\subsubsection{Symmetry of $V^*$}

In this section, we prove that $V^*$ is symmetric in the variables $(x_2,\ldots,x_N)$. To this end, we define 
$$
C_{\sym}(\sX^N) := \left\{ V \in C_b(\sX^N): \text{$V$ is symmetric in $(x_2,\ldots,x_N)$} \right\}
$$
We also define the optimality operator $T$ as follows
$$
TV(\bx) := \min_{a \in \sA} \left[ C(\bx,a) + \beta \int_{\sX^N} V(\by) \, P(d\by|\bx,a) \right]
$$
for all $\bx \in \sX^N$. Note that $T$ is a $\beta$-contraction on $C_b(\sX^N)$ with an unique fixed point $V^*$. 

\begin{lemma}\label{lemma2}
$T$ maps $C_{\sym}(\sX^N)$ into itself. 
\end{lemma}

\begin{proof}
Let $V \in C_{\sym}(\sX^N)$. Then
\small
\begin{align*}
TV(\bx) = \min_{a \in \sA} \left[ c(x_1,a,\e[\,\cdot\,|\,\bx]) + \beta \int_{\sX^N} \hspace{-5pt} V(\by) \, \prod_{i=2}^N p(dy_i|x_i,\pi^*(x_i),\e[\,\cdot\,|\,\bx]) \bigotimes p(dy_1|x_1,a,\e[\,\cdot\,|\,\bx]) \right]
\end{align*}
\normalsize
Let $\sigma \in \S_{\{2:N\}}$. We apply this permutation to the above equation by noting the fact that $\e[\,\cdot\,|\,\bx] = \e[\,\cdot\,|\,x_1,\bx_{\sigma(2:N)}]$ 
\small
\begin{align*}
&TV(x_1,\bx_{\sigma(2:N)}) \\
&= \min_{a \in \sA} \left[ c(x_1,a,\e[\,\cdot\,|\,\bx]) + \beta \int_{\sX^N} \hspace{-5pt} V(\by) \, \prod_{i=2}^N p(dy_i|x_{\sigma(i)},\pi^*(x_{\sigma(i)}),\e[\,\cdot\,|\,\bx] ) \bigotimes p(dy_1|x_1,a,\e[\,\cdot\,|\,\bx] ) \right] \\
&= \min_{a \in \sA} \left[ c(x_1,a,\e[\,\cdot\,|\,\bx]) + \beta \int_{\sX^N} \hspace{-5pt}V(y_1,\by_{\sigma(2:N)}) \, \prod_{i=2}^N p(dy_i|x_i,\pi^*(x_i),\e[\,\cdot\,|\,\bx] ) \bigotimes p(dy_1|x_1,a,\e[\,\cdot\,|\,\bx] ) \right] \\
&= \min_{a \in \sA} \left[ c(x_1,a,\e[\,\cdot\,|\,\bx]) + \beta \int_{\sX^N} \hspace{-5pt} V(y_1,\by_{\{2:N\}}) \, \prod_{i=2}^N p(dy_i|x_i,\pi^*(x_i),\e[\,\cdot\,|\,\bx] ) \bigotimes p(dy_1|x_1,a,\e[\,\cdot\,|\,\bx] ) \right] \\
&\text{(since $V \in C_{\sym}(\sX^N)$)} \\
&=: TV(\bx)
\end{align*}
\normalsize
Therefore, $TV \in C_{\sym}(\sX^N)$, which completes the proof.
\end{proof}

Note that given any $V \in C_{\sym}(X^N)$, the sequence $\{T^nV\}_{n\geq 1}$ converges in sup-norm to $V^*$ by Banach fixed point theorem, where $T^n V := TT^{n-1}V$ for any $n\geq2$. Since $C_{\sym}(X^N)$ is a closed subset of $C_b(X^N)$ in sup-norm topology, we have $V^* \in C_{\sym}(X^N)$; that is, the optimal value function $V^*$ is symmetric in $(x_2,\ldots,x_N)$.

\subsubsection{Lipschitz Continuity of $V^*$}

In this section, we prove that $V^*$ is Lipschitz continuous in some sense that will be made precise in the following lemma. Before stating the lemma, let us introduce the following constants
\begin{alignat*}{3}
\a_1^N&:=\L_1 + \frac{\L_3}{N}, \,\,\,\, &&\b_1^N:=\beta \left(\K_1 +\frac{\K_3}{N}\right), \,\,\,\, &&\c_1^N:= \beta \frac{\K_3}{N} \\
\a_2&:=\L_3,\,\, &&\b_2:=\beta \K_3, \,\, &&\c_2:=\beta (\K_1+\K_2 \L^* +\K_3)
\end{alignat*}
By assumption (e), $\c_2<1$. In addition to assumptions (a)-(e), we impose the following condition on $N$ in the remainder of the paper.
\begin{itemize}
\item[(f)] We assume that $N$ is large enough so that $\b_1^N<1$ and 
$$
(1-\b_1^N)(1-\c_2)-\c_1^N\a_2\b_2 > 0, \,\,\,\, (1-\b_1^N)(1-\c_2)-\c_1^N\a_1^N\b_2>0
$$
\end{itemize}
This assumption holds for sufficiently large $N$ values since $\b_1^N \rightarrow \beta \K_1 < \c_2$ and $\c_1^N \rightarrow 0$, as $N \rightarrow \infty$.

\begin{lemma}\label{lemma3}
The optimal value function $V^*$ satisfies the following Lipschitz bound
\begin{align*}
|V^*(\bx)-V^*(\by)| \leq \K_1^{*,N} d_{\sX}(x_1,y_1) + \K_2^{*,N} W_1(\e[\,\cdot\,|\,\bx_{\{2:N\}}],\e[\,\cdot\,|\,\by_{\{2:N\}}])
\end{align*}
where 
$$
\K_1^{*,N}= \frac{\a_1^N (1-\c_2)}{(1-\b_1^N)(1-\c_2)-\c_1^N\a_2\b_2}, \,\,\, \K_2^{*,N}= \frac{\a_2 (1-\b_1^N)}{(1-\b_1^N)(1-\c_2)-\c_1^N\a_1^N\b_2} 
$$
\end{lemma}

\begin{proof}
Let $V \in C_{\sym}(\sX^N)$ and let $V$ also satisfy the following Lipschitz bound
\begin{align}\label{aux1}
|V(\bx)-V(\by)| \leq K_1 d_{\sX}(x_1,y_1) + K_2 W_1(\e[\,\cdot\,|\,\bx_{\{2:N\}}],\e[\,\cdot\,|\,\by_{\{2:N\}}])
\end{align}
for some $K_1,K_2>0$. Then, for any $\bx,\bz \in \sX^N$, we have 
\begin{align}
|TV(\bx)-TV(\bz)| &\leq \sup_{a \in \sA} \, |c(x_1,a,\e[\,\cdot\,|\,\bx])-c(z_1,a,\e[\,\cdot\,|\,\bz])| \nonumber \\
&+\beta \, \sup_{a \in \sA} \bigg| \int_{\sX^N} V(\by) \, \, \prod_{i=2}^N p(dy_i|x_i,\pi^*(x_i),\e[\,\cdot\,|\,\bx]) \bigotimes p(dy_1|x_1,a,\e[\,\cdot\,|\,\bx])  \nonumber \\
&\phantom{xxxxx}- \int_{\sX^N} V(\by) \, \, \prod_{i=2}^N p(dy_i|z_i,\pi^*(z_i),\e[\,\cdot\,|\,\bz]) \bigotimes p(dy_1|z_1,a,\e[\,\cdot\,|\,\bz]) \bigg| \label{eq1}
\end{align}
Let us consider the second term in (\ref{eq1}) without $\beta$ and supremum. We can bound this term as follows
\begin{align}
&\bigg| \int_{\sX^N} V(\by) \, \, \prod_{i=2}^N p(dy_i|x_i,\pi^*(x_i),\e[\,\cdot\,|\,\bx]) \bigotimes p(dy_1|x_1,a,\e[\,\cdot\,|\,\bx])  \nonumber \\
&\phantom{xxxxxxxxx}- \int_{\sX^N} V(\by) \, \, \prod_{i=2}^N p(dy_i|z_i,\pi^*(z_i),\e[\,\cdot\,|\,\bz]) \bigotimes p(dy_1|z_1,a,\e[\,\cdot\,|\,\bz]) \bigg| \nonumber \\
&\leq \bigg| \int_{\sX^N} V(\by) \, \, \prod_{i=2}^N p(dy_i|x_i,\pi^*(x_i),\e[\,\cdot\,|\,\bx]) \bigotimes p(dy_1|x_1,a,\e[\,\cdot\,|\,\bx])  \nonumber \\
&\phantom{xxxxxxxxx}- \int_{\sX^N} V(\by) \, \, \prod_{i=2}^N p(dy_i|z_i,\pi^*(z_i),\e[\,\cdot\,|\,\bz]) \bigotimes p(dy_1|x_1,a,\e[\,\cdot\,|\,\bx]) \bigg| \label{eq2} \\
&+\bigg| \int_{\sX^N} V(\by) \, \, \prod_{i=2}^N p(dy_i|z_i,\pi^*(z_i),\e[\,\cdot\,|\,\bz]) \bigotimes p(dy_1|x_1,a,\e[\,\cdot\,|\,\bx])  \nonumber \\
&\phantom{xxxxxxxxx}- \int_{\sX^N} V(\by) \, \, \prod_{i=2}^N p(dy_i|z_i,\pi^*(z_i),\e[\,\cdot\,|\,\bz]) \bigotimes p(dy_1|z_1,a,\e[\,\cdot\,|\,\bz]) \bigg| \label{eq3}
\end{align}

Now we bound (\ref{eq2}) and (\ref{eq3}) as follows.

\paragraph*{Bounding (\ref{eq2}):}

Note that for any $\sigma \in \S_{\{2:N\}}$ and $y_1 \in \sX$, we have 
\small
\begin{align*}
&\int_{\sX^{N-1}} V(\by) \, \prod_{i=2}^N p(dy_i|x_i,\pi^*(x_i),\e[\,\cdot\,|\,\bx]) - \int_{\sX^{N-1}} V(\by) \, \prod_{i=2}^N p(dy_i|z_i,\pi^*(z_i),\e[\,\cdot\,|\,\bz])  \\
&= \int_{\sX^{N-1}} \hspace{-7pt} V(\by) \, \prod_{i=2}^N p(dy_i|x_i,\pi^*(x_i),\e[\,\cdot\,|\,\bx]) - \int_{\sX^{N-1}} \hspace{-7pt} V(\by) \, \prod_{i=2}^N p(dy_i|z_{\sigma(i)},\pi^*(z_{\sigma(i)}),\e[\,\cdot\,|\,z_1,\bz_{\sigma(2:N)}])  
\end{align*}
\normalsize
since $\e[\,\cdot\,|\,z_1,\bz_{\sigma(2:N)}] = \e[\,\cdot\,|\,\bz]$ and $V \in C_{\sym}(\sX^N)$. Moreover, for any $y_1 \in \sX$, by Lemma~\ref{result1}, we have 
\begin{align*}
|V(y_1,\by_{\{2:N\}})-V(y_1,\bz_{\{2:N\}})| &\leq K_2 \, W_1(\e[\,\cdot\,|\,\by_{\{2:N\}}],\e[\,\cdot\,|\,\bz_{\{2:N\}}]) \\
&\leq \frac{K_2}{N-1} \sum_{i=2}^N d_{\sX}(y_i,z_i)
\end{align*}
Using these two facts, we obtain the following bound for any $\sigma \in \S_{\{2:N\}}$
\small
\begin{align*}
&\bigg| \int_{\sX^N} V(\by) \, \, \prod_{i=2}^N p(dy_i|x_i,\pi^*(x_i),\e[\,\cdot\,|\,\bx]) \bigotimes p(dy_1|x_1,a,\e[\,\cdot\,|\,\bx])  \nonumber \\
&\phantom{xxxxxxxxx}- \int_{\sX^N} V(\by) \, \, \prod_{i=2}^N p(dy_i|z_{\sigma(i)},\pi^*(z_{\sigma(i)}),\e[\,\cdot\,|\,\bz]) \bigotimes p(dy_1|x_1,a,\e[\,\cdot\,|\,\bx]) \bigg| \\
&\leq \int_{\sX} \bigg| \int_{\sX^{N-1}} V(\by) \, \, \prod_{i=2}^N p(dy_i|x_i,\pi^*(x_i),\e[\,\cdot\,|\,\bx]) - \int_{\sX^{N-1}} V(\by) \, \, \prod_{i=2}^N p(dy_i|z_{\sigma(i)},\pi^*(z_{\sigma(i)}),\e[\,\cdot\,|\,\bz]) \bigg| \\
&\phantom{xxxxxxxxxxxxxxxxxxxxxxxxxxxxxxxxxxxxxxxxxxxxxxxx}p(dy_1|x_1,a,\e[\,\cdot\,|\,\bx])\\
&\leq \int_{\sX} K_2 \, W_1\left(\prod_{i=2}^N p(\,\cdot\,|x_i,\pi^*(x_i),\e[\,\cdot\,|\,\bx]),\prod_{i=2}^N p(\,\cdot\,|z_{\sigma(i)},\pi^*(z_{\sigma(i)}),\e[\,\cdot\,|\,\bz])\right) \, p(dy_1|x_1,a,\e[\,\cdot\,|\,\bx]) \\
&\leq \frac{K_2}{N-1} \sum_{i=2}^N W_1\left(p(\,\cdot\,|x_i,\pi^*(x_i),\e[\,\cdot\,|\,\bx]),p(\,\cdot\,|z_{\sigma(i)},\pi^*(z_{\sigma(i)}),\e[\,\cdot\,|\,\bz])\right) \,\, \text{(by Lemma~\ref{result2})} \\
&\leq \frac{K_2}{N-1} \sum_{i=2}^N \left\{ \K_1 d_{\sX}(x_i,z_{\sigma(i)}) + \K_2 \|\pi^*(x_i)-\pi^*(z_{\sigma(i)})\| + \K_3 W_1(\e[\,\cdot\,|\,\bx],\e[\,\cdot\,|\,\bz]) \right\} \\
&\leq \frac{K_2}{N-1} \sum_{i=2}^N \left\{ (\K_1+\K_2 \L^*) d_{\sX}(x_i,z_{\sigma(i)}) + \K_3 W_1(\e[\,\cdot\,|\,\bx],\e[\,\cdot\,|\,\bz]) \right\}  \,\, \text{(by Lemma~\ref{lemma1})}
\end{align*}
\normalsize
Since above inequality holds for any $\sigma \in \S_{\{2:N\}}$, we can obtain the following bound for (\ref{eq2})
\begin{align*}
&\bigg| \int_{\sX^N} V(\by) \, \, \prod_{i=2}^N p(dy_i|x_i,\pi^*(x_i),\e[\,\cdot\,|\,\bx]) \bigotimes p(dy_1|x_1,a,\e[\,\cdot\,|\,\bx])  \nonumber \\
&\phantom{xxxxxxxxx}- \int_{\sX^N} V(\by) \, \, \prod_{i=2}^N p(dy_i|z_i,\pi^*(z_i),\e[\,\cdot\,|\,\bz]) \bigotimes p(dy_1|x_1,a,\e[\,\cdot\,|\,\bx]) \bigg| \\
&\leq K_2 (\K_1+\K_2 \L^*) \inf_{\sigma \in \S_{\{2:N\}}} \frac{1}{N-1} \sum_{i=2}^N d_{\sX}(x_i,z_{\sigma(i)}) + K_2 \K_3 W_1(\e[\,\cdot\,|\,\bx],\e[\,\cdot\,|\,\bz]) \\
&=  K_2 (\K_1+\K_2 \L^*) \, W_1(\e[\,\cdot\,|\,\bx_{\{2:N\}}],\e[\,\cdot\,|\,\bz_{\{2:N\}}])+ K_2 \K_3 W_1(\e[\,\cdot\,|\,\bx],\e[\,\cdot\,|\,\bz]) \,\, \text{(by Lemma~\ref{result3})}\\
&\leq \left\{ K_2 (\K_1+\K_2 \L^*) + K_2 \K_3 \frac{N-1}{N} \right\} W_1(\e[\,\cdot\,|\,\bx_{\{2:N\}}],\e[\,\cdot\,|\,\bz_{\{2:N\}}]) + \frac{K_2\K_3}{N}  d_{\sX}(x_1,z_1) \\ &\text{(by Lemma~\ref{result4})} \\
&\leq  K_2 (\K_1+\K_2 \L^* +\K_3) \,  W_1(\e[\,\cdot\,|\,\bx_{\{2:N\}}],\e[\,\cdot\,|\,\bz_{\{2:N\}}]) + \frac{K_2\K_3}{N}  d_{\sX}(x_1,z_1)
\end{align*}

\paragraph*{Bounding (\ref{eq3}):}

Since $|V(x_1,\bx_{\{2:N\}})-V(z_1,\bx_{\{2:N\}})| \leq K_1 d_{\sX}(x_1,z_1)$ for all $\bx_{\{2:N\}} \in \sX^{N-1}$, we have 
\begin{align*}
&\bigg| \int_{\sX^N} V(\by) \, \, \prod_{i=2}^N p(dy_i|z_i,\pi^*(z_i),\e[\,\cdot\,|\,\bz]) \bigotimes p(dy_1|x_1,a,\e[\,\cdot\,|\,\bx])  \nonumber \\
&\phantom{xxxxxxxxx}- \int_{\sX^N} V(\by) \, \, \prod_{i=2}^N p(dy_i|z_i,\pi^*(z_i),\e[\,\cdot\,|\,\bz]) \bigotimes p(dy_1|z_1,a,\e[\,\cdot\,|\,\bz]) \bigg| \\
&\leq \int_{\sX^{N-1}} \bigg| \int_{\sX} V(\by) \, p(dy_1|x_1,a,\e[\,\cdot\,|\,\bx]) - \int_{\sX} V(\by) \, p(dy_1|z_1,a,\e[\,\cdot\,|\,\bz]) \bigg| \\
&\phantom{xxxxxxxxx}\prod_{i=2}^N p(dy_i|z_i,\pi^*(z_i),\e[\,\cdot\,|\,\bz]) \\
&\leq K_1 \, W_1\left(p(\cdot|x_1,a,\e[\,\cdot\,|\,\bx]),p(\cdot|z_1,a,\e[\,\cdot\,|\,\bz])\right) \\
&\leq K_1 \K_1 d_{\sX}(x_1,z_1) + K_1 \K_3 W_1(\e[\,\cdot\,|\,\bx],\e[\,\cdot\,|\,\bz]) \\
&\leq  \left(K_1 \K_1 +\frac{K_1 \K_3}{N}\right) d_{\sX}(x_1,z_1)  + K_1 \K_3 \frac{N-1}{N} W_1(\e[\,\cdot\,|\,\bx_{\{2:N\}}],\e[\,\cdot\,|\,\bz_{\{2:N\}}]) \\ 
&\text{(by Lemma~\ref{result4})}\\
&\leq  \left(K_1 \K_1 +\frac{K_1 \K_3}{N}\right) d_{\sX}(x_1,z_1)  + K_1 \K_3  W_1(\e[\,\cdot\,|\,\bx_{\{2:N\}}],\e[\,\cdot\,|\,\bz_{\{2:N\}}])
\end{align*}

Now, by combining the bound for (\ref{eq2}) and (\ref{eq3}), we obtain the following
\begin{align}\label{eq4}
&\bigg| \int_{\sX^N} V(\by) \, \, \prod_{i=2}^N p(dy_i|x_i,\pi^*(x_i),\e[\,\cdot\,|\,\bx]) \bigotimes p(dy_1|x_1,a,\e[\,\cdot\,|\,\bx])  \nonumber \\
&\phantom{xxxxxxxxx}- \int_{\sX^N} V(\by) \, \, \prod_{i=2}^N p(dy_i|z_i,\pi^*(z_i),\e[\,\cdot\,|\,\bz]) \bigotimes p(dy_1|z_1,a,\e[\,\cdot\,|\,\bz]) \bigg| \nonumber \\
&\leq \left\{K_1\K_1 +\frac{K_1\K_3+K_2 \K_3}{N} \right\} d_{\sX}(x_1,y_1) \nonumber \\
&\phantom{xxxxxxxxx}+ \left\{ K_2 (\K_1+\K_2 \L^* +\K_3) + K_1 \K_3  \right\} W_1(\e[\,\cdot\,|\,\bx_{\{2:N\}}],\e[\,\cdot\,|\,\bz_{\{2:N\}}]) 
\end{align}
Moreover, for any $a \in \sA$, we have 
\begin{align}\label{eq5}
&|c(x_1,a,\e[\,\cdot\,|\,\bx])-c(z_1,a,\e[\,\cdot\,|\,\bz])| \leq \L_1 d_{\sX}(x_1,z_1) + \L_3 W_1(\e[\,\cdot\,|\,\bx],\e[\,\cdot\,|\,\bz]) \nonumber \\
&\leq \left\{\L_1 + \frac{\L_3}{N}\right\} \, d_{\sX}(x_1,z_1) + \L_3 \frac{N-1}{N} W_1(\e[\,\cdot\,|\,\bx_{\{2:N\}}],\e[\,\cdot\,|\,\bz_{\{2:N\}}]) \,\, \text{(by Lemma~\ref{result4})} \nonumber \\ 
&\leq \left\{\L_1 + \frac{\L_3}{N}\right\} \, d_{\sX}(x_1,z_1) + \L_3  W_1(\e[\,\cdot\,|\,\bx_{\{2:N\}}],\e[\,\cdot\,|\,\bz_{\{2:N\}}])
\end{align}
By combining the bounds (\ref{eq4}) and (\ref{eq5}), for any $\bx,\bz \in \sX^N$, we have 
\begin{align}\label{eq6}
|TV(\bx)-TV(\bz)| 
&\leq \left\{\L_1+\frac{\L_3}{N} + \beta \left(K_1\K_1 +\frac{K_1\K_3+K_2 \K_3}{N}\right) \right\} d_{\sX}(x_1,z_1) \nonumber \\
&\hspace{-10pt}+ \left\{ \L_3 + \beta \left(K_2 (\K_1+\K_2 \L^* +\K_3) + K_1 \K_3 \right) \right\} \,  W_1(\e[\,\cdot\,|\,\bx_{\{2:N\}}],\e[\,\cdot\,|\,\bz_{\{2:N\}}])
\end{align}
For any $K_1,K_2>0$, we define 
\begin{align*}
f_1^N(K_1,K_2) &:= \L_1 + \frac{\L_3}{N} + \beta \left(\K_1 +\frac{\K_3}{N}\right) K_1 + \beta \frac{\K_3}{N} K_2 := \a_1^N + \b_1^N K_1 + \c_1^N K_2 \\
 f_2^N(K_1,K_2) &:= \L_3 +\beta \K_3 K_1 +\beta (\K_1+\K_2 \L^* +\K_3) K_2 := \a_2 + \b_2 K_1 + \c_2 K_2
\end{align*}
The bound (\ref{eq6}) implies that the optimality operator $T$ maps $(K_1,K_2)$-Lipschitz continuous function $V \in C_{\sym}(\sX^N)$ to a $(f_1^N(K_1,K_2),f_2^N(K_1,K_2))$-Lipschitz continuous function $TV \in C_{\sym}(\sX^N)$, where Lipschitz continuity notion that is adopted here is defined in (\ref{aux1}). Since $TV^*=V^*$, then $V^*$ should be $(\K_1^{*,N},\K_2^{*,N})$-Lipschitz continuous where 
\begin{align}\label{eq7}
\K_1^{*,N} = f_1^N(\K_1^{*,N},\K_2^{*,N}), \,\,\,\, \K_2^{*,N} = f_2^N(\K_1^{*,N},\K_2^{*,N})
\end{align}
Solving affine equations in (\ref{eq7}) lead to the following Lipschitz constants for $V^*$
$$
\K_1^{*,N}= \frac{\a_1^N (1-\c_2)}{(1-\b_1^N)(1-\c_2)-\c_1^N\a_2\b_2}, \,\,\, \K_2^{*,N}= \frac{\a_2 (1-\b_1^N)}{(1-\b_1^N)(1-\c_2)-\c_1^N\a_1^N\b_2}. 
$$
\end{proof}

Now using symmetry and Lipschitz continuity properties of the optimal value function $V^*$ of $\text{MDP}_{N}$, we prove the Lipschitz continuity of the best-response policy, or equivalently, the optimal policy of $\text{MDP}_{N}$.

\subsubsection{Lipschitz Continuity of Best-response Policy}

Now it is time to prove Lipschitz continuity of best-response policy to the MFE policy $\pi^*$. Recall that the optimal value function $V^*$ of  $\text{MDP}_{N}$ has the following properties
\begin{itemize}
\item[(P1)] $V^*:\sX^N\rightarrow [0,\infty)$ is symmetric in $(x_2,\ldots,x_N)$; that is, $V^* \in C_{\sym}(\sX^N)$.
\item[(P2)] $V^*$ is $(\K_1^{*,N},\K_2^{*,N})$-Lipschitz continuous; that is, for any $\bx,\bz \in \sX^N$, we have 
$$
|V^*(\bx)-V^*(\by)| \leq \K_1^{*,N} d_{\sX}(x_1,y_1) + \K_2^{*,N} W_1(\e[\,\cdot\,|\,\bx_{\{2:N\}}],\e[\,\cdot\,|\,\by_{\{2:N\}}])
$$
\end{itemize}
Note that (P2) and Lemma~\ref{result1} imply that for any $y \in \sX, \by_{\{2:N\}}, \bz_{\{2:N\}} \in \sX^{N-1}$, we have
\begin{align}\label{aux2}
|V^*(y,\by_{\{2:N\}}) -V^*(y,\bz_{\{2:N\}})| \leq \K_2^{*,N} d_{\av}(\by_{\{2:N\}},\bz_{\{2:N\}})
\end{align}
Recall also the following definition for any $M\geq1$: $\P_{M}(\sX):=\e[\,\cdot\,|\,\sX^M]$; that is, $\P_{M}(\sX)$ is the image of $\e$ into $\P(\sX)$. Now, we define the following set
$$
\sS^N := \left\{(\mu,x) \in \P_N(\sX)\times\sX: \frac{N}{N-1} \mu - \frac{1}{N-1} \delta_x \in \P_{N-1}(\sX) \right\}
$$
Note that for any policy $\pi \in \Pi$, initial state $\bx$, and $t\geq0$, $\text{MDP}_{N}$ satisfies the following
$$
(\e[\,\cdot\,|\,\bx(t)],x_1(t)) \in \sS^N 
$$
that is; $\sS^N$ is the reachability set of the pair  $(\e[\,\cdot\,|\,\bx(t)],x_1(t))$ for any $t\geq0$.

Recall the optimality equation for $\text{MDP}_{N}$
\small
\begin{align*}
V^*(\bx) = \min_{a \in \sA} \left[ c(x_1,a,\e[\,\cdot\,|\,\bx]) + \beta \int_{\sX^N} V^*(\by) \, \prod_{i=2}^N p(dy_i|x_i,\pi^*(x_i),\e[\,\cdot\,|\,\bx]) \bigotimes p(dy_1|x_1,a,\e[\,\cdot\,|\,\bx]) \right]
\end{align*}
\normalsize
It is known that the best-response policy, denoted as $\gamma_{\opt}^N:\sX^N\rightarrow\sA$, is the minimizer of the above optimality equation for all $\bx \in \sX^N$. We define $T^N:\sX\times\sS^N\rightarrow [0,\infty)$ as follows
\begin{align*}
T^N(y,\mu,x) := \int_{\sX^{N-1}} V^*(y,\by_{\{2:N\}}) \, \prod_{i=2}^N p(dy_i|x_i,\pi^*(x_i),\mu)
\end{align*}
where 
\small
$$(x_2,\ldots,x_N) \in \e^{-1}\left(\frac{N}{N-1}\mu-\frac{1}{N-1}\delta_x\right):=\left\{\bx_{\{2:N\}} \in \sX^{N-1}: \e[\,\cdot\,|\,\bx_{\{2:N\}}] = \frac{N}{N-1} \mu -\frac{1}{N-1}\delta_x\right\}
$$
\normalsize
Since $V^*$ is symmetric in $(x_2,\ldots,x_N)$, we have 
$$
\int_{\sX^{N-1}} V^*(y,\by_{\{2:N\}}) \, \prod_{i=2}^N p(dy_i|x_i,\pi^*(x_i),\mu) = \int_{\sX^{N-1}} V^*(y,\by_{\{2:N\}}) \, \prod_{i=2}^N p(dy_i|\hx_i,\pi^*(\hx_i),\mu)
$$for any 
$$\bx_{\{2:N\}},\hat{\bx}_{\{2:N\}} \in \e^{-1}\left(\frac{N}{N-1}\mu-\frac{1}{N-1}\delta_x\right)$$ 
as these two vectors are related via some permutation; that is, for some $\sigma \in \S_{\{2:N\}}$, we have $\bx_{\{2:N\}}=\hat{\bx}_{\sigma(2:N)}$. Hence the definition of $T^N$ does not depend on a particular choice of the vector $(x_2,\ldots,x_N) \in \e^{-1}\left(\frac{N}{N-1}\mu-\frac{1}{N-1}\delta_x\right)$. Therefore, $T^N$ is well-defined. In view of this definition, one can re-write the right side of the optimality equation as follows
\begin{align*}
\min_{a \in \sA} \left[ c(x,a,\mu) + \beta \int_{\sX} T^N(y,\mu,x) \, p(dy|x,a,\mu) \right]
\end{align*}
for any $(\mu,x) \in \sS^N$; that is, the right side of the optimality equation depends only on mean-field term and the state variable of agent~$1$. Since $\gamma_{\opt}^N$ is the minimizer of this re-written optimality equation, we can immediately conclude that $\gamma_{\opt}^N$ depends only on $(\mu,x) \in \sS^N$. Let us state this as a separate lemma. 

\begin{lemma}\label{lemma4}
There exists an optimal policy $\gamma_{\opt}^N$ of $\text{MDP}_{N}$ that is a function of the mean-field term $\e[\,\cdot\,|\,\bx(t)]$ and the state $x_1(t)$ of agent~$1$; that is, $\gamma_{\opt}^N:\sS^N\rightarrow\sA$.
\end{lemma}

To establish Lipschitz continuity of $\gamma_{\opt}^N$ via assumption (d), we need to first prove that $T^N$ is Lipschitz continuous. 

\begin{lemma}\label{lemmaa-new}
$T^N:\sX\times\sS^N\rightarrow[0,\infty)$ satisfies the following Lipschitz bound
$$
|T^N(y,\mu,x)-T^N(z,\nu,q)| \leq \R_1^N d_{\sX}(y,z) + \R_2^N W_1(\mu,\nu) + \R_3^N d_{\sX}(x,q)
$$
where 
$$
\R_1^N := \K_1^{*,N}, \,\, \R_2^N:=\left\{\K_2^{*,N} \K_3 + \K_2^{*,N}(\K_1+\K_2\L^*) \frac{N}{N-1}\right\}, \,\, \R_3^N:=\frac{\K_2^{*,N}(\K_1+\K_2\L^*)}{N-1}
$$
\end{lemma}

\begin{proof}
Pick any two triples $(y,\mu,x), (z,\nu,q) \in \sX\times\sS^N$. Let $$(x_2,\ldots,x_N) \in \e^{-1}\left(\frac{N}{N-1}\mu-\frac{1}{N-1}\delta_x\right)$$ and 
$(z_2,\ldots,z_N) \in \e^{-1}\left(\frac{N}{N-1}\nu-\frac{1}{N-1}\delta_q\right)$.
Then we have 
\small
\begin{align}\label{eq8}
&|T^N(y,\mu,x)-T^N(z,\nu,q)| \nonumber \\
&= \left|\int_{\sX^{N-1}} V^*(y,\by_{\{2:N\}}) \, \prod_{i=2}^N p(dy_i|x_i,\pi^*(x_i),\mu)-\int_{\sX^{N-1}} V^*(z,\by_{\{2:N\}}) \, \prod_{i=2}^N p(dy_i|z_i,\pi^*(z_i),\nu)\right| \nonumber \\
&\leq \left|\int_{\sX^{N-1}} V^*(y,\by_{\{2:N\}}) \, \prod_{i=2}^N p(dy_i|x_i,\pi^*(x_i),\mu)-\int_{\sX^{N-1}} V^*(z,\by_{\{2:N\}}) \, \prod_{i=2}^N p(dy_i|x_i,\pi^*(x_i),\mu)\right| \nonumber \\
&+\left|\int_{\sX^{N-1}} V^*(z,\by_{\{2:N\}}) \, \prod_{i=2}^N p(dy_i|x_i,\pi^*(x_i),\mu)-\int_{\sX^{N-1}} V^*(z,\by_{\{2:N\}}) \, \prod_{i=2}^N p(dy_i|z_i,\pi^*(z_i),\nu)\right| \nonumber\\
&\leq \K_1^{*,N} d_{\sX}(y,z) + \K_2^{*,N} W_1\left(\prod_{i=2}^N p(\cdot|x_i,\pi^*(x_i),\mu),\prod_{i=2}^N p(\cdot|z_i,\pi^*(z_i),\nu)\right) \,\, \text{(by (P2) and (\ref{aux2}))} \nonumber \\
&\leq \K_1^{*,N} d_{\sX}(y,z) + \frac{\K_2^{*,N}}{N-1} \sum_{i=2}^N W_1\left(p(\cdot|x_i,\pi^*(x_i),\mu),p(\cdot|z_i,\pi^*(z_i),\nu)\right) \,\, \text{(by Lemma~\ref{result2})} \nonumber \\
&\leq \K_1^{*,N} d_{\sX}(y,z) + \frac{\K_2^{*,N}}{N-1} \sum_{i=2}^N \left\{\K_1 d_{\sX}(x_i,z_i)+\K_2\L^* d_{\sX}(x_i,z_i)+ \K_3 W_1(\mu,\nu) \right\} \nonumber \\
&= \K_1^{*,N} d_{\sX}(y,z) + \K_2^{*,N} \K_3 W_1(\mu,\nu)+\frac{\K_2^{*,N}(\K_1+\K_2\L^*)}{N-1} \sum_{i=2}^N d_{\sX}(x_i,z_i)
\end{align}
\normalsize
The inequality (\ref{eq8}) is true for any $(x_2,\ldots,x_N) \in \e^{-1}(\frac{N}{N-1}\mu-\frac{1}{N-1}\delta_x)$ and $(z_2,\ldots,z_N) \in \e^{-1}(\frac{N}{N-1}\nu-\frac{1}{N-1}\delta_q)$. Hence, it is still true if we permute $(z_2,\ldots,z_N)$ using any $\sigma \in \S_{\{2:N\}}$. This implies that 
\small
\begin{align}
&|T^N(y,\mu,x)-T^N(z,\nu,q)| \nonumber \\
&\leq \K_1^{*,N} d_{\sX}(y,z) + \K_2^{*,N} \K_3 W_1(\mu,\nu)+\K_2^{*,N}(\K_1+\K_2\L^*) \inf_{\sigma \in \S_{\{2:N\}}} \frac{1}{N-1} \sum_{i=2}^N d_{\sX}(x_i,z_{\sigma(i)}) \nonumber \\
&= \K_1^{*,N} d_{\sX}(y,z) + \K_2^{*,N} \K_3 W_1(\mu,\nu) \nonumber \\
&+\K_2^{*,N}(\K_1+\K_2\L^*) \frac{N}{N-1} W_1(\e[\,\cdot\,|\,x,\bx_{\{2:N\}}],\e[\,\cdot\,|\,x,\bz_{\{2:N\}}]) \,\, \text{(by Lemma~\ref{result5})} \nonumber \\
&\leq \K_1^{*,N} d_{\sX}(y,z) +\K_2^{*,N} \K_3 W_1(\mu,\nu) \nonumber \\
&+\K_2^{*,N}(\K_1+\K_2\L^*) \frac{N}{N-1} \left\{ W_1(\e[\,\cdot\,|\,x,\bx_{\{2:N\}}],\e[\,\cdot\,|\,q,\bz_{\{2:N\}}]) +W_1(\e[\,\cdot\,|\,q,\bz_{\{2:N\}}],\e[\,\cdot\,|\,x,\bz_{\{2:N\}}])  \right\} \nonumber \\
&\leq \K_1^{*,N} d_{\sX}(y,z) +\K_2^{*,N} \K_3 W_1(\mu,\nu) +\K_2^{*,N}(\K_1+\K_2\L^*) \frac{N}{N-1} \left\{ W_1(\mu,\nu) + \frac{1}{N} d_{\sX}(x,q) \right\} \nonumber \\
&\text{(since $\e[\,\cdot\,|\,x,\bx_{\{2:N\}}] = \mu$ and $e[\,\cdot\,|\,q,\bz_{\{2:N\}}]=\nu$)} \nonumber \\
&= \K_1^{*,N} d_{\sX}(y,z) + \left\{\K_2^{*,N} \K_3 + \K_2^{*,N}(\K_1+\K_2\L^*) \frac{N}{N-1}\right\} W_1(\mu,\nu) +  \frac{\K_2^{*,N}(\K_1+\K_2\L^*)}{N-1} d_{\sX}(x,q) \nonumber 
\end{align}
\normalsize
This completes the proof. 
\end{proof}

Recall that $\sS^N$ is the set of all reachable points by $(\e[\,\cdot\,|\,\bx(t)],x_1(t))$  for any $t\geq0$ under any policy $\pi \in \Pi$. Since $\gamma_{\opt}^N$ depends only on $(\e[\,\cdot\,|\,\bx(t)],x_1(t))$, it is sufficient to define $\gamma_{\opt}^N$ only on $\sS^N$ for the $N$-agent game problem. However, for the approximation analysis in the sequel, we need to extend the definition of  $\gamma_{\opt}^N$ to the whole $\P(\sX) \times \sX$. To do this, we first extend the definition of $T^N:\sX\times\sS^N\rightarrow[0,\infty)$ to $\sX\times\P(\sX)\times\sX$ as follows
\begin{align*}
H^N(y,\mu,x) := \inf_{(z,\nu,r) \in \sX\times\sS^N} \left\{T^N(z,\nu,r) + \R_1^{N} d_{\sX}(y,z) + \R_2^N W_1(\mu,\nu) + \R_3^N d_{\sX}(x,r) \right\}
\end{align*}
One can prove that $H^N = T^N$ on $\sX\times\sS^N$ and $H^N$ is $(\R_1^N,\R_2^N,\R_3^N)$-Lipschitz continuous on $\sX\times\P(\sX)\times\sX$. These properties can be established easily, and so, we omit the details. Now, for any $(\mu,x) \in \P(\sX)\times\sX$, let us define the following policy
\begin{align*}
\gamma^{*,N}(\mu,x) := \argmin_{a \in \sA} \left[ c(x,a,\mu) + \beta \int_{\sX} H^N(y,\mu,x) \, p(dy|x,a,\mu) \right]
\end{align*}
Since $H^N = T^N$ on $\sX\times\sS^N$, we have $\gamma^{*,N}=\gamma_{\opt}^N$ on $\sS^N$. Hence, without loss of generality, we can take $\gamma^{*,N} \in \Pi_{\e}$ as the best response policy in place of $\gamma_{\opt}^N$ as they have the same behavior on the reachable set $\sS^N$, and so, have the same cost function in response to MFE policy $\pi^*$.

Now, we define 
$$
F_N(x,a,\mu) := c(x,a,\mu) + \beta \int_{\sX} H^N(y,\mu,x) \, p(dy|x,a,\mu)
$$
Since  $H^N$ is $(\R_1^N,\R_2^N,\R_3^N)$-Lipschitz continuous, by assumption (d), $F_N$ has the following properties
\begin{itemize}
\item[(F1)] For any $(x,\mu) \in \sX\times\P(\sX)$, $F_N(x,\cdot,\mu)$ is $\rho$-strongly convex in $a$.
\item[(F2)] For any $a \in \sA$, $\nabla_a F_N(x,a,\mu)$ is $\rL(\R_1^N,\R_2^N,\R_3^N)$-Lipschitz continuous in $(x,\mu)$. 
\end{itemize}
In view of these properties, we now show that best response policy $\gamma^{*,N}$
is Lipschitz continuous. 

\begin{lemma}\label{lemma6}
The best response policy $\gamma^{*,N}:\P(\sX)\times\sX\rightarrow\sA \in \Pi_{\e}$ is $(\L^{*,N}_1,\L^{*,N}_2)$-Lipschitz continuous where 
$$
(\L^{*,N}_1,\L^{*,N}_2):= \frac{2\rL(\R_1^N,\R_2^N,\R_3^N)}{\rho}=\left(\frac{2\rL_1(\R_1^N,\R_2^N,\R_3^N)}{\rho},\frac{2\rL_2(\R_1^N,\R_2^N,\R_3^N)}{\rho}\right)
$$
\end{lemma}

\begin{proof}
The proof is very similar to the proof of Lemma~\ref{lemma1}. Fix any
$(\mu,x),(\nu,z) \in \P(\sX)\times\sX$. Then, by strong convexity property (F1), we have 
\begin{align*}
0 &\geq F_N(x,\gamma^{*,N}(x,\mu),\mu) - F_N(x,\gamma^{*,N}(z,\nu),\mu) \\
&\geq \langle \nabla_a F_N(x,\gamma^{*,N}(z,\nu),\mu), \gamma^{*,N}(x,\mu)-\gamma^{*,N}(z,\nu)\rangle +\frac{\rho}{2} \|\gamma^{*,N}(z,\nu)-\gamma^{*,N}(x,\mu)\|^2
\end{align*}
Hence
$$
\|\gamma^{*,N}(z,\nu)-\gamma^{*,N}(x,\mu)\|^2 \leq \frac{2}{\rho} \langle -\nabla_a F_N(x,\gamma^{*,N}(z,\nu),\mu), \gamma^{*,N}(x,\mu)-\gamma^{*,N}(z,\nu)\rangle
$$
By first order optimality condition we also have 
$$
\langle \nabla_a F_N(z,\gamma^{*,N}(z,\nu),\nu), a-\gamma^{*,N}(z,\nu)\rangle \geq 0 
$$
for all $a \in \sA$. Therefore, we can write the following 
\begin{align*}
&\|\gamma^{*,N}(z,\nu)-\gamma^{*,N}(x,\mu)\|^2 \nonumber \\
&\leq  \frac{2}{\rho} \langle  \nabla_a F_N(z,\gamma^{*,N}(z,\nu),\nu)-\nabla_a F_N(x,\gamma^{*,N}(z,\nu),\mu), \gamma^{*,N}(x,\mu)-\gamma^{*,N}(z,\nu)\rangle \\
&\leq \frac{2}{\rho} \|\nabla_a F_N(z,\gamma^{*,N}(z,\nu),\nu)-\nabla_a F_N(x,\gamma^{*,N}(z,\nu),\mu)\| \|\gamma^{*,N}(z,\nu)-\gamma^{*,N}(x,\mu)\| \\
&\leq \left\{\frac{2\rL_1(\R_1^N,\R_2^N,\R_3^N)}{\rho} d_{\sX}(x,z) + \frac{2\rL_2(\R_1^N,\R_2^N,\R_3^N)}{\rho} W_1(\mu,\nu) \right\} \, \|\gamma^{*,N}(z,\nu)-\gamma^{*,N}(x,\mu)\| \\
&\text{(by property (F2))}
\end{align*}
This implies that 
\begin{align*}
\|\gamma^{*,N}(z,\nu)-\gamma^{*,N}(x,\mu)\| &\leq \left\{\frac{2\rL_1(\R_1^N,\R_2^N,\R_3^N)}{\rho} d_{\sX}(x,z) + \frac{2\rL_2(\R_1^N,\R_2^N,\R_3^N)}{\rho} W_1(\mu,\nu) \right\} \\
&=: \L^{*,N}_1 d_{\sX}(x,z) + \L^{*,N}_2 W_1(\mu,\nu)
\end{align*}
for all $(\mu,x),(\nu,z) \in \P(\sX)\times\sX$, which completes the proof.  
\end{proof}

Until now, we proved that MFE policy $\pi^*:\sX\rightarrow\sA$ is $\L^*$-Lipschitz continuous and the corresponding best-response policy to MFE policy $\gamma^{*,N}:\P(\sX)\times\sX\rightarrow\sA$ is $(\L_1^{*,N},\L_2^{*,N})$-Lipschitz continuous for any $N$ that satisfies assumption (f). Now, using these results, we establish that the joint policy $(\pi^*,\ldots,\pi^*)$ is an approximate Nash equilibrium for the $N$-agent linear mean-field game.

\section{Approximate Equilibrium in $N$-agent Linear MFGs}\label{sec3}

In this section we prove that the joint policy $(\pi^*,\ldots,\pi^*)$ is approximately Nash equilibrium for the finite agent games. To this end, we prove a series of results.  

In the remainder of this section, $\{(\bx(t),\ba(t))\}_{t\geq0}$ denotes the state-action vectors under the joint policy $\bpi^*:=(\pi^*,\ldots,\pi^*)$ and initial distribution $\mu^{*,\otimes^N}$ for $N$-agent linear mean-field game. Similarly, $\{(x(t),a(t))\}_{t\geq0}$ denotes the state-action pairs under the policy $\pi^*$ and initial distribution $\mu^*$ of a generic agent in the infinite population limit.

For $t=0$, $(x_1(0),\ldots,x_N(0)) \sim \mu^{*,\otimes^N}$, and therefore, there exists, by Glivenko-Cantelli's theorem, a function $\alpha(N)$ such that 
\begin{align} \label{LLN}
\rE[W_1(\e[\,\cdot\,|\,\bx(0)],\mu^*)] \leq \alpha(N)
\end{align}
and $\alpha(N) \rightarrow 0$ as $N\rightarrow\infty$. Indeed, if $\sX$ is a compact subset of some finite dimensional Euclidean space, then we have an explicit expression for $\alpha(N)$ which is due to \cite[Theorem 1]{FoGu15}: for any $q  > 1$, there exists a constant $C$ depending on $q$ and $\dim(\sX)$ such that 
\begin{align*}
\alpha(N):= C \, \left(\int_{\sX} \|x\|^q \, \mu^*(dx) \right)^{\frac{1}{q}} \, \begin{cases} 
\frac{1}{\sqrt{N}} + \frac{\sqrt[q]{N}}{N} & \text{if} \,\, \dim(\sA) < 2 \\
\frac{\log(1+N)}{\sqrt{N}} + \frac{\sqrt[q]{N}}{N} & \text{if} \,\, \dim(\sA) = 2 \\
\frac{1}{\sqrt[d]{N}} + \frac{\sqrt[q]{N}}{N} & \text{if} \,\, \dim(\sA) > 2 
\end{cases}
\end{align*}

The first result is about the convergence of the mean-field term $\e[\,\cdot\,|\,\bx(t)]$ to $\mu^*$ under mean-field equilibrium policy $(\pi^*,\ldots,\pi^*)$.

\begin{lemma}\label{lemma7}
Given any $\varepsilon>0$, for all $t\geq0$, we have
\begin{align*}
\rE\left[\sup_{g \in \F} \left| \int_{\sX} g(x) \, \e[dx|\,\bx(t)] - \int_{\sX} g(x) \, \mu^*(dx) \right| \right] = \rE[W_1(\e[\,\cdot\,|\,\bx(t)],\mu^*)] \leq \alpha_t(N,\varepsilon)
\end{align*}
where error bounds $\{\alpha_t(N,\varepsilon)\}_{t\geq0}$ are recursively defined as follows
\begin{align*}
\alpha_0(N,\varepsilon) &:= \alpha(N) \\
\alpha_{t+1}(N,\varepsilon) &:= (\K_1 + \K_2 \, \L^* + \K_3) \, \alpha_t(N,\varepsilon) + \sqrt{\frac{2}{N}} \diam(\sX) \, \N(\varepsilon,\F) + 2\varepsilon \\
&=: \kappa_1 \, \alpha_t(N,\varepsilon) + \kappa_2 \, \sqrt{\frac{2}{N}} + 2\varepsilon
\end{align*}
Therefore, for any $t\geq1$, we have 
$$
\alpha_t(N,\varepsilon) = \kappa_1^t \, \alpha(N) + \kappa_2 \, \sqrt{\frac{2}{N}} \sum_{i=0}^{t-1} \kappa_1^i  + 2\varepsilon\sum_{i=0}^{t-1} \kappa_1^i 
$$
\end{lemma}

\begin{proof}
Since $(x_1(0),\ldots,x_N(0)) \sim \mu^{*,\otimes^N}$, the statement is true for $t=0$. Suppose that it is true for some $t\geq0$ and consider $t+1$. Fix any $g \in \F$. Then we can write 
\begin{align}
&\left| \int_{\sX} g(x) \, \e[dx|\,\bx(t+1)] - \int_{\sX} g(x) \, \mu^*(dx) \right| \nonumber \\
&\leq \left| \int_{\sX} g(x) \, \e[dx|\,\bx(t+1)] - \int_{\sX\times\sX} g(x) \, p(dx|z,\pi^*(z),\e[\,\cdot\,|\,\bx(t)]) \, \e[dz|\,\bx(t)] \right| \label{eqq1} \\
&+ \left| \int_{\sX\times\sX} g(x) \, p(dx|z,\pi^*(z),\e[\,\cdot\,|\,\bx(t)]) \, \e[dz|\,\bx(t)] - \int_{\sX\times\sX} g(x) \, p(dx|z,\pi^*(z),\mu^*) \, \e[dz|\,\bx(t)] \right| \label{eqq2} \\
&+\left| \int_{\sX\times\sX} g(x) \, p(dx|z,\pi^*(z),\mu^*) \, \e[dz|\,\bx(t)] - \int_{\sX\times\sX} g(x) \, p(dx|z,\pi^*(z),\mu^*) \, \mu^*(dz)] \right| \label{eqq3}
\end{align}
Here, (\ref{eqq3}) is true since
$$
\mu^*(\,\cdot\,) = \int_{\sX} p(\,\cdot\,|z,\pi^*(z),\mu^*) \, \mu^*(dz)
$$
Now, let us bound the expectations (uniform in $g \in \F$) of the terms (\ref{eqq1}), (\ref{eqq2}), and (\ref{eqq3}). 

\paragraph*{Bounding (\ref{eqq2}):}

We have 
\small
\begin{align}
&\rE \left[ \sup_{g \in \F} \left| \int_{\sX\times\sX} g(x) \, p(dx|z,\pi^*(z),\e[\,\cdot\,|\,\bx(t)]) \, \e[dz|\,\bx(t)] - \int_{\sX\times\sX} g(x) \, p(dx|z,\pi^*(z),\mu^*) \, \e[dz|\,\bx(t)] \right| \right] \nonumber \\
&= \rE \bigg[\bigg| \int_{\sX\times\sX\times\sX} g(x) \, \bp(dx|z,\pi^*(z),y) \, \e[dz|\,\bx(t)]) \, \e[dy|\,\bx(t)] \nonumber \\
&\phantom{xxxxxxxxxxxx}- \int_{\sX\times\sX\times\sX} g(x) \, \bp(dx|z,\pi^*(z),y) \, \e[dz|\,\bx(t)] \, \mu^*(dy)\bigg| \bigg]\label{eqq2-1}
\end{align}
\normalsize
Define $l(y):=\int_{\sX\times\sX} g(x) \, \bp(dx|z,\pi^*(z),y) \, \e[dz|\,\bx(t)]$. Then $l(\sX) \subset [0,K]$ and for any $y,r \in \sX$, we have $|l(y)-l(r)| \leq \K_3 \, d_{\sX}(y,r)$. Hence, $\frac{l}{\K_3} \in \F$. 
This implies the following bound on (\ref{eqq2-1})
\begin{align}
(\ref{eqq2-1}) &\leq \K_3 \, \rE\left[ \sup_{g \in \F} \left|\int_{\sX} g(y) \, \e[dy|\,\bx(t)] - \int_{\sX} g(y) \, \mu^*(dy) \right| \right] \nonumber \\
&\leq \K_3 \, \alpha_t(N,\varepsilon) \,\, \text{(by induction hypothesis)} \label{eqq2-final}
\end{align}

\paragraph*{Bounding (\ref{eqq3}):}

We define $q(z) := \int_{\sX} g(x) \, p(dx|z,\pi^*(z),\mu^*)$. Note that $q(\sX) \subset [0,K]$ and for any $(z,y) \in \sX$, we have 
\begin{align*}
|q(z)-q(y)| &= \left|\int_{\sX} g(x) \, p(dx|z,\pi^*(z),\mu^*)-\int_{\sX} g(x) \, p(dx|y,\pi^*(y),\mu^*)\right| \\
&\leq \K_1 d_{\sX}(z,y) + \K_2 \, \L^* \, d_{\sX}(z,y) = (\K_1+\K_2 \, \L^*)\, d_{\sX}(z,y)
\end{align*}
Hence, 
$$
\frac{q}{(\K_1+\K_2 \, \L^*)} \in \F
$$
which implies that 
\begin{align}
&\rE \left[ \sup_{g \in \F}\left| \int_{\sX\times\sX} g(x) \, p(dx|z,\pi^*(z),\mu^*) \, \e[dz|\,\bx(t)] - \int_{\sX\times\sX} g(x) \, p(dx|z,\pi^*(z),\mu^*) \, \mu^*(dz)] \right| \right] \nonumber \\
&\leq (\K_1+\K_2 \, \L^*) \, \rE\left[ \sup_{g \in \F} \left|\int_{\sX} g(z) \, \e[dz|\,\bx(t)] - \int_{\sX} g(z) \, \mu^*(dz) \right| \right] \nonumber \\
&\leq (\K_1+\K_2 \, \L^*) \, \alpha_t(N,\varepsilon) \,\, \text{(by induction hypothesis)}  \label{eqq3-final}
\end{align}

\paragraph*{Bounding (\ref{eqq1}):}

Note that for any $g \in \F$, we can bound the expectation of (\ref{eqq1}) as follows
\footnotesize
\begin{align}
&\rE\left[\left| \int_{\sX} g(x) \, \e[dx|\,\bx(t+1)] - \int_{\sX\times\sX} g(x) \, p(dx|z,\pi^*(z),\e[\,\cdot\,|\,\bx(t)]) \, \e[dz|\,\bx(t)] \right|\right] \nonumber \\
&= \rE\left[\rE\left[\left| \int_{\sX} g(x) \, \e[dx|\,\bx(t+1)] - \int_{\sX\times\sX} g(x) \, p(dx|z,\pi^*(z),\e[\,\cdot\,|\,\bx(t)]) \, \e[dz|\,\bx(t)] \right| \bigg| \bx(t) \right]\right] \nonumber \\ 
&\leq \rE\left[\sqrt{\frac{1}{N^2} \sum_{i=1}^N \left\{\int_{\sX} g^2(x) \, p(dx|x_i(t),\pi^*(x_i(t)),\e[\,\cdot\,|\,\bx(t)]) - \left(\int_{\sX} g(x) \, p(dx|x_i(t),\pi^*(x_i(t)),\e[\,\cdot\,|\,\bx(t)])\right)^2\right\}}\right] \label{eqq1-1} \\
&\text{(by Lemma~\ref{result6})} \nonumber
\end{align}
\normalsize
Note that $\sup_{x \in \sX} g(x) = \sup_{x \in \sX} g(x) - g(y_*) \leq \sup_{x \in \sX} d_{\sX}(x,y_*) \leq \diam(\sX)$. Moreover, we also have $\sup_{x \in \sX} g^2(x) = \sup_{x \in \sX} g^2(x) - g^2(y_*) =  \sup_{x \in \sX} (g(x) - g(y_*)) (g(x) + g(y_*)) = \sup_{x \in \sX} (g(x) - g(y_*)) \, g(x) \leq \sup_{x \in \sX} d_{\sX}(x,y_*) \, g(x) \leq \diam(\sX)^2$. Therefore, we can bound the last term as follows
\begin{align}
(\ref{eqq1-1}) &\leq \sqrt{\frac{2}{N}} \, \diam(\sX)\label{eqq1-2}
\end{align}
Note that the bound (\ref{eqq1-2}) is independent of $g$. Let $\{g_1,\ldots,g_{\N(\varepsilon,F)}\}$ be minimal $\varepsilon$-cover of $\F$. Then we have
\small 
\begin{align}
&\rE\left[\sup_{g \in \F} \left| \int_{\sX} g(x) \, \e[dx|\,\bx(t+1)] - \int_{\sX\times\sX} g(x) \, p(dx|z,\pi^*(z),\e[\,\cdot\,|\,\bx(t)]) \, \e[dz|\,\bx(t)] \right|\right] \nonumber \\
&\leq \rE\bigg[\sup_{g \in \F} \inf_{i=1,\ldots,\N(\varepsilon,\F)} \bigg( 2 \|g-g_i\|_{\infty} + \bigg| \int_{\sX} g_i(x) \, \e[dx|\,\bx(t+1)]  \nonumber \\
&\phantom{xxxxxxxxxxxxxx}- \int_{\sX\times\sX} g_i(x) \, p(dx|z,\pi^*(z),\e[\,\cdot\,|\,\bx(t)]) \, \e[dz|\,\bx(t)] \bigg| \bigg) \bigg] \nonumber \\
&\leq 2\varepsilon + \rE\left[\sup_{i=1,\ldots,\N(\varepsilon,\F)} \left| \int_{\sX} g_i(x) \, \e[dx|\,\bx(t+1)] - \int_{\sX\times\sX} g_i(x) \, p(dx|z,\pi^*(z),\e[\,\cdot\,|\,\bx(t)]) \, \e[dz|\,\bx(t)] \right| \right] \nonumber \\
&\leq 2\varepsilon + \sum_{i=1}^{\N(\varepsilon,\F)} \rE\left[\left| \int_{\sX} g_i(x) \, \e[dx|\,\bx(t+1)] - \int_{\sX\times\sX} g_i(x) \, p(dx|z,\pi^*(z),\e[\,\cdot\,|\,\bx(t)]) \, \e[dz|\,\bx(t)] \right| \right]  \nonumber \\
&\leq 2\varepsilon + \sqrt{\frac{2}{N}} \, \diam(\sX) \, \N(\varepsilon,\F) \,\, \text{(by (\ref{eqq1-2}))}\label{eqq1-final}
\end{align}
\normalsize
Now if we combine the bounds (\ref{eqq1-final}), (\ref{eqq2-final}), and (\ref{eqq3-final}) for the expectations (uniform in $g \in \F$) of the terms (\ref{eqq1}), (\ref{eqq2}), and (\ref{eqq3}), we obtain the following
\begin{align*}
&\rE\left[\sup_{g \in \F} \left| \int_{\sX} g(x) \, \e[dx|\,\bx(t+1)] - \int_{\sX} g(x) \, \mu^*(dx) \right| \right] \\
&\leq (\K_1 + \K_2 \, \L^* + \K_3) \, \alpha_t(N,\varepsilon) + \sqrt{\frac{2}{N}} \diam(\sX) \, \N(\varepsilon,\F) + 2\varepsilon
\end{align*}
This completes the proof.
\end{proof}

Now it is time to prove our first important result. Here, using above lemma and Lipschitz continuity properties of $\bc$ and $\pi^*$, we now deduce that the cost of agent~$1$ under joint policy $\bpi^*$ in the $N$-agent game, where $N$ is large, should be close to the cost of a generic agent under optimal policy $\pi^*$ in the infinite population limit.  

\begin{theorem}\label{theorem1}
Given any $\varepsilon>0$, we have 
$$
|J_1(\mu^*;\bpi^*)-J(\mu^*;\pi^*)| := |J_1(\mu^*;\bpi^*)- \inf_{\pi \in \Pi_l}J(\mu^*;\pi)| \leq \Theta_1(N,\varepsilon)
$$
where 
$$
\Theta_1(N,\varepsilon) :=  (\L_1+\L_2 \, \L^*+\L_3) \, \left\{ \frac{\alpha(N) }{1-\beta \, \kappa_1} + \beta\, \left(\kappa_2 \, \sqrt{\frac{2}{N}}+2\varepsilon\right) \, \frac{1}{(1-\beta \, \kappa_1)(1-\beta)}\right\}
$$
\end{theorem}

\begin{proof}
Note that since every agent applies the same policy, one can prove that for any $t\geq0$, we have 
$$
{\cal L}(x_1(t),\ldots,x_N(t),\e[\,\cdot\,|\,\bx(t)]) = {\cal L}(x_{\sigma(1)}(t),\ldots,x_{\sigma(N)}(t),\e[\,\cdot\,|\,\bx(t)])
$$
for any permutation $\sigma \in \S_{\{1:N\}}$. Hence, for any $i=1,\ldots,N$, we have 
$$
\rE[\bc(x_1(t),a_1(t),z_1(t))] = \rE[c(x_1(t),a_1(t),\e[\,\cdot\,|\,\bx(t)])] = \rE[c(x_i(t),a_i(t),\e[\,\cdot\,|\,\bx(t)])]
$$
Therefore 
\begin{align*}
\rE[\bc(x_1(t),a_1(t),z_1(t))] &= \frac{1}{N} \sum_{i=1}^N \rE[c(x_i(t),a_i(t),\e[\,\cdot\,|\,\bx(t)])]  \\
&= \rE\left[\frac{1}{N} \sum_{i=1}^N c(x_i(t),a_i(t),\e[\,\cdot\,|\,\bx(t)])\right]  \\
&= \rE \left[ \int_{\sX} c(x,\pi^*(x),\e[\,\cdot\,|\,\bx(t)]) \, \e[dx|\,\bx(t)]  \right] \\
&= \rE \left[ \int_{\sX\times\sX} \bc(x,\pi^*(x),z) \, \e[dx|\,\bx(t)] \, \e[dz|\,\bx(t)] \right]
\end{align*}
In view of the last identity, let us obtain a bound on the following term
\begin{align}
&\left| \rE[\bc(x_1(t),a_1(t),z_1(t))] - \rE[\bc(x(t),a(t),z(t))] \right| \nonumber \\
&= \left| \rE\left[ \int_{\sX\times\sX} \bc(x,\pi^*(x),z) \, \e[dx|\,\bx(t)] \, \e[dz|\,\bx(t)] \right] - \int_{\sX\times\sX} \bc(x,\pi^*(x),z) \, \mu^*(dx) \, \mu^*(dz) \right| \nonumber \\ 
&\text{(as ${\cal L}(x(t))=\mu^*$)} \nonumber \\
&\leq \left| \rE\left[ \int_{\sX\times\sX} \hspace{-7pt} \bc(x,\pi^*(x),z) \, \e[dx|\,\bx(t)] \, \e[dz|\,\bx(t)] \right] - \rE\left[\int_{\sX\times\sX} \hspace{-7pt} \bc(x,\pi^*(x),z) \, \e[dx|\,\bx(t)] \, \mu^*(dz) \right] \right| \nonumber \\
&+ \left| \rE\left[ \int_{\sX\times\sX} \bc(x,\pi^*(x),z) \, \e[dx|\,\bx(t)] \, \mu^*(dz) \right] - \int_{\sX\times\sX} \bc(x,\pi^*(x),z) \, \mu^*(dx) \, \mu^*(dz) \right| \label{thm1-1}
\end{align}
Now define 
\begin{align*}
l_1(z) &:= \int_{\sX} \bc(x,\pi^*(x),z) \, \e[dx|\,\bx(t)] \\
l_2(x) &:= \int_{\sX} \bc(x,\pi^*(x),z) \, \mu^*(dz)
\end{align*}
Note that $l_1(\sX),l_1(\sX) \subset [0,\bc_{\max}]$ and moreover for any $z,x \in \sX$, we have $|l_1(z)-l_1(x)| \leq \L_3 \, d_{\sX}(z,x)$ and $|l_2(x)-l_2(z)| \leq (\L_1+\L_2 \, \L^*) \, d_{\sX}(x,z)$. Hence 
$$
\frac{l_1}{\L_3}, \, \frac{l_2}{\L_1+\L_2 \, \L^*} \in \F
$$
Therefore, by Lemma~\ref{lemma7}, we have
$$
(\ref{thm1-1}) \leq (\L_1+\L_2 \, \L^*+\L_3) \, \alpha_t(N,\varepsilon)
$$
This implies that 
\begin{align}
&|J_1(\mu^*;\bpi^*)-J(\mu^*;\pi^*)|  \leq \sum_{t=0}^{\infty} \beta^t \, \left| \rE[\bc(x_1(t),a_1(t),z_1(t))] - \rE[\bc(x(t),a(t),z(t))] \right| \nonumber \\
&\leq  \sum_{t=0}^{\infty} \beta^t \, (\L_1+\L_2 \, \L^*+\L_3) \, \alpha_t(N,\varepsilon) \nonumber \\
&= (\L_1+\L_2 \, \L^*+\L_3) \sum_{t=0}^{\infty} \beta^t \left\{\kappa_1^t \, \alpha(N) + \kappa_2 \, \sqrt{\frac{2}{N}} \sum_{i=0}^{t-1} \kappa_1^i  + 2\varepsilon\sum_{i=0}^{t-1} \kappa_1^i \right\} \nonumber \\
&= (\L_1+\L_2 \, \L^*+\L_3) \left\{ \frac{\alpha(N)}{1-\beta \, \kappa_1} + \beta\, \left(\kappa_2 \, \sqrt{\frac{2}{N}}+2\varepsilon\right) \, \sum_{t=0}^{\infty} \beta^t \sum_{i=0}^{t-1} \kappa_1^i \right\}  \nonumber \\
&\text{(by assumption (e), $\beta \, \kappa_1 < 1$)} \nonumber \\
&= (\L_1+\L_2 \, \L^*+\L_3) \left\{ \frac{\alpha(N)}{1-\beta \, \kappa_1} + \beta\, \left(\kappa_2 \, \sqrt{\frac{2}{N}}+2\varepsilon\right) \, \sum_{i=0}^{\infty} \kappa_1^i \sum_{t=i}^{\infty} \beta^t  \right\}  \nonumber \\
&= (\L_1+\L_2 \, \L^*+\L_3) \left\{ \frac{\alpha(N)}{1-\beta \, \kappa_1} + \beta\, \left(\kappa_2 \, \sqrt{\frac{2}{N}}+2\varepsilon\right) \, \sum_{i=0}^{\infty} \kappa_1^i \, \beta^i \, \sum_{t=0}^{\infty} \beta^t  \right\}  \nonumber \\
&= (\L_1+\L_2 \, \L^*+\L_3) \left\{ \frac{\alpha(N)}{1-\beta \, \kappa_1} + \beta\, \left(\kappa_2 \, \sqrt{\frac{2}{N}}+2\varepsilon\right) \, \frac{1}{(1-\beta \, \kappa_1)(1-\beta)}\right\}  \nonumber
\end{align}
This completes the proof.
\end{proof}

In the remainder of this section, $\{(\tbx(t),\tba(t))\}_{t\geq0}$ denotes the state-action vectors under the joint policy $(\gamma^{*,N},\pi^*,\ldots,\pi^*)$ and initial distribution $\mu^{*,\otimes^N}$ for the $N$-agent game. Similarly, $\{(\hx(t),\ha(t))\}_{t\geq0}$ denotes the state-action pairs under the policy $\gamma^{*,N}$ and initial distribution $\mu^*$ of a generic agent in the infinite population limit. Note that in the infinite population limit, a generic agent applies the best-response policy $\gamma^{*,N}:\P(\sX)\times\sX\rightarrow\sA$ as follows: $\ha(t) = \gamma^{*,N}(\mu^*,\hx(t))$ for all $t\geq0$.

\begin{lemma}\label{lemma8}
Given  any $\varepsilon>0$, for all $t\geq0$, we have
\begin{align*}
\rE\left[\sup_{g \in \F} \left| \int_{\sX} g(x) \, \e[dx|\,\tbx(t)] - \int_{\sX} g(x) \, \mu^*(dx) \right| \right] = \rE[W_1(\e[\,\cdot\,|\,\tbx(t)],\mu^*)] \leq \talpha_t(N,\varepsilon)
\end{align*}
where error bounds $\{\talpha_t(N,\varepsilon)\}_{t\geq0}$ are recursively defined as follows
\begin{align*}
\talpha_0(N,\varepsilon) &:= \alpha(N) \\
\talpha_{t+1}(N,\varepsilon) &:= (\K_1 + \K_2 \, \L^* + \K_3) \, \talpha_t(N,\varepsilon) + \sqrt{\frac{2}{N-1}} \diam(\sX) \left\{ \N(\varepsilon,\F) + \sqrt{\frac{8}{N-1}}\right\}+ 2\varepsilon \\
&=: \kappa_1 \, \talpha_t(N,\varepsilon) + \kappa_2^N \, \sqrt{\frac{2}{N-1}} + 2\varepsilon
\end{align*}
Therefore, for any $t\geq1$, we have 
$$
\talpha_t(N,\varepsilon) = \kappa_1^t \, \alpha(N) + \kappa_2^N \, \sqrt{\frac{2}{N-1}} \sum_{i=0}^{t-1} \kappa_1^i  + 2\varepsilon\sum_{i=0}^{t-1} \kappa_1^i 
$$
\end{lemma}

\begin{proof}
Since $(\tx_1(0),\ldots,\tx_N(0)) \sim \mu^{*,\otimes^N}$, the statement is true for $t=0$. Suppose that it is true for some $t\geq0$ and consider $t+1$. Fix any $g \in \F$. Then we can write 
\begin{align}
&\left| \int_{\sX} g(x) \, \e[dx|\,\tbx(t+1)] - \int_{\sX} g(x) \, \mu^*(dx) \right| \nonumber \\
&\leq \left| \int_{\sX} g(x) \, \e[dx|\,\tbx(t+1)] - \int_{\sX} g(x) \e[dx|\,\tbx_{\{2:N\}}(t+1)] \right| \label{n-eqq1} \\
&+\left| \int_{\sX} g(x) \e[dx|\,\tbx_{\{2:N\}}(t+1)] - \int_{\sX\times\sX} g(x) \, p(dx|z,\pi^*(z),\e[\,\cdot\,|\,\tbx(t)]) \, \e[dz|\,\tbx_{\{2:N\}}(t)] \right| \label{n-eqq2} \\
&+ \bigg| \int_{\sX\times\sX} g(x) \, p(dx|z,\pi^*(z),\e[\,\cdot\,|\,\tbx(t)]) \, \e[dz|\,\tbx_{\{2:N\}}(t)] \label{n-eqq3} \\
&\phantom{xxxxxxxxxxxxxxxxxx}- \int_{\sX\times\sX} g(x) \, p(dx|z,\pi^*(z),\e[\,\cdot\,|\,\tbx(t)]) \, \e[dz|\,\tbx(t)] \bigg| \nonumber \\
&+\left| \int_{\sX\times\sX} g(x) \, p(dx|z,\pi^*(z),\e[\,\cdot\,|\,\tbx(t)]) \, \e[dz|\,\tbx(t)] - \int_{\sX\times\sX} g(x) \, p(dx|z,\pi^*(z),\mu^*) \, \mu^*(dz)] \right| \label{n-eqq4} 
\end{align}
Now, let us bound the expectations (uniform in $g \in \F$) of the terms (\ref{n-eqq1}), (\ref{n-eqq2}), (\ref{n-eqq3}), and (\ref{n-eqq4}). 

\paragraph*{Bounding (\ref{n-eqq1}):}

Note that we can bound (\ref{n-eqq1}) as follows
\footnotesize
$$
(\ref{n-eqq1}) \leq \max \left\{ \left| \frac{1}{N} \sum_{i=1}^N g(\tx_i(t+1)) -  \frac{1}{N} \sum_{i=2}^N g(\tx_i(t+1)) \right|, \, \, \left| \frac{1}{N} \sum_{i=1}^N g(\tx_i(t+1)) -  \frac{1}{N-1} \sum_{i=1}^N g(\tx_i(t+1)) \right|\right\}
$$
\normalsize
In view of this, we can obtain the following bound
\begin{align}
&\rE \left[ \sup_{g \in \F} \left| \int_{\sX} g(x) \, \e[dx|\,\tbx(t+1)] - \int_{\sX} g(x) \e[dx|\,\tbx_{\{2:N\}}(t+1)] \right| \right] \nonumber \\
&\leq \frac{1}{N} \rE \left[\sup_{g \in \F} |g(\tx_1(t+1))|\right] + \frac{1}{N(N-1)} \rE\left[\sup_{g \in \F} \sum_{i=1}^N |g(\tx_i(t+1))| \right] \nonumber \\
&\leq \frac{\diam(\sX)}{N}+\frac{\diam(\sX)}{N-1} \,\, \text{(as $\|g\|_{\infty} \leq \diam(\sX)$ for all $g \in \F$)} \nonumber \\
&\leq \frac{2\diam(\sX)}{N-1} \label{n-eqq1-final}
\end{align}

\paragraph*{Bounding (\ref{n-eqq3}):}

Note that we can bound (\ref{n-eqq3}) as follows
\footnotesize
\begin{align*}
&(\ref{n-eqq3}) \\
&\leq \max \bigg\{ \left| \frac{1}{N} \sum_{i=1}^N \int_{\sX} g(x) \, p(dx|\tx_i(t),\pi^*(\tx_i(t)),\e[\,\cdot\,|\,\tbx(t)]) - \frac{1}{N} \sum_{i=2}^N \int_{\sX} g(x) \, p(dx|\tx_i(t),\pi^*(\tx_i(t)),\e[\,\cdot\,|\,\tbx(t)]) \right|, \nonumber \\
&\left| \frac{1}{N} \sum_{i=1}^N \int_{\sX} g(x) \, p(dx|\tx_i(t),\pi^*(\tx_i(t)),\e[\,\cdot\,|\,\tbx(t)]) - \frac{1}{N-1} \sum_{i=1}^N \int_{\sX} g(x) \, p(dx|\tx_i(t),\pi^*(\tx_i(t)),\e[\,\cdot\,|\,\tbx(t)]) \right| \bigg\}
\end{align*}
\normalsize
In view of this, we can obtain the following bound that is the same with (\ref{n-eqq1-final})
\footnotesize
\begin{align}
&\rE\left[ \left| \int_{\sX\times\sX} g(x) \, p(dx|z,\pi^*(z),\e[\,\cdot\,|\,\tbx(t)]) \, \e[dz|\,\tbx_{\{2:N\}}(t)] - \int_{\sX\times\sX} g(x) \, p(dx|z,\pi^*(z),\e[\,\cdot\,|\,\tbx(t)]) \, \e[dz|\,\tbx(t)] \right| \right] \nonumber \\
&\leq  \frac{2\diam(\sX)}{N-1} \label{n-eqq3-final}
\end{align}
\normalsize

\paragraph*{Bounding (\ref{n-eqq4}):}

Note that we can bound (\ref{n-eqq4}) as follows
\small
\begin{align*}
&(\ref{n-eqq4}) \leq \left| \int_{\sX\times\sX} g(x) \, p(dx|z,\pi^*(z),\e[\,\cdot\,|\,\tbx(t)]) \, \e[dz|\,\tbx(t)] - \int_{\sX\times\sX} g(x) \, p(dx|z,\pi^*(z),\mu^*) \, \e[dz|\,\tbx(t)] \right| \\
&+ \left| \int_{\sX\times\sX} g(x) \, p(dx|z,\pi^*(z),\mu^*) \, \e[dz|\,\tbx(t)] - \int_{\sX\times\sX} g(x) \, p(dx|z,\pi^*(z),\mu^*) \, \mu^*(dz) \right|
\end{align*}
\normalsize
Define 
\begin{align*}
l_1(y) &:= \int_{\sX\times\sX} g(x) \, \bp(dx|z,\pi^*(z),y) \, \e[dz|\,\tbx(t)] \\
l_2(z) &:= \int_{\sX} g(x) \, p(dx|z,\pi^*(z),\mu^*)
\end{align*}
Note that $l_1(\sX), l_2(\sX) \subset [0,K]$ and for any $y,z \in \sX$, we have $|l_1(y)-l_1(z)|\leq \K_3 \, d_{\sX}(y,z)$ and $|l_2(z)-l_2(y)| \leq (\K_1+\K_2 \, \L^*) \, d_{\sX}(z,y)$. Hence
$$
\frac{l_1}{\K_3}, \,\, \frac{l_2}{\K_1+\K_2 \, \L^*} \in \F
$$
This implies the following bound
\small
\begin{align}
&\rE\left[ \sup_{g \in \F} \left| \int_{\sX\times\sX} g(x) \, p(dx|z,\pi^*(z),\e[\,\cdot\,|\,\tbx(t)]) \, \e[dz|\,\tbx(t)] - \int_{\sX\times\sX} g(x) \, p(dx|z,\pi^*(z),\mu^*) \, \mu^*(dz)] \right| \right] \nonumber \\
&\leq \rE\left[ \sup_{g \in \F} \left\{ \left| \int_{\sX} l_1(y) \, \e[dy|\,\tbx(t)] - \int_{\sX} l_1(y) \, \mu^*(dy) \right| +\left| \int_{\sX} l_2(z) \, \e[dz|\,\tbx(t)] - \int_{\sX} l_2(z) \, \mu^*(dz) \right| \right\} \right] \nonumber \\
&\leq (\K_1 + \K_2 \, \L^* + \K_3) \, \rE\left[\sup_{g \in \F} \left| \int_{\sX} g(x) \, \e[dx|\,\tbx(t)] - \int_{\sX} g(x) \, \mu^*(dx) \right| \right] \nonumber \\
&\leq (\K_1 + \K_2 \, \L^* + \K_3) \, \talpha_t(N,\varepsilon) \,\, \text{(by induction hypothesis)} \label{n-eqq4-final}
\end{align} 
\normalsize

\paragraph*{Bounding (\ref{n-eqq2}):}

Note that for any $g \in \F$, we can bound the expectation of (\ref{n-eqq2}) as follows
\footnotesize
\begin{align}
&\rE\left[ \left| \int_{\sX} g(x) \, \e[dx|\,\tbx_{\{2:N\}}(t+1)] - \int_{\sX\times\sX} g(x) \, p(dx|z,\pi^*(z),\e[\,\cdot\,|\,\tbx(t)]) \, \e[dz|\,\tbx_{\{2:N\}}(t)] \right|\right] \nonumber \\
&= \rE\left[\rE\left[ \left| \int_{\sX} g(x) \, \e[dx|\,\tbx_{\{2:N\}}(t+1)] - \int_{\sX\times\sX} g(x) \, p(dx|z,\pi^*(z),\e[\,\cdot\,|\,\tbx(t)]) \, \e[dz|\,\tbx_{\{2:N\}}(t)] \right| \bigg| \bx(t) \right]\right] \nonumber \\ 
&\leq \rE\left[\sqrt{\frac{1}{(N-1)^2} \sum_{i=2}^N \left\{\int_{\sX} g^2(x) \, p(dx|\tx_i(t),\pi^*(\tx_i(t)),\e[\,\cdot\,|\,\tbx(t)]) - \left(\int_{\sX} g(x) \, p(dx|\tx_i(t),\pi^*(\tx_i(t)),\e[\,\cdot\,|\,\tbx(t)])\right)^2\right\}}\right] \label{n-eqq2-1} \\
&\text{(by Lemma~\ref{result7})} \nonumber
\end{align}
\normalsize
Recall that $\sup_{x \in \sX} g(x) \leq \diam(\sX)$ and $\sup_{x \in \sX} g^2(x) \leq \diam(\sX)^2$. Therefore, we can bound the last term as follows
\begin{align}
(\ref{n-eqq2-1}) &\leq \sqrt{\frac{2}{N-1}} \, \diam(\sX)\label{n-eqq2-2}
\end{align}
Note that the bound (\ref{n-eqq2-2}) is independent of $g$. Let $\{g_1,\ldots,g_{\N(\varepsilon,F)}\}$ be minimal $\varepsilon$-cover of $\F$. Then using the same trick to establish (\ref{eqq1-final}), we can obtain the following bound
\small
\begin{align}
&\rE\left[\sup_{g \in \F} \left| \int_{\sX} g(x) \, \e[dx|\,\tbx_{\{2:N\}}(t+1)] - \int_{\sX\times\sX} g(x) \, p(dx|z,\pi^*(z),\e[\,\cdot\,|\,\tbx(t)]) \, \e[dz|\,\tbx_{\{2:N\}}(t)] \right|\right] \nonumber \\
&\leq 2\varepsilon + \hspace{-7pt}\sum_{i=1}^{\N(\varepsilon,\F)} \hspace{-5pt} \rE\left[\left| \int_{\sX} g_i(x) \, \e[dx|\,\tbx_{\{2:N\}}(t+1)] - \int_{\sX\times\sX} \hspace{-7pt} g_i(x) \, p(dx|z,\pi^*(z),\e[dx|\,\tbx(t)]) \, \e[dz|\,\tbx_{\{2:N\}}(t)] \right| \right]  \nonumber \\
&\leq 2\varepsilon + \sqrt{\frac{2}{N-1}} \, \diam(\sX) \, \N(\varepsilon,\F) \,\, \text{(by (\ref{n-eqq2-2}))}\label{n-eqq2-final}
\end{align}
\normalsize
Now if we combine the bounds (\ref{n-eqq1-final}), (\ref{n-eqq2-final}), (\ref{n-eqq3-final}), and (\ref{n-eqq4-final}) for the expectations (uniform in $g \in \F$) of the terms (\ref{n-eqq1}), (\ref{n-eqq2}), (\ref{n-eqq3}), and (\ref{n-eqq4}), we obtain the following
\begin{align*}
&\rE\left[\sup_{g \in \F} \left| \int_{\sX} g(x) \, \e[dx|\,\tbx(t)+1] - \int_{\sX} g(x) \, \mu^*(dx) \right| \right] \\
&\leq (\K_1 + \K_2 \, \L^* + \K_3) \, \talpha_t(N,\varepsilon) + \sqrt{\frac{2}{N-1}} \diam(\sX) \left\{ \N(\varepsilon,\F) + \sqrt{\frac{8}{N-1}}\right\}+ 2\varepsilon
\end{align*}
This completes the proof.
\end{proof}

We now prove our second important result. Using the above lemma and the Lipschitz continuity of $\bc$ and $\gamma^{*,N}$, we deduce that the cost of agent 1 under the joint policy $(\gamma^{*,N},\pi^,\ldots,\pi^)$ in the $N$-agent game, where $N$ is large, should be close to the cost of a generic agent under the policy $\gamma^{*,N}$ in the infinite population limit.
To establish this result, in addition to assumptions (a)-(f), we impose the following condition on $N$ in the remainder of the paper
\begin{itemize}
\item[(g)] We assume that $N$ is large enough so that $\beta \, \hkappa_2^N<1$, where  $\hkappa_2^N := \K_1 + \K_2 \, \L_2^{*,N}$.
\end{itemize}
Note that $\L_2^{*,N}:= \frac{2\rL_2(\R_1^N,\R_2^N,\R_3^N)}{\rho}$, where 
$$
\R_1^N := \K_1^{*,N}, \,\, \R_2^N:=\left\{\K_2^{*,N} \K_3 + \K_2^{*,N}(\K_1+\K_2\L^*) \frac{N}{N-1}\right\}, \,\, \R_3^N:=\frac{\K_2^{*,N}(\K_1+\K_2\L^*)}{N-1}
$$
Here 
\begin{align*}
\K_1^{*,N} \rightarrow \frac{\L_1}{1-\beta\,\K_1}, \,\,\, \K_2^{*,N} \rightarrow \frac{\L_3}{1-\beta(\K_1+\K_2 \, \L^* + \K_3)} 
\end{align*}
as $N \rightarrow \infty$. Hence $\R_1^N \rightarrow \R_1$, $\R_2^N \rightarrow \R_2$, and $\R_3^N \rightarrow 0$ as $N \rightarrow \infty$, where $\R_1$ and $\R_2$ are defined in assumption (e). This implies that since $\rL_2$ is continuous,  $\rL_2(\R_1^N,\R_2^N,\R_3^N) \rightarrow \rL_2(\R_1,\R_2,0)$, and so, $\L_2^{*,N} \rightarrow \L_2^*$ as $N\rightarrow\infty$.  Therefore, assumption (g) holds for sufficiently large $N$ values in view of assumption (e).

\begin{theorem}\label{theorem2}
Given any $\varepsilon>0$, we have 
\begin{align*}
&|J_1(\mu^*;(\gamma^{*,N},\pi^*,\ldots,\pi^*))-J(\mu^*;\gamma^{*,N})| \\
&\phantom{xxxxxx}:= |\inf_{\pi \in \Pi} J_1(\mu^*;(\pi,\pi^*,\ldots,\pi^*))- J(\mu^*;\gamma^{*,N})| \leq \Theta_2(N,\varepsilon)
\end{align*}
where 
\small
$$
\Theta_2(N,\varepsilon) := \left\{\L_2 \, \L_1^{*,N} + \L_3 + (\L_1+\L_2 \, \L_2^{*,N}) \, \hkappa_1^N \, \beta \, \frac{1}{1-\beta \, \hkappa_2^N} \right\} \, \left\{\frac{\alpha(N)}{1-\beta\kappa_1}+\frac{\beta\left(\kappa_2^N \, \sqrt{\frac{2}{N-1}}+2\varepsilon\right)}{(1-\beta)(1-\beta\kappa_1)} \right\}
$$
\normalsize
\end{theorem}

\begin{proof}
We complete the proof in two steps.

\subsubsection*{Step 1:} We first prove that for any $t\geq0$, we have 
\begin{align}
\sup_{g \in \F} \left| \rE[g(\tx_1(t))] - \rE[g(\hx(t))] \right| \leq \halpha_t(N,\varepsilon)
\end{align}
where error bounds $\{\halpha_t(N,\varepsilon)\}_{t\geq0}$ are defined recursively as follows
\begin{align*}
\halpha_0(N,\varepsilon) &:=0 \\
\halpha_{t+1}(N,\varepsilon)&:= (\K_3 + \K_2 \, \L_1^{*,N}) \talpha_t(N,\varepsilon) + (\K_1+\K_2 \, \L_2^{*,N}) \, \halpha_t(N,\varepsilon) \\
&:= \hkappa_1^N \, \talpha_t(N,\varepsilon) + \hkappa_2^N \, \halpha_t(N,\varepsilon)
\end{align*}
Therefore, for all $t\geq 1$, we can write 
$$
\halpha_t(N,\varepsilon) = \hkappa_1^N \, \sum_{i=0}^{t-1} \left(\hkappa_2^N\right)^{i} \, \talpha_{t-1-i}(N,\varepsilon) 
$$

Since $\tx_1(0),\hx(0) \sim \mu^*$,  the claim is true for $t=0$. Suppose it is true for some $t\geq0$ and consider $t+1$. Fix any $g \in \F$. Then we have 
\footnotesize
\begin{align}
&\left| \rE[g(\tx_1(t+1))] - \rE[g(\hx(t+1))] \right| \nonumber \\
&= \left| \rE\left[ \int_{\sX} g(y) \, p(dy|\tx_1(t),\gamma^{*,N}(\e[\,\cdot\,|\,\tbx(t)],\tx_1(t)),\e[\,\cdot\,|\,\tbx(t)]) \right] - \rE\left[ \int_{\sX} g(y) \, p(dy|\hx(t),\gamma^{*,N}(\mu^*,\hx(t)),\mu^*) \right] \right| \nonumber \\
&\leq \left| \rE\left[ \int_{\sX} g(y) \, p(dy|\tx_1(t),\gamma^{*,N}(\e[\,\cdot\,|\,\tbx(t)],\tx_1(t)),\e[\,\cdot\,|\,\tbx(t)]) \right] - \rE\left[ \int_{\sX} g(y) \, p(dy|\tx_1(t),\gamma^{*,N}(\e[\,\cdot\,|\,\tbx(t)],\tx_1(t)),\mu^*) \right] \right| \label{thm2-eq1} \\
&+\left| \rE\left[ \int_{\sX} g(y) \, p(dy|\tx_1(t),\gamma^{*,N}(\e[\,\cdot\,|\,\tbx(t)],\tx_1(t)),\mu^*) \right] - \rE\left[ \int_{\sX} g(y) \, p(dy|\tx_1(t),\gamma^{*,N}(\mu^*,\tx_1(t)),\mu^*) \right] \right| \label{thm2-eq2} \\
&+\left| \rE\left[ \int_{\sX} g(y) \, p(dy|\tx_1(t),\gamma^{*,N}(\mu^*,\tx_1(t)),\mu^*) \right] - \rE\left[ \int_{\sX} g(y) \, p(dy|\hx(t),\gamma^{*,N}(\mu^*,\hx(t)),\mu^*) \right] \right| \label{thm2-eq3}
\end{align}
\normalsize
Let us now bound the terms (\ref{thm2-eq1}), (\ref{thm2-eq2}), and (\ref{thm2-eq3}).

\paragraph*{Bounding (\ref{thm2-eq1}):}

Note that we can write
\begin{align}
(\ref{thm2-eq1}) &= \bigg| \int_{\P(\sX)\times\sX^3} g(y) \, \bp(dy|x,\gamma^{*,N}(\nu,x),z) \, \nu(dz) \, {\cal L}(\tx_1(t)|\e[\,\cdot\,|\,\tbx(t)])(dx|\nu) \, {\cal L}(\e[\,\cdot\,|\,\tbx(t)])(d\nu) \nonumber \\
&- \int_{\P(\sX)\times\sX^3} g(y) \, \bp(dy|x,\gamma^{*,N}(\nu,x),z) \, \mu^*(dz) \, {\cal L}(\tx_1(t)|\e[\,\cdot\,|\,\tbx(t)])(dx|\nu) \, {\cal L}(\e[\,\cdot\,|\,\tbx(t)])(d\nu) \bigg| \label{eqqq1-1}
\end{align}
Define 
$$
l(z,\nu) := \int_{\sX\times\sX} g(y) \, \bp(dy|x,\gamma^{*,N}(\nu,x),z) \, {\cal L}(\tx_1(t)|\e[\,\cdot\,|\,\tbx(t)])(dx|\nu)
$$
Then we have 
\begin{align}
(\ref{eqqq1-1}) &= \left| \int_{\P(\sX)\times\sX} l(z,\nu) \, \nu(dz) \,  {\cal L}(\e[\,\cdot\,|\,\tbx(t)])(d\nu) - \int_{\P(\sX)\times\sX} l(z,\nu) \, \mu^*(dz) \,  {\cal L}(\e[\,\cdot\,|\,\tbx(t)])(d\nu) \right| \nonumber \\
&\leq \int_{\P(\sX)} \left| \int_{\sX} l(z,\nu) \, \nu(dz) - \int_{\sX} l(z,\nu) \, \mu^*(dz) \right|  {\cal L}(\e[\,\cdot\,|\,\tbx(t)])(d\nu) \nonumber \\
&\leq \int_{\P(\sX)} \sup_{\mu \in \P(\sX)} \left| \int_{\sX} l(z,\mu) \, \nu(dz) - \int_{\sX} l(z,\mu) \, \mu^*(dz) \right|  {\cal L}(\e[\,\cdot\,|\,\tbx(t)])(d\nu)  \nonumber \\
&= \rE\left[ \sup_{\mu \in \P(\sX)} \left| \int_{\sX} l(z,\mu) \, \e[dz|\,\tbx(t)] - \int_{\sX} l(z,\mu) \, \mu^*(dz) \right| \right] \label{eqqq1-2}
\end{align}
Note that for all $\mu$, $l(\sX,\mu) \subset [0,K]$ and for any $z,r \in \sX$, we have $|l(z,\mu)-l(r,\mu)| \leq \K_3 \, d_{\sX}(z,r)$. Hence $\frac{l(\cdot,\mu)}{\K_3} \in \F$ for all $\mu$. This implies that 
\begin{align}
(\ref{eqqq1-2}) &\leq \K_3 \, \rE\left[ \sup_{g \in \F} \left| \int_{\sX} g(z) \, \e[dz|\,\tbx(t)] - \int_{\sX} g(z) \, \mu^*(dz) \right| \right] \nonumber \\
&\leq \K_3 \, \talpha_t(N,\varepsilon) \,\, \text{(by Lemma~\ref{lemma8})}\label{thm2-eq1-final}
\end{align}

\paragraph*{Bounding (\ref{thm2-eq3}):}

Define 
$$
l(x) := \int_{\sX\times\sX} g(y) \, p(dy|x,\gamma^{*,N}(\mu^*,x),\mu^*) 
$$
Note that $l(\sX) \subset [0,K]$ and for all $x,y \in \sX$, we have $|l(x)-l(y)|\leq (\K_1+\K_2 \, \L_2^{*,N}) \, d_{\sX}(x,y)$. Hence $\frac{l}{\K_1+\K_2 \, \L_2^{*,N}} \in \F$. This implies the following bound on (\ref{thm2-eq3})
\begin{align}
(\ref{thm2-eq3})&= \left| \rE[l(\tx_1(t))] - \rE[l(\hx(t))] \right| \nonumber \\
&\leq (\K_1+\K_2 \, \L_2^{*,N}) \, \sup_{g \in \F} \left| \rE[g(\tx_1(t))] - \rE[g(\hx(t))] \right| \nonumber \\
&\leq  (\K_1+\K_2 \, \L_2^{*,N}) \, \halpha_t(N,\varepsilon) \,\, \text{(by induction hypothesis)} \label{thm2-eq3-final}
\end{align}

\paragraph*{Bounding (\ref{thm2-eq2}):}

We can bound (\ref{thm2-eq2}) as follows
\small
\begin{align}
&(\ref{thm2-eq2}) \nonumber \\
&\leq \rE\left[ \left| \int_{\sX} g(y) \, p(dy|\tx_1(t),\gamma^{*,N}(\e[\,\cdot\,|\,\tbx(t)],\tx_1(t)),\mu^*) - \int_{\sX} g(y) \, p(dy|\tx_1(t),\gamma^{*,N}(\mu^*,\tx_1(t)),\mu^*) \right| \right] \nonumber \\
&\leq \K_2 \, \L_1^{*,N} \, \rE[W_1(\e[\,\cdot\,|\,\tbx(t)],\mu^*)] \nonumber \\
&\leq \K_2 \, \L_1^{*,N} \, \talpha_t(N,\varepsilon) \,\, \text{(by Lemma~\ref{lemma8})} \label{thm2-eq2-final}
\end{align}
\normalsize
By combining the bounds (\ref{thm2-eq1-final}), (\ref{thm2-eq2-final}), and (\ref{thm2-eq3-final}), we obtain 
\begin{align*}
\sup_{g \in \F} \left| \rE[g(\tx_1(t+1))] - \rE[g(\hx(t+1))] \right|  \leq (\K_3 + \K_2 \, \L_1^{*,N}) \talpha_t(N,\varepsilon) + (\K_1+\K_2 \, \L_2^{*,N}) \, \halpha_t(N,\varepsilon)
\end{align*}
This completes the proof of Step~1.

\subsubsection*{Step 2:}

Secondly, we prove that for any $t\geq0$, we have 
\begin{align}
\left| \rE[c(\tx_1(t),\ta_1(t),\e[\,\cdot\,|\,\tbx(t)])] - \rE[c(\hx(t),\ha(t),\mu^*)] \right| \leq \balpha_t(N,\varepsilon)
\end{align}
where error bounds $\{\balpha_t(N,\varepsilon)\}_{t\geq0}$ are defined as follows
\begin{align*}
\balpha_{t}(N,\varepsilon)&:= (\L_2 \, \L_1^{*,N} + \L_3) \, \talpha_t(N,\varepsilon)  + (\L_1 + \L_2 \, \L_2^{*,N}) \, \halpha_t(N,\varepsilon)
\end{align*}
Therefore, for all $t\geq 1$, we can write 
$$
\balpha_t(N,\varepsilon) = \hkappa_1^N \, \sum_{i=1}^{t-1} \left(\hkappa_2^N\right)^{i-1} \, \talpha_{t-i}(N,\varepsilon) 
$$

Indeed, since $\tx_1(0),\hx(0) \sim \mu^*$, we have 
\begin{align}
&\left| \rE[c(\tx_1(0),\ta_1(0),\e[\,\cdot\,|\,\tbx(0)])] - \rE[c(\hx(0),\ha(0),\mu^*)] \right| \nonumber \\
&= \bigg| \int_{\P(\sX)\times\sX} c(x,\gamma^{*,N}(\nu,x),\nu) \, \mu^*(dx) \, {\cal L}(\e[\,\cdot\,|\,\tbx(0)])(d\nu) -\int_{\sX} c(x,\gamma^{*,N}(\mu^*,x),\mu^*) \, \mu^*(dx) \bigg| \nonumber \\
&\leq \bigg| \int_{\P(\sX)\times\sX} c(x,\gamma^{*,N}(\nu,x),\nu) \, \mu^*(dx) \, {\cal L}(\e[\,\cdot\,|\,\tbx(0)])(d\nu) \nonumber \\
&\phantom{xxxxxxxxxxxxx}
- \int_{\P(\sX)\times\sX} c(x,\gamma^{*,N}(\nu,x),\mu^*) \, \mu^*(dx) \, {\cal L}(\e[\,\cdot\,|\,\tbx(0)])(d\nu) \bigg| \nonumber \\
&+ \left| \int_{\P(\sX)\times\sX} c(x,\gamma^{*,N}(\nu,x),\mu^*) \, \mu^*(dx) \, {\cal L}(\e[\,\cdot\,|\,\tbx(0)])(d\nu) - \int_{\sX} c(x,\gamma^{*,N}(\mu^*,x),\mu^*) \, \mu^*(dx) \right| \nonumber \\
&\leq \left| \int_{\P(\sX)\times\sX} l(z,\nu) \, \nu(dz) \, {\cal L}(\e[\,\cdot\,|\,\tbx(0)])(d\nu) - \int_{\P(\sX)\times\sX} l(z,\nu) \, \mu^*(dz) \, {\cal L}(\e[\,\cdot\,|\,\tbx(0)])(d\nu) \right| \nonumber \\
&\text{(where $l(z,\nu):= \int_{\sX} \bc(x,\gamma^{*,N}(x,\nu),z) \, \mu^*(dx)$)} \nonumber \\
&+  \int_{\P(\sX)} \, \left| \int_{\sX} c(x,\gamma^{*,N}(\nu,x),\mu^*) \, \mu^*(dx) - \int_{\sX} c(x,\gamma^{*,N}(\mu^*,x),\mu^*) \, \mu^*(dx) \right| \, {\cal L}(\e[\,\cdot\,|\,\tbx(0)])(d\nu) \nonumber \\
&\leq \int_{\P(\sX)} \, \sup_{\mu \in \P(\sX)} \, \left| \int_{\sX} l(z,\mu) \, \nu(dz) - \int_{\sX} l(z,\mu) \, \mu^*(dz) \right|  \, {\cal L}(\e[\,\cdot\,|\,\tbx(0)])(d\nu) \nonumber \\
&+ \L_2 \, \L_1^{*,N} \, \int_{\P(\sX)} W_1(\nu,\mu^*) \, {\cal L}(\e[\,\cdot\,|\,\tbx(0)])(d\nu) \nonumber \\ 
&\leq \L_3 \, \rE\left[\sup_{g \in \F} \left| \int_{\sX} g(x) \, \e[dx|\,\tbx(0)]-\int_{\sX} g(x) \, \mu^*(dx) \right| \right] + \L_2 \, \L_1^{*,N} \rE[W_1(\e[\,\cdot\,|\,\tbx(0)],\mu^*)] \nonumber \\ 
&\text{(as $\frac{l(\cdot,\mu)}{\L_3}\in \F$, \, $\forall \mu$)} \nonumber \\
&\leq (\L_2 \, \L_1^{*,N} + \L_3) \, \talpha_0(N,\varepsilon) \,\, \text{(by Lemma~\ref{lemma8})} \nonumber 
\end{align}
Hence the claim is true for $t=0$. Now consider any $t\geq1$. Then we have 
\begin{align}
&\left| \rE[c(\tx_1(t),\ta_1(t),\e[\,\cdot\,|\,\tbx(t)])] - \rE[c(\hx(t),\ha(t),\mu^*)] \right| \nonumber \\ 
&\leq \left| \rE[c(\tx_1(t),\gamma^{*,N}(\e[\,\cdot\,|\,\tbx(t)],\tx_1(t)),\e[\,\cdot\,|\,\tbx(t)])] - \rE[c(\tx_1(t),\gamma^{*,N}(\e[\,\cdot\,|\,\tbx(t)],\tx_1(t)),\mu^*)]  \right| \label{step2-eq1} \\
&+ \left| \rE[c(\tx_1(t),\gamma^{*,N}(\e[\,\cdot\,|\,\tbx(t)],\tx_1(t)),\mu^*)]  - \rE[c(\tx_1(t),\gamma^{*,N}(\mu^*,\tx_1(t)),\mu^*)]  \right| \label{step2-eq2} \\
&+ \left| \rE[c(\tx_1(t),\gamma^{*,N}(\mu^*,\tx_1(t)),\mu^*)] - \rE[c(\hx(t),\gamma^{*,N}(\mu^*,\hx(t)),\mu^*)]  \right| \label{step2-eq3}
\end{align}
Now let us bound the terms (\ref{step2-eq1}), (\ref{step2-eq2}), and (\ref{step2-eq3}). 

\paragraph*{Bounding (\ref{step2-eq1}):}

To this end, define 
$$
l(z,\nu) := \int_{\sX} \bc(x,\gamma^{*,N}(\nu,x),z) \, {\cal L}(\tx_1(t)|\e[\,\cdot\,|\,\tbx(t)])(dx|\nu)
$$
Then we can write (\ref{step2-eq1}) as follows
\begin{align}
 &(\ref{step2-eq1}) = \left| \int_{\P(\sX)\times\sX} l(z,\nu) \, \nu(dz) \, {\cal L}(\e[\,\cdot\,|\,\tbx(t)])(d\nu) - \int_{\P(\sX)\times\sX} l(z,\nu) \, \mu^*(dz) \, {\cal L}(\e[\,\cdot\,|\,\tbx(t)])(d\nu) \right| \nonumber \\
 &\leq \int_{\P(\sX)} \sup_{\mu \in \P(\sX)} \left|  \int_{\sX} l(z,\mu) \, \nu(dz)  - \int_{\sX} l(z,\mu) \, \mu^*(dz) \right| \, {\cal L}(\e[\,\cdot\,|\,\tbx(t)])(d\nu) \nonumber \\
 &\leq \L_3 \, \rE\left[\sup_{g \in \F} \left| \int_{\sX} g(x) \, \e[dx|\,\tbx(t)] - \int_{\sX} g(x) \, \mu^*(dx) \right| \right] \,\, \text{(as $\frac{l(\cdot,\mu)}{\L_3}\in \F$, \, $\forall \mu$)} \nonumber \\
 &\leq \L_3 \, \talpha_t(N,\varepsilon) \,\, \text{(by Lemma~\ref{lemma8})}\label{step2-eq1-final}
\end{align}

\paragraph*{Bounding (\ref{step2-eq2}):}

Note that we have
\begin{align}
(\ref{step2-eq2}) &\leq \L_2 \, \L_1^{*,N} \, \rE[W_1(\e[\,\cdot\,|\,\tbx(t)],\mu^*)] \leq \L_2 \, \L_1^{*,N} \,\talpha_t(N,\varepsilon) \,\, \text{(by Lemma~\ref{lemma8})} \label{step2-eq2-final}
\end{align}

\paragraph*{Bounding (\ref{step2-eq3}):}

To this end, define $l(x) := c(x,\gamma^{*,N}(\mu^*,x),\mu^*)$. Then we can bound
(\ref{step2-eq3}) as follows
\begin{align}
(\ref{step2-eq3}) &= |\rE[l(\tx_1(t))]-\rE[l(\hx(t))]| \nonumber \\
&\leq (\L_1 + \L_2 \, \L_2^{*,N}) \, \sup_{g \in \F} |\rE[g(\tx_1(t))]-\rE[g(\hx(t))]| \,\, \text{(as $\frac{l}{\L_1 + \L_2 \, \L_2^{*,N}} \in \F$)} \nonumber \\
&\leq (\L_1 + \L_2 \, \L_2^{*,N}) \, \halpha_t(N,\varepsilon) \,\, \text{(by Step 1)} \label{step2-eq3-final}
\end{align}

Now if we combine the bounds (\ref{step2-eq1-final}), (\ref{step2-eq2-final}), and (\ref{step2-eq3-final}), we obtain the following
\begin{align}
&\left| \rE[c(\tx_1(t),\ta_1(t),\e[\,\cdot\,|\,\tbx(t)])] - \rE[c(\hx(t),\ha(t),\mu^*)] \right| \nonumber \\
&\leq (\L_2 \, \L_1^{*,N} + \L_3) \, \talpha_t(N,\varepsilon)  + (\L_1 + \L_2 \, \L_2^{*,N}) \, \halpha_t(N,\varepsilon) \nonumber
\end{align}
This completes the proof of Step 2.

Now it is time to complete the proof. Note that we have 
\begin{align}
&|J_1(\mu^*;(\gamma^{*,N},\pi^*,\ldots,\pi^*))-J(\mu^*;\gamma^{*,N})| \nonumber \\
&\leq \sum_{t=0}^{\infty} \beta^t \,  \left| \rE[c(\tx_1(t),\ta_1(t),\e[\,\cdot\,|\,\tbx(t)])] - \rE[c(\hx(t),\ha(t),\mu^*)] \right| \nonumber \\
&\leq  (\L_2 \, \L_1^{*,N} + \L_3) \, \sum_{t=0}^{\infty} \beta^t \,\talpha_t(N,\varepsilon) + (\L_1 + \L_2 \, \L_2^{*,N}) \, \sum_{t=0}^{\infty} \beta^t \,\halpha_t(N,\varepsilon) \nonumber \\
&= (\L_2 \, \L_1^{*,N} + \L_3) \, \sum_{t=0}^{\infty} \beta^t \, \left\{\kappa_1^t \, \alpha(N) + \kappa_2^N \, \sqrt{\frac{2}{N-1}} \sum_{i=0}^{t-1} \kappa_1^i  + 2\varepsilon\sum_{i=0}^{t-1} \kappa_1^i\right\} \nonumber \\
&+ (\L_1 + \L_2 \, \L_2^{*,N}) \, \sum_{t=0}^{\infty} \beta^t \, \left\{\hkappa_1^N \, \sum_{i=0}^{t-1} \left(\hkappa_2^N\right)^{i} \, \talpha_{t-1-i}(N,\varepsilon) \right\} \nonumber \\
&=  (\L_2 \, \L_1^{*,N} + \L_3) \, \left\{\frac{\alpha(N)}{1-\beta\kappa_1}+\frac{\beta\left(\kappa_2^N \, \sqrt{\frac{2}{N-1}}+2\varepsilon\right)}{(1-\beta)(1-\beta\kappa_1)} \right\} \nonumber \\
&+ (\L_1+\L_2 \, \L_2^{*,N}) \, \hkappa_1^N \, \beta \, \frac{1}{1-\beta \, \hkappa_2^N} \, \sum_{t=0}^{\infty} \beta^t \, \talpha_t(N,\varepsilon) \,\, \text{(as $\beta \, \hkappa_2^N < 1$ by assumption (g))} \nonumber \\
&= \left\{\L_2 \, \L_1^{*,N} + \L_3 + (\L_1+\L_2 \, \L_2^{*,N}) \, \hkappa_1^N \, \beta \, \frac{1}{1-\beta \, \hkappa_2^N} \right\} \, \left\{\frac{\alpha(N)}{1-\beta\kappa_1}+\frac{\beta\left(\kappa_2^N \, \sqrt{\frac{2}{N-1}}+2\varepsilon\right)}{(1-\beta)(1-\beta\kappa_1)} \right\} 
\end{align}
This completes the proof. 
\end{proof}

So far we proved two important results:
\begin{itemize}
\item[(1) ] Theorem~\ref{theorem1}: Given any $\varepsilon>0$, we have 
$$
|J_1(\mu^*;\bpi^*)-J(\mu^*;\pi^*)| := |J_1(\mu^*;\bpi^*)- \inf_{\pi \in \Pi_l}J(\mu^*;\pi)| \leq \Theta_1(N,\varepsilon)
$$
where 
$$
\Theta_1(N,\varepsilon) :=  (\L_1+\L_2 \, \L^*+\L_3) \, \left\{ \frac{\alpha(N) }{1-\beta \, \kappa_1} +  \frac{\beta\, \left(\kappa_2 \, \sqrt{\frac{2}{N}}+2\varepsilon\right)}{(1-\beta \, \kappa_1)(1-\beta)}\right\}
$$
\item[(2)] Theorem~\ref{theorem2}: Given any $\varepsilon>0$, we have 
\begin{align*}
&|J_1(\mu^*;(\gamma^{*,N},\pi^*,\ldots,\pi^*))-J(\mu^*;\gamma^{*,N})| \\
&\phantom{xxxxxxxxx}:= |\inf_{\pi \in \Pi} J_1(\mu^*;(\pi,\pi^*,\ldots,\pi^*))- J(\mu^*;\gamma^{*,N})| \leq \Theta_2(N,\varepsilon)
\end{align*}
where 
\small
$$
\hspace{-10pt}\Theta_2(N,\varepsilon) := \left\{\L_2 \, \L_1^{*,N} + \L_3 + (\L_1+\L_2 \, \L_2^{*,N}) \, \hkappa_1^N \, \beta \, \frac{1}{1-\beta \, \hkappa_2^N} \right\} \, \left\{\frac{\alpha(N)}{1-\beta\kappa_1}+\frac{\beta\left(\kappa_2^N \, \sqrt{\frac{2}{N-1}}+2\varepsilon\right)}{(1-\beta)(1-\beta\kappa_1)} \right\}
$$
\normalsize
\end{itemize}
The constants $\kappa_2 := \diam(\sX) \, \N(\varepsilon,\F)$ and $\kappa_2^N := \diam(\sX) \left\{ \N(\varepsilon,\F) + \sqrt{\frac{8}{N-1}} \right\}$ also depend on $\varepsilon$, meaning that they increase as $\varepsilon$ approaches $0$. However, because these constants are multiplied by the terms $\sqrt{\frac{2}{N}}$ and $\sqrt{\frac{2}{N-1}}$, respectively, the bounds $\Theta_1(N,\varepsilon)$ and $\Theta(N,\varepsilon)$ can be made arbitrarily small by first choosing a small value for $\varepsilon$ and then choosing a large value for $N$. We are now ready to prove the main result of this paper.

\begin{theorem}\label{theorem3}
If MFE policy $\pi^*$ is applied by all agents in the $N$-agent linear mean-field game with initial distribution $\mu^{*,\otimes N}$, then the joint policy $(\pi^*,\ldots,\pi^*)$ is $\Theta_1(N,\varepsilon)+\Theta_2(N,\varepsilon)$-Nash equilibrium for any $\varepsilon>0$; that is,
\begin{align} \label{main1}
J_i(\mu^*;\bpi^*) \leq \inf_{\pi \in \Pi} J_i(\mu^*;\pi,\bpi^*_{-i}) + \Theta_1(N,\varepsilon)+\Theta_2(N,\varepsilon)
\end{align}
for all $i=1,\ldots,N$.
\end{theorem}

\begin{proof}
Since the game model is symmetric, it is sufficient to prove (\ref{main1}) for agent~1. In this case, by Theorem~\ref{theorem1} and Theorem~\ref{theorem2}, we have
\begin{align*}
&J_1(\mu^*;\bpi^*) - \inf_{\pi \in \Pi} J_1(\mu^*;\pi,\pi^*,\ldots,\pi^*) = J_1(\mu^*;\bpi^*) - J_1(\mu^*;\gamma^{*,N},\pi^*,\ldots,\pi^*) \\
&= J_1(\mu^*;\bpi^*) - J(\mu^*;\pi^*) + J(\mu^*;\pi^*) - J(\mu^*;\gamma^{*,N}) + J(\mu^*;\gamma^{*,N}) - J_1(\mu^*;\gamma^{*,N},\pi^*,\ldots,\pi^*) \\
&\leq |J_1(\mu^*;\bpi^*) - J(\mu^*;\pi^*)| + |J(\mu^*;\gamma^{*,N}) - J_1(\mu^*;\gamma^{*,N},\pi^*,\ldots,\pi^*)| \\
&\text{(as $J(\mu^*;\pi^*) - J(\mu^*;\gamma^{*,N}) = \inf_{\pi \in \Pi_l} J(\mu^*;\pi) - J(\mu^*;\gamma^{*,N}) \leq 0$)} \\
&\leq \Theta_1(N,\varepsilon)+\Theta_2(N,\varepsilon)
\end{align*}
This completes the proof.
\end{proof}

\section{Infinite Population Game as GNEP}\label{sec4}

In this section, we express the game problem in the infinite population limit as a generalized Nash equilibrium problem (GNEP). We then use this new formulation to compute the mean-field equilibrium using existing algorithms that have been developed for GNEPs in the literature \cite{FaKa10}.

In mean-field games, the most naive approach for computing MFE is the following.
Note that, given any limiting mean-field term $\mu \in \P(\sX)$, the optimal control problem for the mean-field game reduces to finding an optimal stationary policy for a Markov decision process (MDP). Hence, one can compute the optimal policy $\pi$ for $\mu$ using various algorithms such as value iteration, policy iteration, and $Q$-iteration. Then, given this optimal policy $\pi$, one can compute the invariant distribution $\mu^+$ of the transition probability $p_{\mu}^{\pi}(\,\cdot\,|x)$. This overall process defines an operator $H$ from $\P(\sX)$ to $\P(\sX)$ as follows $\mu^+ = H(\mu)$. If one can prove that $H$ is a contraction, then by Banach fixed point theorem, the iterates in this recursion converges to the unique fixed point $\mu^*$ of the operator $H$; that is, $\mu^*=H(\mu^*)$, and so, $\mu^*$ and the corresponding optimal policy $\pi^*$ constitute a MFE. This approach was indeed adapted in \cite{AnKaSa20} for classical mean-field games. However, it turns out that to make $H$ contraction, we need quite restrictive conditions on the system components of the model. 

In this paper, we follow the following alternative route. Given any limiting mean-field term $\mu \in \P(\sX)$, we formulate the corresponding MDP as a linear program (LP) using occupation measures, which is a well-established method in stochastic control. Then, we incorporate the mean-field consistency condition to this LP formulation and obtain a generalized Nash equilibrium problem. By adapting one of the methods developed for solving GNEPs to our problem, we establish an algorithm for computing MFE.  

\subsection{GNEP Formulation}\label{sub1sec3}

Note that, given any $\mu \in \P(\sX)$, the corresponding optimal control problem is an MDP. Therefore, in this section, we first introduce the LP formulation of this MDP using occupation measures. We refer the reader to \cite{HeGo00} and \cite[Chapter 6]{HeLa96} for the LP formulation of MDPs with discounted cost.

For any metric space $\sE$, let $\M(\sE)$ denote the set of finite signed measures on $\sE$ and $B(\sE)$ denotes the set of bounded measurable real functions. Consider the vector spaces $\bigl(\M(\sX\times\sA),B(\sX\times\sA)\bigr)$ and $\bigl(\M(\sX),B(\sX)\bigr)$. Let us define bilinear forms on $\bigl(\M(\sX\times\sA),B(\sX\times\sA)\bigr)$ and on $\bigl(\M(\sX),B(\sX)\bigr)$ as follows
\begin{align}
\langle \zeta,v  \rangle &\coloneqq \int_{\sX\times\sA} v(x,a) \, \zeta(dx,da) \label{n-eqqq1} \\
\langle \nu,u  \rangle &\coloneqq \int_{\sX} u(x) \, \nu(dx) \label{n-eqqq2}
\end{align}
where $\zeta \in \M(\sX\times\sA)$, $v \in B(\sX\times\sA)$, $\nu \in \M(\sX)$, and $u \in B(\sX)$. The bilinear form in (\ref{n-eqqq1}) constitutes duality between $\M(\sX\times\sA)$ and $B(\sX\times\sA)$, and the bilinear form in (\ref{n-eqqq2}) constitutes duality between $\M(\sX)$ and $B(\sX)$. For any $\zeta \in \M(\sX\times\sA)$, let $\hat{\zeta} \in \M(\sX)$ denote the marginal of $\zeta$ on $\sX$, i.e.,
\begin{align}
\hat{\zeta}(\,\cdot\,)=\zeta(\,\cdot\,\times \sA) \nonumber
\end{align}
We define the linear map $\T_{\mu}: \M(\sX\times\sA) \rightarrow \M(\sX)$ by
\begin{align}
{\T}_{\mu}\zeta(\,\cdot\,) &= \hat{\zeta}(\,\cdot\,) - \beta \int_{\sX\times\sA} p_{\mu}(\,\cdot\,|x,a) \, \zeta(dx,da) =: \hat{\zeta} - \beta \, \zeta \, p_{\mu} \nonumber 
\end{align}
which depends on $\mu$.

Recall that $\text{MDP}_{\mu}$ has the following components
\begin{align*}
\left\{\sX,\sA,c_{\mu},p_{\mu},\mu\right\}
\end{align*}
where 
\begin{align*}
c_{\mu}(x,a) &:= c(x,a,\mu) = \int_{\sX} \bc(x,a,z) \, \mu(dz) \\
p_{\mu}(\,\cdot\,|\,x,a) &:= p(\,\cdot\,|\,x,a,\mu) = \int_{\sX} \bp(\,\cdot\,|\,x,a,z) \, \mu(dz)
\end{align*}
It is indeed equivalent to the following equality constrained linear program \cite[Lemma 3.3 and Section 4]{HeGo00}:
\begin{align}
\text{                         }&\text{minimize}_{\zeta\in \M_+(\sX\times\sA)} \text{ } \langle \zeta,c_{\mu} \rangle
\nonumber \\*
&\text{subject to  } {\T}_{\mu}(\zeta) = (1-\beta)\mu  \label{aaaaa}
\end{align}
Indeed, for any policy $\pi$, define the $\beta$-discount expected occupation measure as
\begin{align}
\zeta^{\pi}(C) := (1-\beta) \sum_{t=0}^{\infty} \beta^t \, \sPr \biggl[ (x(t),a(t)) \in C \biggr], \text{ } C \in \B(\sX\times\sA) \nonumber
\end{align}
Note that $\zeta^{\pi}$ is a probability measure on $\sX \times \sA$ as a result of the normalizing constant $(1-\beta)$. One can prove that $\zeta^{\pi}$ satisfies
\begin{align}
\hat{\zeta}^{\pi}(\,\cdot\,) = (1-\beta) \, \mu(\,\cdot\,) + \beta \int_{\sX \times \sA} p_{\mu}(\,\cdot\,|x,a) \, \zeta^{\pi}(dx,da) \label{occup}
\end{align}
Conversely, if any finite measure $\zeta$ satisfies (\ref{occup}), then it is a $\beta$-discount expected occupation measure of some policy $\pi$ \cite[Lemma 3.3]{HeGo00}. Using the $\beta$-discount expected occupation measure, we can write
\begin{align}
J(\mu;\pi) = \frac{1}{(1-\beta)} \, \langle \zeta^{\pi}, c_{\mu} \rangle  \nonumber
\end{align}
Therefore, since $(1-\beta)$ is just a constant, $\text{MDP}_{\mu}$  is equivalent to the above linear program. Using LP formulation, we first establish the following result.

\begin{lemma}\label{com-lemma1}
Let $(\zeta^*,\mu^*) \in \M(\sX\times\sA)_+\times\M(\sX)_+$ be a pair with the following properties
\begin{itemize}
\item[(a)] Given $\mu^*$, $\zeta^*$ is the optimal solution to the above LP formulation of $\text{MDP}_{\mu^*}$.
\item[(b)] Given $\zeta^*$, $\mu^*$ satisfies the following linear equation
$$
\mu^*(\,\cdot\,) = \int_{\sX\times\sA} p_{\mu^*}(\,\cdot\,|x,a) \, \zeta^*(dx,da)
$$
\end{itemize} 
If we disintegrate $\zeta^*$ as follows $\zeta^*(dx,da) = \pi^*(da|x) \, \hat{\zeta}^*(dx)$, then $(\mu^*,\pi^*)$ is MFE.
\end{lemma}

\begin{proof}
Note that since $\mu^*$ and $\zeta^*$ are not assumed to be probability measures at the beginning, we need to establish this first. Since 
$$
\hat{\zeta}^* = (1-\beta) \, \mu^* + \beta \, \zeta^* \, p_{\mu^*}
$$
we have $\zeta^*(\sX\times\sA) = (1-\beta) \, \mu^*(\sX) + \beta \, \zeta^*(\sX\times\sA) \, \mu^*(\sX)$. Similarly, since 
$$
\mu^* = \zeta^* \, p_{\mu^*}
$$
we have $\mu^*(\sX) = \zeta^*(\sX\times\sA) \, \mu^*(\sX)$, which implies that $\zeta^*$ is a probability measure. In view of this and using the first identity, we obtain the following
$$
1 = (1-\beta) \, \mu^*(\sX) + \beta \, \mu^*(\sX) = \mu^*(\sX)
$$ 
that is, $\mu^*$ is also a probability measure. 

Note that $\zeta^*$ is the optimal occupation measure of the LP formulation of 
$\text{MDP}_{\mu^*}$, and so, $\pi^*$ is the optimal policy. Hence, $\pi^* \in \Lambda(\mu^*)$. Furthermore, since 
\begin{align}
\zeta^{*}(C) := (1-\beta) \sum_{t=0}^{\infty} \beta^t \, \sPr \biggl[ (x(t),a(t)) \in C \biggr], \text{ } C \in \B(\sX\times\sA) \nonumber
\end{align}
we have 
\begin{align*}
\zeta^* \, p_{\mu^*}(\,\cdot\,) &= \int_{\sX\times\sA} p_{\mu^*}(\,\cdot\,|x,a) \, \zeta^*(dx,da) \\
&=  \int_{\sX\times\sA} p_{\mu^*}(\,\cdot\,|x,a) \, \bigg\{ (1-\beta) \sum_{t=0}^{\infty} \beta^t \, \sPr \biggl[ (x(t),a(t)) \in dx \times da \biggr] \bigg\} \\
&= (1-\beta) \sum_{t=0}^{\infty} \beta^t \bigg\{ \int_{\sX\times\sA} p_{\mu^*}(\,\cdot\,|x,a) \, \sPr \biggl[ (x(t),a(t)) \in dx \times da \biggr] \bigg\} \\
&= (1-\beta) \sum_{t=0}^{\infty} \beta^t \, \sPr \biggl[ x(t+1) \in \,\cdot\, \biggr] \\
&= \frac{1-\beta}{\beta} \, \sum_{t=1}^{\infty} \beta^t \, \sPr \biggl[ x(t+1) \in \,\cdot\, \biggr] + \frac{1-\beta}{\beta} \,  \sPr \biggl[ x(0) \in \,\cdot\, \biggr] - \frac{1-\beta}{\beta} \,  \sPr \biggl[ x(0) \in \,\cdot\, \biggr] \\
&= \frac{1-\beta}{\beta} \, \sum_{t=0}^{\infty} \beta^t \, \sPr \biggl[ x(t+1) \in \,\cdot\, \biggr] - \frac{1-\beta}{\beta} \, \mu^(\,\cdot\,) \,\, \text{(as $x(0) \sim \mu^*$)} \\
&= \frac{\hat{\zeta}^*(\,\cdot\,)}{\beta} - \frac{\mu^*(\,\cdot\,)}{\beta}  + \mu^*(\,\cdot\,) 
\end{align*}
But since $\mu^* = \zeta^* \, p_{\mu^*}$, the last expression implies that $\hat{\zeta}^* = \mu^*$. Hence, $\mu^*$ satisfies the following in view of property (b)
$$
\mu^*(\,\cdot\,) = \int_{\sX\times\sA} p_{\mu^*}(\,\cdot\,|x,a) \, \pi^*(da|x) \,  \mu^*(dx)
$$
that is, $\mu^*$ is an invariant distribution of the transition probability $p_{\mu^*}^{\pi^*}$. Hence, $\mu^* \in \Phi(\pi^*)$. This means that $(\mu^*,\pi^*)$ is MFE.
\end{proof}

Hence, to find MFE, it is sufficient to compute a pair $(\zeta^*,\mu^*)$ that satisfies the properties in Lemma~\ref{com-lemma1}. To compute such a pair, we now formulate an artificial game with two players, which turns our to be a generalized Nash equilibrium problem and whose Nash equilibrium gives such a pair. In this artificial game, first player represents the generic agent in mean-field game and the second player represents the overall population. To formulate the problem, we need to define another cost function in addition to  $\langle \zeta,c_{\mu} \rangle$. This new cost function will serve as the cost of the second player in the game. It is important to note that we are completely free to choose this cost function. Hence, one can think of this additional cost as a design parameter that can be used to achieve certain objectives. Let $g: \M(\sX\times\sA) \times \M(\sX) \rightarrow [0,\infty)$ be some continuous function of $(\zeta,\mu)$. Then we define the following generalized Nash equilibrium problem

\begin{multicols}{2}
\centering{Player 1}
\begin{align*}
\text{Given $\mu$:} \,\,
&\text{minimize}_{\zeta\in \M_+(\sX\times\sA)} \text{ } \langle \zeta,c_{\mu} \rangle
\nonumber \\*
&\text{subject to  } \hat{\zeta} = (1-\beta)\mu + \beta \, \zeta \, p_{\mu}  
\end{align*}\\
\centering{Player 2}
\begin{align*}
\text{Given $\zeta$:} \,\,
&\text{minimize}_{\mu \in \M_+(\sX)} \text{ } g(\zeta,\mu)
\nonumber \\*
&\text{subject to  } \mu = \zeta \, p_{\mu}
\end{align*}
\end{multicols}

Note that in above game, both the cost function and the admissible strategy sets are coupled to each other. Therefore, it is indeed a generalized Nash equilibrium problem (see the survey paper \cite{FaKa10} for an introduction to GNEPs). Hence, we can use techniques developed for such games for computing MFE.

The following result is obvious in view of Lemma~\ref{com-lemma1}.

\begin{lemma}\label{com-lemma2}
If $(\zeta^*,\mu^*)$ is an equilibrium solution of the above GNEP, then $(\mu^*,\pi^*)$ is MFE, where $\zeta^*(dx,da) = \pi^*(da|x) \, \hat{\zeta}^*(dx)$.
\end{lemma}

In general, GNEP problems are formulated via inequality constraints instead of equality constraints. Although it is possible to transform equality constraints into inequality constraints by doubling the number of constraints, we can alternatively formulate above GNEP via inequality constraints without increasing the number of constraints much as follows. 

\begin{multicols}{2}
\centering{Player 1}
\begin{align*}
\text{Given $\mu$:} \,\,
&\text{minimize}_{\zeta\in \M_+(\sX\times\sA)} \text{ } \langle \zeta,c_{\mu} \rangle
\nonumber \\*
&\text{subject to  } \hat{\zeta} \geq (1-\beta)\mu + \beta \, \zeta \, p_{\mu}  
\end{align*}\\
\centering{Player 2}
\begin{align*}
\text{Given $\zeta$:} \,\,
&\text{minimize}_{\mu \in \M_+(\sX)} \text{ } g(\zeta,\mu)
\nonumber \\*
&\text{subject to  } \mu \geq \zeta \, p_{\mu}, \,\, \langle \mu, {\bf 1} \rangle \geq 1
\end{align*}
\end{multicols}
Here, ${\bf 1}$ denotes the constant function $1$. To express the GNEP problem using inequality constraints, an additional constraint $\langle \mu, {\bf 1} \rangle \geq 1$ is added without the need to increase the number of constraints significantly. While the following result is similar to Lemma~\ref{com-lemma2}, its proof is not straightforward and will be provided.  

\begin{lemma}\label{com-lemma3}
If $(\zeta^*,\mu^*)$ is an equilibrium solution of the above GNEP with inequality constraint, then $(\mu^*,\pi^*)$ is MFE, where $\zeta^*(dx,da) = \pi^*(da|x) \, \hat{\zeta}^*(dx)$.
\end{lemma}

\begin{proof}
Since 
$$
\hat{\zeta}^* \geq (1-\beta) \, \mu^* + \beta \, \zeta^* \, p_{\mu^*}
$$
we have $\zeta^*(\sX\times\sA) \geq (1-\beta) \, \mu^*(\sX) + \beta \, \zeta^*(\sX\times\sA) \, \mu^*(\sX)$. Similarly, since 
$$
\mu^* \geq \zeta^* \, p_{\mu^*}
$$
we have $\mu^*(\sX) \geq \zeta^*(\sX\times\sA) \, \mu^*(\sX)$. Hence 
\begin{align*}
\zeta^*(\sX\times\sA)  &\geq (1-\beta) \zeta^*(\sX\times\sA) \, \mu^*(\sX) +  \beta \, \zeta^*(\sX\times\sA) \, \mu^*(\sX) = \zeta^*(\sX\times\sA) \, \mu^*(\sX) \\
\intertext{and}
\mu^*(\sX) &\geq \zeta^*(\sX\times\sA) \, \mu^*(\sX)
\end{align*} 
Therefore, $\mu^*(\sX) \leq 1$ and $\zeta^*(\sX\times\sA)\leq 1$. Since $\langle \mu^*, {\bf 1} \rangle = \mu^*(\sX) \geq 1$, we also have $\mu^*(\sX) = 1$ and
$$
\zeta^*(\sX\times\sA) \geq (1-\beta) + \beta \, \zeta^*(\sX\times\sA) 
$$
Hence $\zeta^*(\sX\times\sA) \geq 1$. This implies that $\zeta^*(\sX\times\sA) = 1$. That is, both $\zeta^*$ and $\mu^*$ are probability measures. 
Therefore, $(1-\beta) \, \mu^* + \beta \, \zeta^* \, p_{\mu^*}$ and $\zeta^* \, p_{\mu^*}$ are also probability measures. But it is known that if two probability measures $\nu$ and $\theta$ satisfy $\nu \leq \theta$ for any Borel set, then $\nu=\theta$. Hence, 
\begin{align*}
\hat{\zeta}^* &= (1-\beta) \, \mu^* + \beta \, \zeta^* \, p_{\mu^*} \\
\mu^* &= \zeta^* \, p_{\mu^*}
\end{align*}
Note that given $\mu^*$, the following optimization problems are equivalent
\begin{multicols}{2}
\centering{Problem 1}
\begin{align*}
&\text{minimize}_{\zeta\in \M_+(\sX\times\sA)} \text{ } \langle \zeta,c_{\mu^*} \rangle
\nonumber \\*
&\text{subject to  } \hat{\zeta} \geq (1-\beta)\mu^* + \beta \, \zeta \, p_{\mu^*}  
\end{align*}\\
\centering{Problem 2}
\begin{align*}
&\text{minimize}_{\zeta\in \M_+(\sX\times\sA)} \text{ } \langle \zeta,c_{\mu^*} \rangle
\nonumber \\*
&\text{subject to  } \hat{\zeta} = (1-\beta)\mu^* + \beta \, \zeta \, p_{\mu^*}  
\end{align*}
\end{multicols}
The second problem is a LP formulation of MDP~$_{\mu^*}$ and so $\zeta^*$ is the optimal occupation measure. Hence $\pi^* \in \Lambda(\mu^*)$. Using the same method as in the proof of Lemma~\ref{com-lemma1}, we can also prove that $\hat{\zeta}^*=\mu^*$ using $\mu^* = \zeta^* \, p_{\mu^*}$. Hence, by $\hat{\zeta}^* = (1-\beta) \, \mu^* + \beta \, \zeta^* \, p_{\mu^*}$, we have 
$$
\mu^*(\,\cdot\,) = \int_{\sX\times\sA} p_{\mu^*}(\,\cdot\,|x,a) \, \pi^*(da|x) \,  \mu^*(dx)
$$
that is, $\mu^*$ is an invariant distribution of the transition probability $p_{\mu^*}^{\pi^*}$. Hence, $\mu^* \in \Phi(\pi^*)$. This implies that $(\mu^*,\pi^*)$ is MFE.
\end{proof}

\subsection{Computing Equilibrium of GNEP}\label{sub1sec3}

Note that if the action space $\sA$ is compact and convex subset of some finite-dimensional Euclidean space, it should be necessarily  uncountably infinite. Hence, even if the state space $\sX$ is finite set, the action space $\M(\sX\times\sA)$ for Player 1 in GNEP is an infinite dimensional space.  In the literature, the algorithms developed for GNEPs are in general established for finite dimensional strategy spaces. Therefore, in this section, we suppose that the state space $\sX$  and the action space $\sA$ are finite sets. However, this creates a problem because our main theorem (Theorem~\ref{theorem3}) about approximate Nash equilibrium for linear MFGs is valid when $\sA$ is convex. Hence, it cannot be applied directly to the finite action spaces. 

To circumvent this problem, two approaches are possible. In the first one, it is possible to prove asymptotic version of Theorem~\ref{theorem3}; that is, there is no explicit relation between $\Theta$ functions and $N$ but it can be proved that $\Theta$ functions converge to $0$ as $N \rightarrow \infty$. This can be established via the method that is used to prove \cite[Theorem 4.1]{SaBaRa18}. To have this asymptotic result, it is enough to have finite $\sX$ and $\sA$. No need to put extra conditions on $\bc$ and $\bp$. Hence, in this case, if MFE policy is applied by all the agents in the finite agent game, then it constitutes approximate Nash equilibrium if the number of agents is sufficiently high. But we can not quantify how high it should be. 

In the second approach, we can equivalently reformulate the problem by pretending $\P(\sA)$ as our action space. In this case, we let $\sU := \P(\sA)$ denote our new action space, which is a convex and compact subset of $|\sA|$-dimensional Euclidean space. Then we redefine our transition probability and one stage cost function as follows
\begin{align*}
\bp_{\new}(\,\cdot\,|x,u,z) &:= \sum_{a \in \sA} \bp(\,\cdot\,|x,a,z) \, u(a) \\
\bc_{\new}(x,u,z) &:= \sum_{a \in \sA} \bc(x,a,z) \, u(a)
\end{align*}
In this case, it is possible to apply Theorem~\ref{theorem3} if assumption (a)--(g) are true. However, although $\bp_{\new}$ and $\bc_{\new}$ are linear in $u$ (and so convex), they are not strongly convex in $u$. Hence assumption (d) is not true for this new formulation. To handle this problem, a common approach is to add a strongly convex regularization term to the cost function $\bc_{\new}$ (see \cite{AnKaSa22}). In regularized version, the cost function is given by 
$$
\bc^{\reg}(x,u,z) \coloneqq \bc_{\new}(x,u,z) + \lambda \,h(u) 
$$
where $h$ is a $\theta$-strongly convex function and $\lambda >0$ is some constant. A typical example for $h$ is the negative entropy $h(u) = \sum_{a \in \sA} \ln(u(a)) \, u(a)$. By choosing $\theta$ and $\lambda$ properly, it is possible to satisfy the assumptions (a)--(g), in particular assumption (d). If $\theta$ and $\lambda$ are small enough, then one can prove that equilibrium solutions of un-regularized problems (both for finite population case and infinite population case) are approximately equilibrium solutions for regularized problems. Since Theorem~\ref{theorem3} can be applied to regularized problem, in view of the last observation, it can be applied to the un-regularized original setup as well with an additional error as a result of regularization term $\lambda \,h(u)$. 

The previous discussions provide sufficient justification for why the action space in the model can be considered finite. Indeed, extending the below algorithm being introduced to the infinite dimensional action spaces is a future research direction. In general, solving GNEP problems with finite dimensional action spaces is already very challenging problem. There are a limited number of algorithms available in the literature that can be used for the most general formulation of GNEPs (see \cite{AxFaKaSa11}). It may be difficult to generalize the existing algorithms or develop new ones for infinite dimensional GNEPs. 

Let us now give a more explicit formulation of inequality constrained GNEP that is introduced in the previous section, when state and action spaces are finite.
\begin{multicols}{2}
\centering{Player 1}
\begin{align*}
\text{Given $\mu$:} \,\,
&\text{minimize}_{\zeta\in \cR^{|\sX\times\sA|}} \text{ } \langle \zeta,c_{\mu} \rangle
\nonumber \\*
&\text{subject to  } \hat{\zeta} \geq (1-\beta)\mu + \beta \, \zeta \, p_{\mu}  \\*
&\Id \cdot \, \zeta \geq 0 
\end{align*}\\ 
\centering{Player 2}
\begin{align*}
\text{Given $\zeta$:} \,\,
&\text{minimize}_{\mu \in \cR^{|\sX|}} \text{ } g(\zeta,\mu)
\nonumber \\*
&\text{subject to  } \mu \geq \zeta \, p_{\mu}, \,\, \langle \mu, {\bf 1} \rangle \geq 1 \\*
&\Id \cdot \, \mu \geq 0 
\end{align*}
\end{multicols}

The remainder of this section will explain an algorithm that was developed to solve GNEPs using an interior-point method in a previous work \cite{AxFaKaSa11}, and how it has been adapted for the current study. To this end, since we are free to choose $g$, we suppose that the auxiliary cost function $g$ for player 2 is twice continuously differentiable and for any $\zeta$, $g(\zeta,\,\cdot\,)$ is convex in $\mu$. With these conditions, our problem satisfies assumptions A1 and A2 in \cite{AxFaKaSa11}. For instance, if one chooses $g$ as a linear function of $\mu$, then both players face with linear programs. In particular, if $g(\zeta,\mu) := -\langle \zeta, c_{\mu} \rangle$, then our problem becomes a zero sum game.  

Let us define the functions $h_1: \cR^{|\sX\times\sA|} \times \cR^{|\sX|} \rightarrow \cR^{|\sX\times\sA|}\times\cR^{|\sX|}$ and $h_2: \cR^{|\sX\times\sA|} \times \cR^{|\sX|} \rightarrow \cR^{|\sX|}\times\cR\times\cR^{|\sX|}$ as follows
\begin{align*}
h_1(\zeta,\mu) := 
\begin{pmatrix}
-\Id \cdot \, \zeta \\
-\hat{\zeta} + (1-\beta)\mu + \beta \, \zeta \, p_{\mu}
\end{pmatrix}, \,\,\,\,\,\,
h_2(\zeta,\mu) := 
\begin{pmatrix}
-\Id \cdot \, \mu \\
-\langle \mu, {\bf 1} \rangle + 1 \\
-\mu + \zeta \, p_{\mu}
\end{pmatrix}
\end{align*}
Then we can write above GNEP in the following form
\begin{multicols}{2}
\centering{Player 1}
\begin{align*}
\text{Given $\mu$:} \,\,
&\text{minimize}_{\zeta\in \cR^{|\sX\times\sA|}} \text{ } \langle \zeta,c_{\mu} \rangle
\nonumber \\*
&\text{subject to  } h_1(\zeta,\mu) \leq 0  
\end{align*}\\
\centering{Player 2}
\begin{align*}
\text{Given $\zeta$:} \,\,
&\text{minimize}_{\mu \in \cR^{|\sX|}} \text{ } g(\zeta,\mu)
\nonumber \\*
&\text{subject to  } h_2(\zeta,\mu) \leq 0
\end{align*}
\end{multicols}
Now, let us derive the joint KKT conditions for player 1 and player 2, whose solution gives a Nash equilibrium for GNEP. To this end, we need to define several functions. First we define 
\begin{align*}
L_1(\zeta,\mu,\lambda) &:= \langle \zeta,c_{\mu} \rangle + \langle h_1(\zeta,\mu),\lambda \rangle \\
L_2(\zeta,\mu,\gamma) &:= g(\zeta,\mu) + \langle h_2(\zeta,\mu),\gamma \rangle
\end{align*}
where $\lambda$ and $\gamma$ are Lagrange multipliers of player 1 and player 2, respectively. Let $\lambda = (\lambda_1,\lambda_2)$, where $\lambda_1 \in \cR^{|\sX\times\sA|}$ and $\lambda_2 \in \cR^{|\sX|}$, and let $\gamma = (\gamma_1,\gamma_2,\gamma_3)$, where $\gamma_1 \in \cR^{|\sX|}$, $\gamma_2 \in \cR$, and $\gamma_3 \in \cR^{|\sX|}$. Note that for any $(x,a) \in \sX\times\sA$, we have 
$$
\partial_{\zeta(x,a)} L_1(\zeta,\mu,\lambda) = c_{\mu}(x,a) - \lambda_1(x,a) - \lambda_2(x) + \beta \, \sum_{y \in \sX} \lambda_2(y) \, p_{\mu}(y|x,a) 
$$
and similarly, for any $x \in \sX$, we have 
$$
\partial_{\mu(x)} L_2(\zeta,\mu,\lambda) = \partial_{\mu(x)} g(\zeta,\mu) - \gamma_1(x) - \gamma_2 - \gamma_3(x) + \sum_{y \in \sX} \gamma_3(y) \, p_{\zeta}(y|x)
$$
where $p_{\zeta}(y|x) := \sum_{(z,a) \in \sX\times\sA} \bp(y|z,a,x) \, \zeta(z,a)$. Let ${\bf F}(\zeta,\mu,\lambda,\gamma) := \left(\nabla_{\zeta} L_1(\zeta,\mu,\lambda), \nabla_{\mu} L_2(\zeta,\mu,\gamma)\right)$ and ${\bf h}(\zeta,\mu) := \left(h_1(\zeta,\mu), h_2(\zeta,\mu)\right)$. Then the joint KKT conditions for player 1 and player 2 can be written as 
$$
{\bf F}(\zeta,\mu,\lambda,\gamma) = 0, \,\, \lambda,\gamma \geq 0, \,\, {\bf h}(\zeta,\mu) \leq 0, \,\, \langle {\bf h}(\zeta,\mu), (\lambda,\gamma) \rangle = 0
$$
More explicitly, we can write joint KKT conditions as follows
\begin{align}
&c_{\mu}(x,a) - \lambda_1(x,a) - \lambda_2(x) + \beta \, \sum_{y \in \sX} \lambda_2(y) \, p_{\mu}(y|x,a) = 0, \,\,\, \forall (x,a) \in \sX\times\sA \label{kkt1} \\
&\partial_{\mu(x)} g(\zeta,\mu) - \gamma_1(x) - \gamma_2 - \gamma_3(x) + \sum_{y \in \sX} \gamma_3(y) \, p_{\zeta}(y|x) =0, \,\,\, \forall x \in \sX \label{kkt2}\\
&\lambda, \gamma, \zeta, \mu \geq 0 \label{kkt3}\\
&\hat{\zeta} \geq (1-\beta)\mu + \beta \, \zeta \, p_{\mu}, \,\,\,  \mu \geq \zeta \, p_{\mu}, \,\,\, \mu(\sX) \geq 1 \label{kkt4} \\
&-\sum_{(x,a) \in \sX\times\sA} \zeta(x,a) \, \lambda_1(x,a) + \sum_{y \in \sX} \left(-\hat{\zeta}(y) + (1-\beta)\mu(y) + \beta \, \zeta \, p_{\mu}(y) \right) \lambda_2(y) = 0 \label{kkt5} \\
&-\sum_{x \in \sX} \mu(x) \, \gamma_1(x) + (1-\mu(\sX)) \, \gamma_2 + \sum_{y \in \sX} \left(-\mu(y) + \zeta \, p_{\mu}(y) \right) \, \gamma_3(y) = 0 \label{kkt5}
\end{align}
Here, (\ref{kkt1}) is indeed the Bellman optimality equation, where $\lambda_2$ is the optimal value function. Now, we transform joint KKT conditions into a root finding problem. To this end, we introduce slack variables $({\bar \lambda},{\bar \gamma})$, where ${\bar \lambda} \in \cR^{|\sX\times\sA|}\times\cR^{|\sX|}$ and ${\bar \gamma} \in \cR^{|\sX|}\times\cR\times\cR^{|\sX|}$, and define 
\begin{align*}
H(z) &:= H(\zeta,\mu,\lambda,\gamma,{\bar \lambda},{\bar \gamma}) :=
\begin{pmatrix}
{\bf F}(\zeta,\mu,\lambda,\gamma)  \\
{\bf h}(\zeta,\mu) + ({\bar \lambda},{\bar \gamma}) \\
 (\lambda,\gamma) \circ ({\bar \lambda},{\bar \gamma}) 
\end{pmatrix} \\
\intertext{and}
Z &:= \left\{ z=(\zeta,\mu,\lambda,\gamma,{\bar \lambda},{\bar \gamma}): (\lambda,\gamma), ({\bar \lambda},{\bar \gamma}) \geq 0 \right\}
\end{align*}
where $(\lambda,\gamma) \circ ({\bar \lambda},{\bar \gamma})$ is the vector formed by diagonal elements of the outer product of the vectors $(\lambda,\gamma)$ and  $({\bar \lambda},{\bar \gamma})$. 
Then it is straightforward to show that $(\zeta,\mu,\lambda,\gamma)$ satisfy joint KKT conditions if and only if $(\zeta,\mu,\lambda,\gamma)$ and some suitable $({\bar \lambda},{\bar \gamma})$ satisfy the constrained root finding problem $H(z) = 0, \,\, z \in \sZ$. In order to find a solution to constrained root finding problem, an interior-point algorithm is developed in \cite{AxFaKaSa11}. In remainder of this section, we explain this algorithm, which depends on potential reduction method from \cite{MoPa99}. Let $n = |\sX\times\sA| + |\sX|$ (number of total variables in GNEP) and $m := |\sX\times\sA| + 3 \, |\sX| + 1$ (number of total constraints in GNEP). Hence $H: \cR^{n} \times \cR^{2m} \rightarrow \cR^{n} \times \cR^{2m}$ and $Z = \cR^n \times \cR_+^{2m}$. We first define a potential function on the interior of $Z$ as follows
$$
p(u,v) = K \, \log\left(\|u\|^2 + \|v\|^2 \right) - \sum_{i=1}^{2m} \log(v_i) 
$$ 
where $K > m$. This function penalizes points that are close to the boundary of $Z$ that are far from the origin. Now, we define the potential function for the constrained root finding problem by composing $p$ and $H$ 
$$
\psi(z) := p(H(z)) 
$$
where $z \in \intr Z \cap H^{-1}(\intr Z) =: Z_I$. Let $\nabla H$ denote the Jacobian of the function $H$. Now it is time to give the algorithm. 

\begin{algorithm}[H]
\caption{}
\label{av-H2}
\begin{algorithmic}
\STATE{Inputs: $\kappa \in (0,1)$ and $a := \begin{pmatrix}
{\bf 0}_n \\
{\bf 1}_{2m} 
\end{pmatrix} \bigg/ \left\|\begin{pmatrix}
{\bf 0}_n \\
{\bf 1}_{2m} 
\end{pmatrix}\right\|$}
\STATE{Start with $z_0$}
\FOR{$k=0,1,2\ldots$}
\STATE{
\begin{itemize}
\item[(a)] Choose $\sigma_k \in [0,1), \eta_k \geq 0$, and compute a vector $d_k \in \cR^n \times \cR^{2m}$ such that 
\begin{align}
\left\| H(z_k) + \nabla H(z_k) \cdot d_k - \sigma_k \langle a,  H(z_k) \rangle \, a  \right\| &\leq \eta_k \, \|H(z_k)\| \label{alg1} \\ 
\intertext{and}
\langle \nabla \psi(z_k) , d_k \rangle &< 0 \label{alg2}
\end{align}
\item[(b)] Compute a stepsize $t_k:= \max \{\kappa^l: l = 0,1,2,\ldots \}$ such that 
\begin{align}
z_k + t_k \, d_k &\in Z_I  \label{alg3} \\
\intertext{and}
\psi(z_k + t_k \, d_k) &\leq \psi(z_k) +  t_k \, \langle \nabla \psi(z_k) , d_k \rangle \label{alg4}
\end{align}
\item[(c)] Set $z_{k+1} := z_k + t_k \, d_k$
\end{itemize}
}
\ENDFOR
\end{algorithmic}
\end{algorithm}

In order to establish the convergence of the algorithm we impose the following condition. 

\begin{itemize}
\item[(h)] For any $z \in Z_I$, the Jacobian $\nabla H(z)$ is invertible. 
\end{itemize}
 
First, note that under assumption (h), the following equation has a solution $d$ for any $z \in \intr Z$ and $\sigma \in [0,1)$
$$
H(z) + \nabla H(z) \, d = \sigma \langle a, H(z) \rangle \, a 
$$
Hence, one can use this solution in (\ref{alg1}) because it is known that for this solution $d_k$, we have $\langle \nabla \psi(z_k) , d_k \rangle < 0$ \cite[Lemma 11.3.3]{FaPa03}. Indeed, in this case, $d_k$ becomes
\begin{align}
d_k = \left(\nabla H(z_k)\right)^{-1} \left(\sigma_k \langle a, H(z_k) \rangle \, a  - H(z_k)\right) \label{exact}
\end{align}
Hence, our update becomes 
$$
z_{k+1} := z_k + t_k \, \left(\nabla H(z_k)\right)^{-1} \left(\sigma_k \langle a, H(z_k) \rangle \, a - H(z_k) \right)
$$
Moreover, since $\langle \nabla \psi(z_k) , d_k \rangle < 0$, one can always find $t_k$ that satisfies (\ref{alg3}) and (\ref{alg4}). The following convergence result follows from \cite[Theorems 4.3 and 4.10]{AxFaKaSa11}.

\begin{theorem}\label{com-theorem}
Suppose that assumption (h) holds. Moreover, pick $\sigma_k$ and $\eta_k$ so that 
$$
\limsup_{k\rightarrow\infty} \sigma_k < 1, \,\,\,\, \lim_{k\rightarrow\infty} \eta_k=0
$$
Then, the sequence $\{z_k\} := \{(\zeta_k,\mu_k,\lambda_k,\gamma_k,{\bar \lambda}_k,{\bar \gamma}_k)\}$ is bounded and any accumulation point $z^* = (\zeta^*,\mu^*,\lambda^*,\gamma^*,{\bar \lambda}^*,{\bar \gamma}^*)$ of this sequence is a solution to the constrained root finding problem $H(z^*) = 0$; that is, $H(z_k) \rightarrow 0$ as $k\rightarrow\infty$. Hence, $(\mu^*,\pi^*)$ is MFE, where $\zeta^*(dx,da) = \pi^*(da|x) \, \hat{\zeta}^*(dx)$. 
\end{theorem}

\begin{proof}
Here, we need to check the conditions (b) and (c) in  \cite[Theorem 4.10]{AxFaKaSa11}. Obviously 
$$
\lim_{\|(\zeta,\mu)\| \rightarrow \infty} \|\max\{0,{\bf h}(\zeta,\mu)\}\| = \infty
$$
Hence condition (b) is true. Since $h_1(\zeta,\mu)$ is linear in $\zeta$ given $\mu$ and $h_2(\zeta,\mu)$ is linear in $\mu$ given $\zeta$, one can also establish condition (c), which is an extended Mangasarian-Fromovitz constraint qualification condition. 
\end{proof}

\subsection{A numerical example}

We consider the malware spread model studied in \cite{SuAd19}. In this model, we suppose that there are large number of agents, where each agent has a local state $x_i(t) \in \{0,1\}$. Here $x_i(t) = 0$ represents the "healthy" state and $x_i(t) = 1$ represents the "infected" state. Each agent can take action $a_i(t) \in \{0,1\}$, where $a_i(t) = 0$ represents "do nothing" and $a_i(t) = 1$ represents "repair". The dynamics are given by 
\begin{align*}
x_i(t+1) =
\begin{cases}
x_i(t) + (1-x_i(t)) \, w_i(t), & \text{if $a_i(t)=0$} \\
0, & \text{if $a_i(t)=1$}
\end{cases}
\end{align*}
where $w_i(t) \in \{0,1\}$ is a Bernoulli random variable with success probability $q$, which gives the probability of an agent getting infected. In this setting, if an agent chooses to not take any action, they may be infected with probability $q$, but if they choose to take a repair action, they return to the healthy state. Each agent pays a cost 
$$
\bc(x_i(t),a_i(t),z_i(t)) = (k+z_i(t)) \, x_i(t) + \theta \, a_i(t)
$$
where $z_i(t) \sim \e[\,\cdot\,|\,\bx(t)]$, $\theta$ is the cost of repair, and $(k+z_i(t))$ represents the risk of being infected. In the infinite population limit, the stationary version of the problem is studied and the model is formulated as a generalized Nash equilibrium problem (GNEP). In this GNEP, the cost function for player 2 is taken to be the same as that of player 1. However, note that player 1 controls the distribution of $(x(t),a(t))$ and player 2 controls the distribution of $z(t)$. For numerical experiments, we use the following system parameters $k=0.2$, $\theta =0.5$, $\beta = 0.9$, $q=0.9$. We use MATLAB to do the numerical experiments. The algorithm runs for $10000$ iterations and uses the following parameters $\sigma_k = 0.4$, $\eta_k =0$, $\kappa=0.001$. Here, we take $\eta_k=0$ because we use 
\begin{align*}
d_k = \left(\nabla H(z_k)\right)^{-1} \left(\sigma_k \langle a,H(z_k) \rangle \, a - H(z_k)\right) 
\end{align*}
to update $z_k$. To perform step (\ref{alg4}) in the algorithm, we use Armijo line search. 

Note that in this example, $H$ function has $28$ outputs. The first $6$ of them represent ${\bf F}(\zeta,\mu,\lambda,\gamma)$, whose evolution is shown in Figure~\ref{H_F}.

\begin{figure}[H]
\centering
	\includegraphics[scale=0.25]{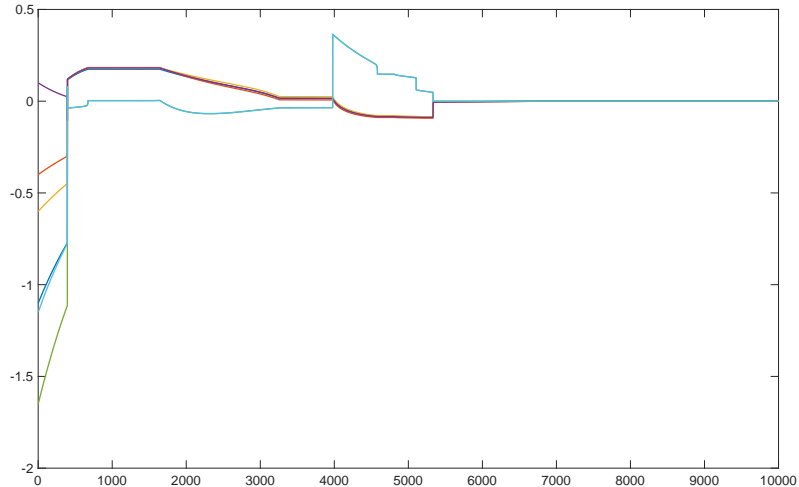}
	\caption{The evolution of ${\bf F}(\zeta,\mu,\lambda,\gamma)$} \label{H_F}
\end{figure}

The next $6$ of them represent $h_1(\zeta,\mu)+\lambda$, whose evolution is shown in Figure~\ref{H_lambda}.

\begin{figure}[H]
\centering
	\includegraphics[scale=0.25]{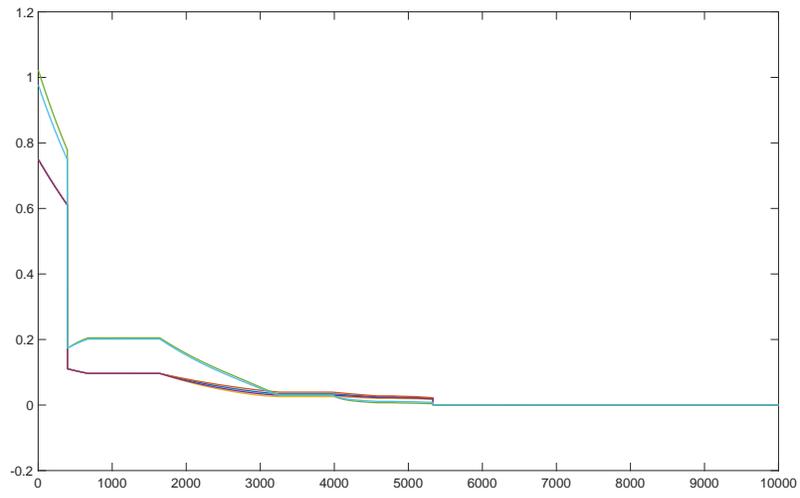}
	\caption{The evolution of $h_1(\zeta,\mu)+\lambda$} \label{H_lambda}
\end{figure}

The next $5$ of them represent $h_2(\zeta,\mu)+\gamma$,  whose evolution is shown in Figure~\ref{H_gamma}.

\begin{figure}[H]
\centering
	\includegraphics[scale=0.25]{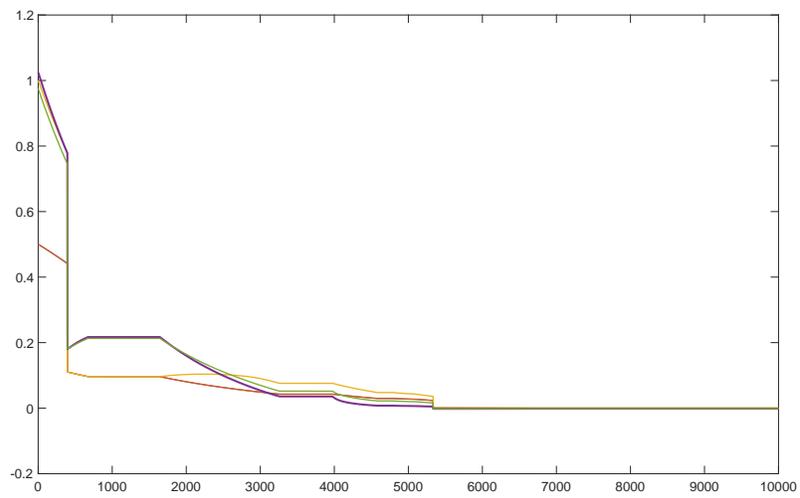}
	\caption{The evolution of $h_2(\zeta,\mu)+\gamma$} \label{H_gamma}
\end{figure}

The next $6$ of them represent $\lambda \circ {\bar \lambda}$, whose evolution is shown in Figure~\ref{H_blambda}.

\begin{figure}[H]
\centering
	\includegraphics[scale=0.25]{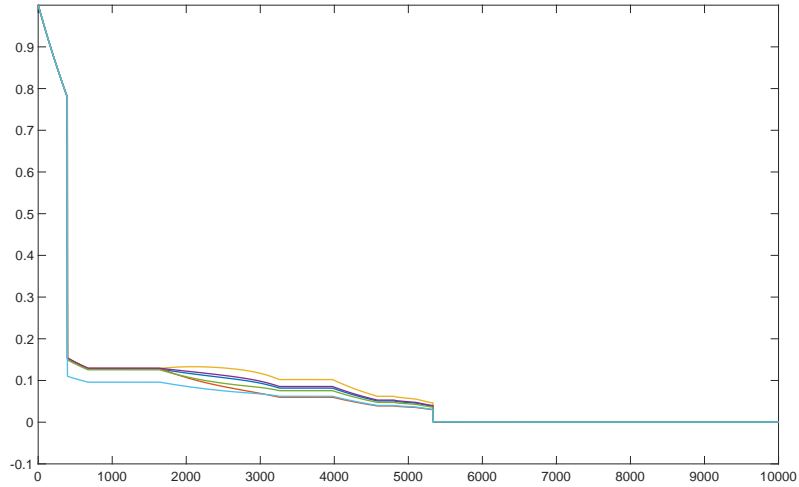}
	\caption{The evolution of $\lambda \circ {\bar \lambda}$} \label{H_blambda}
\end{figure}

The final $5$ of them represent $\gamma \circ {\bar \gamma}$, whose evolution is shown in Figure~\ref{H_bgamma}.

\begin{figure}[H]
\centering
	\includegraphics[scale=0.25]{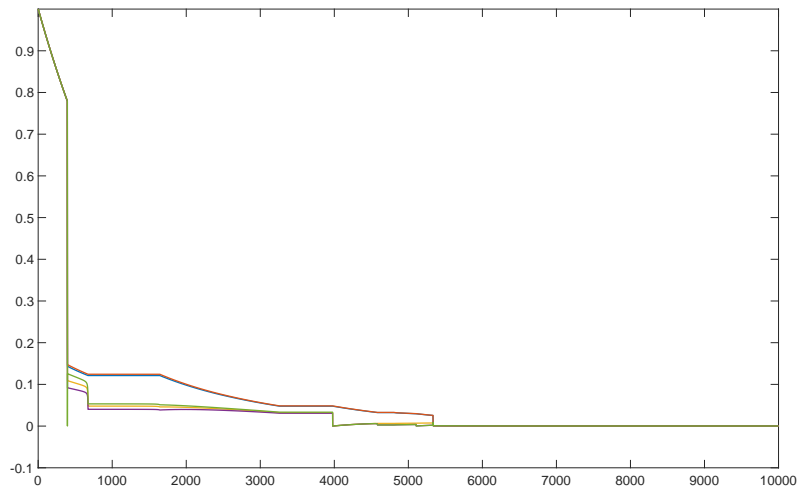}
	\caption{The evolution of $\gamma \circ {\bar \gamma}$} \label{H_bgamma}
\end{figure}

As one can see, the outputs of $H$ converge to zero as expected.

Now let us look at the behavior of the mean-field term. It can be seen in Figure~\ref{H_mu} that mean-field term converges to the following distribution $[0.59, 0.41]$. Hence, at the equilibrium, $59\%$ of the states are healthy. 

\begin{figure}[H]
\centering
	\includegraphics[scale=0.25]{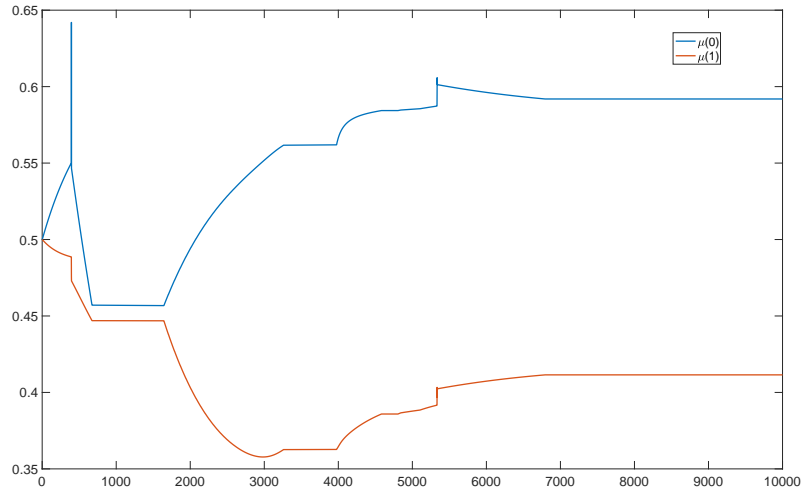}
	\caption{The evolution of mean-field term $\mu$} \label{H_mu}
\end{figure}

If we also analyze the behavior of the equilibrium policy, it can be seen in Figure~\ref{pi_0} and Figure~\ref{pi_1} that equilibrium policy converges to the following conditional distribution $\pi(\,\cdot\,|0) = [0.76,0.24]$ and  $\pi(\,\cdot\,|1) = [0.02,0.98]$. Hence, once an agent is infected, then with probability $0.98$, it should apply repair action. However, if the agent is healthy, then it should do nothing with probability $0.76$. This is probably because of the fact that the cost of repair is more expensive than the risk of infection. 

\begin{figure}[H]
\centering
	\includegraphics[scale=0.25]{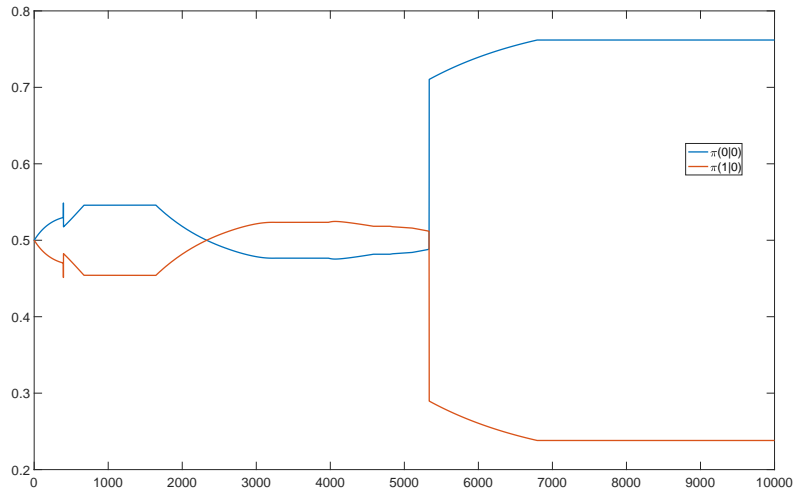}
	\caption{The evolution of equilibrium policy $\pi(\,\cdot\,|0)$} \label{pi_0}
\end{figure}

\begin{figure}[H]
\centering
	\includegraphics[scale=0.25]{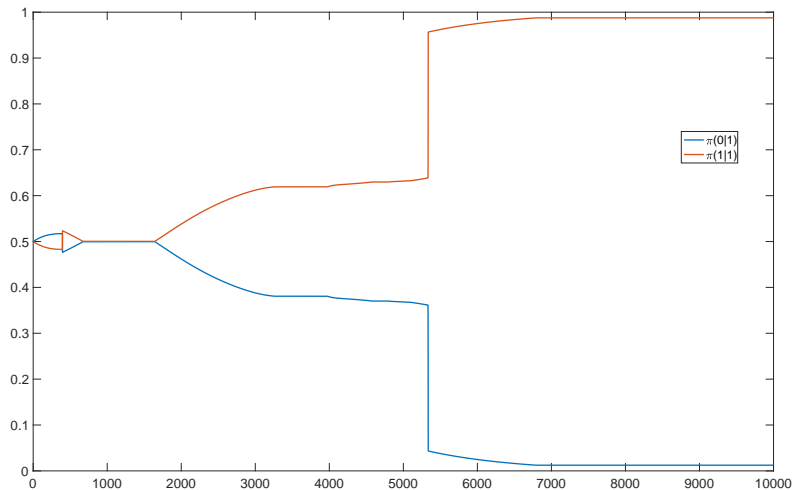}
	\caption{The evolution of equilibrium policy $\pi(\,\cdot\,|1)$} \label{pi_1}
\end{figure}

\section{Conclusion}\label{final_section}

In this paper we have introduced linear mean-field games, in which the interaction between agents is determined by the empirical distribution of their states. Using the mean-field approach, we have demonstrated the existence of approximate Nash equilibria for finite-population games when the number of agents is sufficiently large. Under mild technical conditions, it can be shown that the limiting mean-field problem has an equilibrium. We have then applied the policy obtained from this equilibrium to the finite population game and proved that it constitutes an $\varepsilon(N)$-Nash equilibrium for games with $N$-agents, where an explicit relation between $\varepsilon(N)$ and $N$ has been established. Then, we have used linear programming and the linearity of transition probabilities in the mean-field term to convert the game into a generalized Nash equilibrium problem in the limit of an infinite number of agents. We have also developed an algorithm for finding a mean-field equilibrium with a guarantee of convergence.

\begin{acks}[Acknowledgments]
The author would like to thank Professor Tamer Basar and Professor Serdar Y\"{u}ksel for their constructive comments that improved the quality of this paper.
\end{acks}

\bibliographystyle{imsart-number}

\end{document}